\documentclass[letterpaper,11pt]{article}
\usepackage{times}
\usepackage[utf8]{inputenc}
\usepackage{amsmath, amsthm, amssymb, thm-restate}
\usepackage[table,xcdraw]{xcolor}
\usepackage{algorithm}
\usepackage{algpseudocode}
\usepackage{wrapfig}
\usepackage{subcaption}
\usepackage{setspace}
\usepackage{xspace}
\usepackage{mathtools}
\usepackage{enumerate}
\usepackage{natbib}
\usepackage{wrapfig}
\usepackage{comment} 
\usepackage{tcolorbox} 
\usepackage{xfrac}
\usepackage{hyperref}
\usepackage{multirow}
\usepackage{caption}
\usepackage{bm}
\usepackage{newfloat}
\usepackage{enumitem}
\usepackage{dblfloatfix} 
\usepackage{wrapfig}
\usepackage{dsfont}
\usepackage{mdframed}
\usepackage{tikz}
\usepackage{multirow}
\usepackage{float}
\usepackage{enumerate}
\usepackage[dvipsnames]{xcolor} 

\usetikzlibrary{shapes,positioning,fit,backgrounds,decorations.pathreplacing,calc}

\usepackage{fancyhdr}
\usepackage[noabbrev]{cleveref}
\AddToHook{cmd/appendix/before}{\crefalias{section}{appendix}}

\usepackage[margin=1in]{geometry}

\allowdisplaybreaks

\definecolor{mygreen}{RGB}{10,110,230}
\definecolor{myred}{RGB}{10,110,230}

\hypersetup{
     colorlinks=true,
     citecolor= mygreen,
     linkcolor= myred
}

\newcommand{\E}[0]{\ensuremath{\mathbb{E}}}

\newcommand{\opt}{\ensuremath{\textsc{opt}}\xspace}

\newcommand{\bigo}{\mathcal{O}} 
 
\newcommand{\eps}{\varepsilon}

\newcommand{\polylog}{\mbox{polylog}}

\newcommand{\instance}{\mathcal{I}} 
\newcommand{\sml}[1]{s(#1)} 
 
\newcommand{\avg}[1]{a(#1)} 
 
\newcommand{\cost}{\mbox{Cost}} 
\renewcommand{\opt}{\mbox{OPT}} 
\newcommand{\Tau}{\mathrm{T}}
\renewcommand{\E}{\mathbb{E}}

\newcommand{\mpl}[1]{\mu(#1)} 

\usepackage{amsthm, amsmath, amssymb}
\usepackage{thmtools}

\newcommand{\mc}{\mathcal{C}}

\newcommand{\CC}{{\sc CC}}

\newcommand{\kCC}{{\sc \mbox{$k$}-CC}}
\newcommand{\Pivot}{\textsc{Pivot}}

\newcommand{\BP}{\textsc{BalancedPivot}}

\newcommand{\UP}{\textsc{UniformPivot}}

\theoremstyle{plain}
\newtheorem{theorem}{Theorem}
\newtheorem{lemma}{Lemma}[section]

\theoremstyle{definition}
\newtheorem{definition}[lemma]{Definition}

\newtheorem{fact}[lemma]{Fact}
\newtheorem{claim}[lemma]{Claim}

\makeatletter
\renewcommand{\paragraph}{
  \@startsection{paragraph}{4}
  {\z@}{10pt}{-1em}
  {\normalfont\normalsize\bfseries}
}
\makeatother

\makeatletter
\patchcmd{\@algocf@start} 
  {-1.5em} 
  {0pt} 
  {}{} 
\makeatother

\newcommand*\samethanks[1][\value{footnote}]{\footnotemark[#1]}
\newcommand{\email}[2]{\protect\href{mailto:#2}{#1}}

\title{Competitive Random-Order Correlation $k$-Clustering}

\author{\email{Mahsa Derakhshan}{m.derakhshan@northeastern.edu}\thanks{Northeastern University. $\{$m.derakhshan, ghasemi.e, r.rajaraman, wasim.o$\}$@northeastern.edu} \and \email{Andisheh Ghasemi}{ghasemi.e@northeastern.edu}\samethanks \and \email{Rajmohan Rajaraman}{r.rajaraman@northeastern.edu}\samethanks \and \email{Omer Wasim}{wasim.o@northeastern.edu}\samethanks \and \email{Tegan Wilson}{te.wilson@northeastern.edu}\thanks{Northeastern University and University of Wisconsin Madison. te.wilson@northeastern.edu}
}

\date{}

\begin{document}
\maketitle

\thispagestyle{empty}

\begin{abstract}
    
Correlation clustering has been extensively studied over the last two decades in many different computational models owing to its wide-ranging practical applications. 
The goal is to compute a partition of the vertex set such that the total number of disagreements, i.e. the sum of edges between clusters, and non-edges within clusters, is minimized. 
In this paper, we study \textbf{correlation $k$-clustering}, in which the total number of clusters is restricted to $k$.  Correlation $k$-clustering is NP-hard, and while previous work has presented a polynomial time approximation scheme for constant $k$, there are no known results for general $k$. 

Our first result is a polynomial-time constant-factor approximation algorithm for correlation $k$-clustering for general $k$. 
The main focus of this work is in the more challenging \textbf{online} setting.  Noting that the best competitive ratio under adversarial arrivals is known to be $\Omega(n)$, we concentrate on the well-studied random-order model, where vertices arrive in random order and on arrival of a vertex, edges to its earlier-arrived neighbors are revealed. We prove a surprising lower bound of $\Omega(\log k)$-competitiveness for any online algorithm, which can be extended to $\Omega(\log n)$ when $k = \text{poly}(n)$. 
Finally, the main result of this paper is a polylogarithmic upper bound on the competitive ratio for correlation $k$-clustering, using an 
algorithm inspired by the classic \Pivot\ algorithm for correlation clustering. 
All theorems have been verified using Lean.

\end{abstract}
\newpage
\setcounter{page}{1} 

\section{Introduction}\label{sec:intro}

Clustering is a central technique in unsupervised learning and data analysis \cite{schutze2008introduction}. The overarching goal is to group items into clusters, such that items within the same cluster are more similar than items in different clusters (under some similarity measure). 
Traditional formulations such as $k$-means (resp. $k$-medians) assume items to be points in a metric space, and seek to minimize an objective function which is the sum across the $k$ clusters of squared $\ell_2$ distances (resp. $\ell_1$ distance) of each point in the cluster to its centroid (resp. center). 
However, the ``quality" of such a clustering is dependent on how well the underlying metric reflects the desired (dis)similarity between items and is often application dependent. 

On the other hand, the \emph{correlation clustering} problem, first introduced by Bansal, Blum, and Chawla \cite{bansal2004correlation}, offers a more general and flexible alternative to metric similarity. 
For every pair of items, one only needs a qualitative similarity measure (instead of a numerical measure) captured by a label in \{$+, -\}$, where $+$ corresponds to similarity and $-$ corresponds to dissimilarity. Correlation clustering is practically relevant in settings where the notion of similarity can be captured more easily and reliably in a qualitative rather than a quantitative sense \cite{DBLP:conf/stoc/AilonCN05, DBLP:journals/jcss/CharikarGW05,MathieuSS10}. For example, in social network analysis \cite{DBLP:journals/jmlr/ChiangHNDT14, DBLP:journals/tkdd/BonchiGGTU15, DBLP:conf/www/VeldtGW18, DBLP:journals/corr/abs-2108-01731}, a grouping of individuals based on observed affinities or preferences may be desired, while in document or image categorization \cite{DBLP:conf/icml/FinleyJ05, DBLP:journals/corr/abs-1207-4157, DBLP:journals/pami/KimYNK14}, one might only have binary labels capturing similarity.

Given a set of $n$ items, the correlation-clustering problem (which we refer to as \CC) models the items as vertices on a complete graph. 
The edges $(u,v)$ of this graph are labeled by either $+$ or $-$, corresponding to pairwise similarity between items. 
Given this graph, the most well-studied formulation pertains to "disagreement minimization" \cite{bansal2004correlation, DBLP:conf/stoc/AilonCN05, DBLP:journals/jcss/CharikarGW05, DBLP:conf/focs/Cohen-AddadLN22, DBLP:conf/focs/Cohen-AddadL0N23,CohenAddadLPSTY24} in which we seek to obtain a partition of vertices (i.e., clusters) such that the sum of inter-cluster edges labeled $+$, and the sum of intra-cluster edges labeled $-$ is minimized. Equivalently, we may instead model the problem with an unlabeled graph, where edges correspond to $+$ edges, and non-edges correspond to $-$ edges. This will be our assumption moving forward.

A key feature of the standard correlation clustering model is that it does not require specifying the number of clusters in advance—the optimal number is inferred from the structure of the input. 
However, in many practical scenarios, external constraints such as resource limitations \cite{DBLP:conf/spaa/AndreevR04, DBLP:conf/soda/KrauthgamerNS09} or fairness requirements \cite{DBLP:conf/nips/Chierichetti0LV17, DBLP:conf/aistats/AhmadianE0M20} may impose an upper bound on the number of clusters. 
Furthermore, clusters often correspond to maintained entities, such as customer segments, recommendation cohorts, or moderation categories, motivating practitioners to bound the number of clusters to satisfy operational, interpretability, or resource constraints. 
This motivates the study of \textbf{\emph{correlation $k$-clustering}}, a natural variant in which the algorithm must output a clustering using at most $k$ clusters, with the goal of minimizing the number of disagreements.  We refer to this problem as \kCC. 

For constant $k$, \kCC\ was studied by
Giotis and Guruswami~\cite{GiotisG06}, who presented a polynomial time approximation scheme (PTAS) for the problem.  However, their algorithm runs in time exponential in $k$ and hence does not yield any polynomial time algorithms for general $k$.  For the ``agreement maximization" version of correlation clustering, there exists an SDP-based 0.7666-approximation algorithm by Swamy~\cite{DBLP:conf/soda/Swamy04} which uses at most 6 clusters. We note that the ``disagreement minimization" and ``agreement maximization" versions of correlation clustering are unrelated from the point of view of approximation. In particular, no polynomial-time approximation algorithms are known for the ``disagreement minimization" version of \kCC\ for non-constant $k$.

In this paper, we study the ``disagreement minimization'' version of \kCC\ for general $k$. 
We give the first constant-factor approximation algorithm for \kCC\ for general $k$. 
Our algorithm converts any $\alpha$-approximate solution to unconstrained \CC\ into a $(1+2\alpha+\bigo(\eps))$-approximate solution for \kCC.

After studying approximation algorithms for \kCC\ in the offline setting, we then focus on the more challenging \textit{online} setting. 
In many applications, data arrives incrementally, and decisions must be made in an online fashion. For example, in recommendation systems, users arrive over time and must be placed into behavioral cohorts for personalization; in content moderation, new posts are evaluated relative to existing ones to determine their category or flag status~\cite{MathieuSS10, DBLP:conf/nips/LattanziMVWZ21,
DBLP:conf/icml/Cohen-AddadLMP22}. In such settings, it may be infeasible to recompute a clustering from scratch after the arrival of new data, which may lead to reversals of past clustering decisions. 

\noindent\textbf{Online Correlation Clustering.} Let $G=(V,E)$ be a graph on $n=|V|$ vertices. In \emph{online} \CC, vertices in $V$ arrive in discrete time steps and every vertex must be irrevocably assigned to a cluster immediately after its arrival. On arrival of a vertex $v$, edges of the form $(u,v)$ where $u$ is a vertex which arrived before $v$ are revealed \cite{MathieuSS10, DBLP:conf/nips/LattanziMVWZ21, DBLP:conf/icml/Cohen-AddadLMP22}.  We employ the standard competitive-analysis framework \cite{sleator1985amortized} to quantify the performance of an online algorithm. An online algorithm $\mathcal{A}$ is said to be $\rho$-competitive if the cost incurred by $\mathcal{A}$ is upper bounded by $\rho\cdot OPT$, where $OPT$ denotes the cost of an optimal offline algorithm which knows the entire input instance in advance.

Under adversarial vertex arrivals, a simple instance yields a lower bound of $\Omega(n)$ competitiveness \cite{MathieuSS10}. 
This lower bound naturally extends to \kCC\ for general $k$. 
To circumvent this bottleneck, we consider the widely studied \textit{random-order} model for online computation \cite{ DBLP:books/cu/20/Gupta020}. 
In the random-order \textit{vertex-arrival} model, the graph instance $G=(V,E)$ may be chosen adversarially, but the order of vertex arrivals is chosen uniformly at random from all $n!$ possible arrival orders \cite{DBLP:conf/stoc/KarpVV90, DBLP:conf/focs/Meyerson01, MathieuSS10, DBLP:journals/corr/LibertySS14, DBLP:conf/icml/Cohen-AddadLMP22}. 

The classic \Pivot\ algorithm due to Ailon, Charikar, and Newman, which obtains a 3-approximation algorithm for correlation clustering \cite{DBLP:conf/stoc/AilonCN05}, directly yields 3-competitiveness for online \CC\ in the random-order vertex-arrival model. \Pivot\ works as follows. Every vertex $v$ which is the first among its neighbors to arrive is designated as a \textit{pivot} vertex, and a new cluster containing $v$ (which we call a pivot cluster of $v$) is initialized. Every other vertex (which is not a pivot) is assigned to the cluster to which its earliest arrived pivot neighbor is assigned. 

However, there are no known results for online \kCC, motivating this natural question:

\begin{tcolorbox}
\begin{center}
\emph{Does online correlation $k$-clustering admit competitive algorithms in the random-order model?}
\end{center}
\end{tcolorbox}

We show that not only does online \kCC\ admit a $\polylog (n)$-competitive algorithm in the random order model, but we also provide a lower bound showing that any online algorithm is no better than $\Omega(\log n)$-competitive.
Our algorithm is a simple variation of the classic \Pivot\ algorithm, modified to account for the constraint on the number of clusters. 
Due to its simplicity, our algorithm would be easily implementable in a variety of settings.
While our algorithm relies heavily on the building block of \Pivot\ in its construction, a major focus of our analysis is on how our modifications remain competitive when the number of clusters is limited. Indeed, our analysis requires novel arguments even for the case of disjoint cliques, a setting that in unconstrained correlation clustering \Pivot\ solves trivially with zero cost, yet is nontrivial for \kCC.

\subsection{Our Results}
\paragraph{Offline Approximation.} We begin by studying \textit{offline} approximation algorithms for \kCC.  The PTAS provided by Giotis and Guruswami for the case when $k$ is constant relies on the enumeration of all possible $k$-clusterings for a sampled set of size $\Omega(k^2 \log n)$.  Since the number of such clusterings is $n^{\Omega(k)}$, this approach fails to extend to the case where $k$ is not a fixed constant. 
Furthermore, since \CC\ is APX-hard \cite{DBLP:journals/jcss/CharikarGW05}, \kCC\ is also APX-hard, thus ruling out a PTAS unless P = NP. 
Our first result, presented in Section~\ref{sec:offline}, is that \kCC\ admits constant-factor approximation algorithms.  

\begin{restatable}{theorem}{offline}\label{thm:main-1}
Given a constant $\alpha$ and any polynomial-time $\alpha$-approximation algorithm $\mathcal{A}$ for \CC, and given any $k\in\{1,\hdots,n\}$, there exists a polynomial-time algorithm $\mathcal{A}_k$ for \kCC\ which returns a $(1+2\alpha + \bigo(\eps))$-approximate solution, for any given $\eps > 0$.
\end{restatable}
The proof of Theorem~\ref{thm:main-1} has two main ingredients. 
First, we derive a constant-factor approximation for the special case when the input graph is a collection of disjoint cliques.  We note that offline \kCC\ is NP-hard even for disjoint cliques (this is implicit in~\cite{alon1998approximation}; for completeness, we present a proof in \Cref{app:np-hardness}). 
Our second ingredient uses the algorithm for disjoint cliques to convert any constant-factor approximate solution for \CC\ to a constant-factor approximate solution for \kCC. 
Using the current best approximation algorithm for \CC, which obtains a 1.485-approximation \cite{DBLP:conf/stoc/CaoCL0NV24}, we obtain a
$(3.97 + \bigo(\eps))$-approximation for \kCC.

\paragraph{Lower Bound for Online Algorithms.} The remainder of the paper concerns \textbf{online} \kCC\ in the random order vertex-arrival model, the main focus of our work. 
As noted above, the classic \Pivot\ algorithm \cite{DBLP:conf/stoc/AilonCN05} yields a simple $3$-competitive algorithm for online \CC. By Theorem~\ref{thm:main-1}, the offline approximability of \kCC\ is the same, up to constant factors, as the approximability of \CC.  
It is natural to ask if a constant-factor competitive algorithm exists for online \kCC{} under random vertex arrivals. 
Our second result, presented in Section~\ref{sec:lower-bound}, is a surprising logarithmic lower bound for the best competitive ratio achievable for online correlation $k$-clustering.  Interestingly, our lower bound applies even to the case of disjoint cliques, for which \CC{} has a trivial optimal solution with zero cost, indicating a divergence between \kCC{} and \CC{} in the online setting.

\begin{restatable}{theorem}{lowerbound}\label{thm:onlinelowerbound}
For all sufficiently large $k$ there is an instance of disjoint cliques on $n$ vertices for an $n = \bigo(k\log k)$, on which every (possibly randomized) online algorithm for \kCC{} has expected cost at least $\Omega(\log k \cdot \opt)$ in the random-order model. That is, every online algorithm is $\Omega(\log k)$-competitive. In terms of $n$, this also yields an $\Omega(\log n)$-competitive lower bound.
\end{restatable}

We additionally show a stronger lower bound of $\Omega(\log^2 n)$-competitiveness for online algorithms which never split disjoint cliques. 
This is surprising. As we discuss in \Cref{sec:overview.offline} and \Cref{lem:cliques-intact}, any optimal offline solution to \kCC\ will never split disjoint cliques.
Notably, this lower bound applies not only to our online algorithms, but also any online algorithm that uses \Pivot\ as a building block in a similar way.

\begin{restatable}{theorem}{lblogsq}\label{thm:lb_logsq}
For all sufficiently large $k$ there is an instance of disjoint cliques on $n$ vertices for an $n = \bigo(k\log k)$, on which every (possibly randomized) online algorithm for \kCC{} which does not split the cliques has expected cost at least $\Omega(\log^2 k \cdot \opt)$ in the random-order model. That is, any such online algorithm is $\Omega(\log^2 k)$-competitive. In terms of $n$, this also yields an $\Omega(\log^2 n)$-competitive lower bound for online algorithms which never split disjoint cliques.
\end{restatable}

\paragraph{Online Competitive Algorithms.} 
At a high level, any solution for \kCC\ incurs two kinds of costs: intrinsic disagreement cost incurred by any \CC\ solution, and the cost of placing ``uncorrelated vertices'' in the same cluster due to the constraint of using at most $k$ clusters. 
Since the Pivot algorithm is competitive for \CC, it offers an approach to address costs of the first kind: when a vertex $v$ arrives, if a neighbor $p$ of $v$ has already been designated as a pivot, then place $v$ in the same cluster as $p$; otherwise designate $v$ as a pivot. 
Now, a critical question is \emph{where to place a newly designated pivot $v$}. 
This can be guided by costs of the second kind. 
It is instructive to study the disjoint cliques case, for which unconstrained \CC\ admits a zero-cost solution,
since all costs incurred by any $k$-clustering are of the second kind. 
For disjoint cliques, we show that the intuitive solution -- to distribute the cliques across clusters in a balanced manner -- is provably good. 
Analogously for general graphs, we assign ``pivot clusters'' computed by the Pivot algorithm to clusters in an approximately balanced manner.  

While the best balanced allocation is NP-hard to compute, even offline and for disjoint cliques, the notion of balancing inspires the following procedure: When a new pivot arrives, place it in a least loaded cluster (at that instant).  We call this (deterministic) algorithm \BP.  Our main result shows that this algorithm is polylogarithmic-competitive for general graphs, established in Section~\ref{sec:online}.  

\begin{restatable}{theorem}{mainbalancedpivot}\label{thm:BalancedPivotUpper}
\BP\ has a competitive ratio of $\bigo(\log^4 n)$ in expectation for online correlation $k$-clustering under random-order vertex arrivals.
\end{restatable}
As a warm-up to our main theorem, we present a poly-logarithmic competitive algorithm for the special case of disjoint clique graphs.
This is via an alternative approach to \BP, which turns out to be somewhat easier to analyze: when a new pivot arrives, place it in a cluster chosen uniformly at random from a set of clusters that are ``among the least loaded.''  For the special case of disjoint cliques, we consider a randomized algorithm of this form, which we call \UP.  In Section~\ref{sec:cliques}, we formally define \UP\ and present the following result.

\begin{restatable}{theorem}{cliquesthm}\label{thm-cliques-random}
\UP\ has a competitive ratio of $\bigo(\log^3 n)$ in expectation for online correlation $k$-clustering of disjoint cliques under random-order vertex arrivals. 
\end{restatable}

In addition, all our theorems have been verified using Lean~\cite{lean-lang-website} and are available online at \cite{lean-formalizations}. 
We take some standard external results as unproven hypotheses: Chernoff Bounds and Hypergeometric Chernoff Bounds (Facts~\ref{lem:chernoff} and \ref{lem:chernoff-hypergeometric}), some facts about negatively associated distributions (\Cref{lem:perm_distr_na,lem:na_closure}), and the 3-approximation guarantee for the \Pivot\ algorithm \cite{DBLP:conf/stoc/AilonCN05}.
Otherwise, the main Theorems and Lemmas use only Lean's standard axioms.
Please see our statement on AI usage in \Cref{app:ai} for details on how we generated this code.

\subsection{Related Work}\label{sec: related work}
\paragraph{Correlation Clustering.} 
Correlation clustering was introduced by Bansal, Blum, and Chawla~\cite{bansal2004correlation}. Their work focused on complete unweighted graphs and gave approximation algorithms for both the minimization and maximization versions. Although the minimization and maximization versions essentially correspond to dual views of the problem, they are incomparable from the point of view of approximation. 
In particular, the maximization version admits both a polynomial time approximation scheme (PTAS) \cite{bansal2004correlation}, and an SDP-based $0.766$ approximation algorithm for general weighted graphs \cite{DBLP:conf/soda/Swamy04}, which holds even when the number of clusters is upper bounded by 6. 
Thus, most of the work in the past three decades has focused on the minimization version.

Following \cite{bansal2004correlation}, Demaine, Emanuel, Fiat, and Immorlica~\cite{DBLP:journals/tcs/DemaineEFI06} extended the problem to general weighted graphs and gave logarithmic-factor approximations. 
Ailon, Charikar, and Newman~\cite{DBLP:conf/stoc/AilonCN05} proposed a greatly improved $3$-approximation randomized algorithm, known as the Pivot algorithm.
Due to its simple, combinatorial flavor, this has inspired several later works \cite{ChakrabartyM23, BehnezhadCCGM24, DBLP:conf/icml/DalirrooyfardMM24}, including this one. 
The minimization version of correlation clustering has been extensively studied from the lens of linear-programming (LP)-based approaches. 
Chawla, Makarychev, Schramm, and Yaroslavtsev~\cite{DBLP:conf/stoc/ChawlaMSY15} designed a near-optimal algorithm achieving a $2.06$ approximation on complete graphs, and also gave improved approximations for complete $k$-partite graphs.

Recently, much stronger LP relaxations have been explored: Cohen-Addad, Lee, and Newman~\cite{DBLP:conf/focs/Cohen-AddadLN22} used $\bigo(1/\varepsilon^2)$ rounds of the Sherali-Adams hierarchy to achieve a $(1.994+\varepsilon)$-approximation, and Cohen-Addad, Lee, Li, and Newman~\cite{DBLP:conf/focs/Cohen-AddadL0N23} further improved this to a $1.73$-approximation via a preclustering step that controls correlated rounding error. 
Cao, Cohen-Addad, Lee, Li, Newman, and Vogl~\cite{DBLP:conf/stoc/CaoCL0NV24} analyzed the cluster LP formulation; subsequent work by the same group~\cite{DBLP:conf/stoc/CaoCL0LNTVY025} gave a sublinear-time algorithm for solving the cluster LP. 
A common structural tool utilized in some of these recent works, as well as in \cite{DBLP:conf/icml/Cohen-AddadLMNP21}, is the identification of dense ``almost-cliques.'' While we do not use an almost-clique decomposition in our analyses, we use the fact that an optimal unconstrained correlation clustering solution consists of clusters with high density.

Recent works have also revisited combinatorial approaches to correlation clustering, avoiding heavy LP-based methods. Cordner and Kollios~\cite{DBLP:conf/cocoa/CordnerK23} proposed a parallelizable single-pass local-search refinement applied within \textsc{Pivot} clusters and evaluate it empirically on large graphs.
Their approach avoids LP relaxations and achieves near-optimal approximation guarantees on several graph families. Cohen-Addad, Rasmussen Lolck, Pilipczuk, Thorup, Yan, and Zhang~\cite{CohenAddadLPSTY24} presented a significant combinatorial approach for correlation clustering on general graphs, breaking the classical factor-3 barrier of \textsc{Pivot} with a better than 1.847-approximation, primarily relying on local search techniques; they also present results for the streaming and MPC settings.

\paragraph{Correlation $k$-clustering.} In this paper, we study correlation $k$-clustering, where the number of clusters is upper bounded by a parameter $k$. While the standard (unconstrained) correlation clustering has been well-studied, the literature on correlation $k$-clustering is sparse. 

It is known that the (unconstrained) correlation clustering problem is APX-hard~\cite{DBLP:journals/jcss/CharikarGW05}; this hardness extends to the correlation $k$-clustering problem in general, since they are identical when $k=n$. 
On the other hand, Giotis and Guruswami~\cite{GiotisG06} studied correlation $k$-clustering when $k$ is a constant and presented a PTAS using a sample-and-prune framework. 
Their work demonstrates that near-optimal solutions can be achieved for a constant $k$, albeit with running time exponential in $k$. 
As a result, their approach does not extend to yield polynomial time approximation algorithms when $k$ is super constant.
To the best of our knowledge, our work is the first to study the correlation $k$-clustering problem for general $k$.

\paragraph{Online setting.} In the online setting, work has focused on the disagreement minimization version of the CC problem. Mathieu, Sankur, and Schudy~\cite{MathieuSS10} initiated the study of correlation clustering in the online setting. They showed that under adversarial arrivals, the competitive ratio of any algorithm is lower bounded by $\Omega(n)$. For the maximization version, they gave a $0.5$-competitive algorithm.

Recent work by Cohen-Addad, Lattanzi, Maggiori, and Parotsidis~\cite{DBLP:conf/icml/Cohen-AddadLMP22} gave a constant-competitive online algorithm under adversarial arrivals; the caveat here is that this result holds when recourse is allowed, i.e. vertices are allowed to change cluster assignments. Their algorithm obtains $\bigo(\log n)$ recourse per vertex which they show is tight for any constant-competitive algorithm. Given the strong $\Omega(n)$ lower bound in the online setting under adversarial vertex arrivals, Lattanzi, Moseley, Vassilvitskii, Wang, and Zhou~\cite{DBLP:conf/nips/LattanziMVWZ21} study online correlation clustering in a beyond worst-case (``semi-online") model.
To the best of our knowledge, no non-trivial competitive algorithms exist for online correlation $k$-clustering.

Correlation clustering has also been extensively studied in a variety of other computational settings including streaming \cite{DBLP:conf/icml/AhnCGMW15, DBLP:conf/innovations/Assadi022, DBLP:conf/nips/AssadiSW23, BehnezhadCCGM24, DBLP:journals/corr/abs-2411-09979, DBLP:conf/stoc/AssadiKP25, DBLP:journals/corr/abs-2504-12060, DBLP:journals/theoretics/CambusKLPU25}, 
distributed \cite{DBLP:conf/icml/Cohen-AddadLMNP21, DBLP:conf/focs/BehnezhadCMT22, DBLP:conf/soda/CaoHS24}, and the learning augmented model \cite{DBLP:journals/corr/abs-2510-10705}.

\section{Technical Overview}\label{sec:tech-overview}

Before presenting our techniques, we define the correlation $k$-clustering problem and our random-order vertex arrival model below.
\paragraph{Correlation $k$-clustering.}
In the \emph{correlation $k$-clustering problem}, we are given as input a graph $G=(V,E)$ on $n=|V|$ vertices, edges $e\in E$. A correlation $k$-clustering (or $k$-clustering) of $G$ is a partition of $V$ into \textit{at most} $k$ subsets of $V$.  Let $\mathcal{C} = \{C_1, C_2, \dots, C_k\}$ denote a $k$-clustering of $G$ (note that $C_i$ could be potentially empty for $i\in [k]$). For distinct vertices $u$ and $v$, pair $(u,v)$ is in \emph{disagreement} with respect to $\mc$ if one of the following holds: i) $(u,v) \in E$ and $u\in C_i, v\in C_j$ where $i\neq j$, or ii) $(u,v) \notin E$ and $u,v\in C_i$ for some $i\in [k]$. Let $\phi(u,v)$ be an indicator variable for a pair $(u,v)$ of vertices which is $1$ if $(u,v)$ is in disagreement with $\mc$ and $0$ otherwise. The \textit{cost} of a $k$-clustering $\mc$ is the total disagreements with respect to $\mc$, i.e. $\cost(\mc) = \sum_{\mathrm{unordered} \, u,v\in V}\phi(u,v)$. 
The goal in $k$-correlation clustering is to obtain a $k$-clustering with minimum cost. 

\paragraph{Online Random Order \kCC\ (Vertex Arrivals).}  In the online \kCC\ problem with vertex arrivals, vertices of the input graph $G=(V,E)$ arrive online at discrete time steps. Let $v_t$ for $t\in [1,n]$ denote the vertex revealed at time step $t$. When $v_t$ arrives, all edges of the form $(v_{t'}, v_t)$ where $t'<t$ are revealed. An online algorithm is required to irrevocably assign $v_t$ to one of the (at most $k$) clusters it maintains, at time $t$.  In the \textit{online random order} \kCC\ problem, which we study in this paper, the order of vertex arrivals is chosen uniformly at random from the set of all possible permutations of vertices in $V$.

\subsection{Constant-factor Approximation Algorithm}
\label{sec:overview.offline}
The starting point for our offline approximation algorithm is a derivation of suitable lower bounds on the optimal cost for any given instance. 
An immediate lower bound is the cost of an optimal (unconstrained) correlation clustering solution.  This lower bound is, however, not sufficient; it is zero for disjoint cliques, while the cost of an optimal $k$-clustering can be large if the instance has more than $k$ cliques. 
Indeed, an asymptotically tight lower bound on the cost of an optimal $k$-clustering for an arbitrary disjoint cliques instance is an integral part of not just our approximation algorithm analysis, but also of our main result for online correlation $k$-clustering.

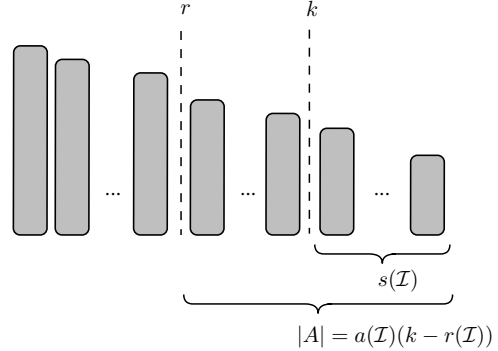
\begin{wrapfigure}{r}{0.45\textwidth}
  \centering
  \vspace{-15pt}

    \resizebox{0.40\textwidth}{!}{
    \tikzset{every picture/.style={line width=0.75pt}}         
    
    \begin{tikzpicture}[scale=0.9,x=0.75pt,y=0.75pt,yscale=-1,xscale=1]
    
    \draw [fill=lightgray, draw=black, line width=0.8pt  ,fill opacity=1 ](165.02,66.99) .. controls (167.77,67) and (170,69.23) .. (169.99,71.98) -- (169.79,201.98) .. controls (169.78,204.73) and (167.55,206.96) .. (164.8,206.95) -- (149.86,206.93) .. controls (147.11,206.92) and (144.89,204.69) .. (144.89,201.94) -- (145.1,71.94) .. controls (145.1,69.19) and (147.33,66.96) .. (150.08,66.97) -- cycle ;
    \draw [fill=lightgray, draw=black, line width=0.8pt  ,fill opacity=1 ] (196.01,77.01) .. controls (198.76,77.02) and (200.98,79.25) .. (200.98,82) -- (200.79,201.98) .. controls (200.78,204.73) and (198.55,206.96) .. (195.8,206.95) -- (180.86,206.93) .. controls (178.11,206.92) and (175.89,204.69) .. (175.89,201.94) -- (176.08,81.96) .. controls (176.08,79.21) and (178.32,76.98) .. (181.07,76.99) -- cycle ;
    \draw [fill=lightgray, draw=black, line width=0.8pt  ,fill opacity=1 ](253.99,87.02) .. controls (256.74,87.02) and (258.97,89.25) .. (258.96,92) -- (258.79,201.98) .. controls (258.78,204.73) and (256.55,206.96) .. (253.8,206.95) -- (238.86,206.93) .. controls (236.11,206.92) and (233.89,204.69) .. (233.89,201.94) -- (234.06,91.96) .. controls (234.07,89.21) and (236.3,86.99) .. (239.05,86.99) -- cycle ;
    \draw [fill=lightgray, draw=black, line width=0.8pt  ,fill opacity=1 ] (295.96,107.02) .. controls (298.71,107.02) and (300.93,109.25) .. (300.93,112) -- (300.79,201.98) .. controls (300.78,204.73) and (298.55,206.96) .. (295.8,206.95) -- (280.86,206.93) .. controls (278.11,206.92) and (275.89,204.69) .. (275.89,201.94) -- (276.03,111.97) .. controls (276.04,109.22) and (278.27,106.99) .. (281.02,106.99) -- cycle ;
    \draw [fill=lightgray, draw=black, line width=0.8pt  ,fill opacity=1 ] (351.94,117.02) .. controls (354.69,117.02) and (356.92,119.25) .. (356.91,122) -- (356.79,201.98) .. controls (356.78,204.73) and (354.55,206.96) .. (351.8,206.95) -- (336.86,206.93) .. controls (334.11,206.92) and (331.89,204.69) .. (331.89,201.94) -- (332.02,121.96) .. controls (332.02,119.21) and (334.25,116.99) .. (337,116.99) -- cycle ;
    \draw [fill=lightgray, draw=black, line width=0.8pt  ,fill opacity=1 ] (391.92,128.01) .. controls (394.67,128.02) and (396.9,130.25) .. (396.9,133) -- (396.79,201.98) .. controls (396.78,204.73) and (394.55,206.96) .. (391.8,206.95) -- (376.86,206.93) .. controls (374.11,206.92) and (371.89,204.69) .. (371.89,201.94) -- (372,132.96) .. controls (372,130.21) and (374.24,127.99) .. (376.99,127.99) -- cycle ;
    \draw [fill=lightgray, draw=black, line width=0.8pt  ,fill opacity=1 ] (458.89,148.01) .. controls (461.64,148.02) and (463.87,150.25) .. (463.86,153) -- (463.79,201.98) .. controls (463.78,204.73) and (461.55,206.96) .. (458.8,206.95) -- (443.86,206.93) .. controls (441.11,206.92) and (438.89,204.69) .. (438.89,201.94) -- (438.97,152.96) .. controls (438.97,150.21) and (441.21,147.98) .. (443.96,147.99) -- cycle ;
    \draw  [dash pattern={on 4.5pt off 4.5pt}]  (269,56) -- (269,207) ;
    \draw  [dash pattern={on 4.5pt off 4.5pt}]  (364,55) -- (364,206) ;
    \draw   (368,214) .. controls (368,218.67) and (370.33,221) .. (375,221) -- (408,221) .. controls (414.67,221) and (418,223.33) .. (418,228) .. controls (418,223.33) and (421.33,221) .. (428,221)(425,221) -- (461,221) .. controls (465.67,221) and (468,218.67) .. (468,214) ;
    \draw   (271,251) .. controls (271,255.67) and (273.33,258) .. (278,258) -- (360.5,258) .. controls (367.17,258) and (370.5,260.33) .. (370.5,265) .. controls (370.5,260.33) and (373.83,258) .. (380.5,258)(377.5,258) -- (463,258) .. controls (467.67,258) and (470,255.67) .. (470,251) ;
    
    \draw (211,174) node [anchor=north west][inner sep=0.75pt]   [align=left] {...};
    \draw (311,174) node [anchor=north west][inner sep=0.75pt]   [align=left] {...};
    \draw (410,174) node [anchor=north west][inner sep=0.75pt]   [align=left] {...};
    \draw (360,32) node [anchor=north west][inner sep=0.75pt]   [align=left] {$k$};
    \draw (267,35) node [anchor=north west][inner sep=0.75pt]   [align=left] {$r$};
    \draw (413,231) node [anchor=north west][inner sep=0.75pt]   [align=left] {$\sml{\mathcal{I}}$};
    \draw (353,272) node [anchor=north west][inner sep=0.75pt]   [align=left] {$|A| = \avg{\mathcal{I}}(k-r(\instance))$};

    \end{tikzpicture}
    }

    \caption{Illustration of an instance $\instance$}
    \label{fig:instance}
  \vspace{-10pt}
\end{wrapfigure}

Our algorithm and analysis for the disjoint cliques case relies on a few simple observations.
(1), in any optimal $k$-clustering, no clique gets split across multiple clusters (\Cref{lem:cliques-intact}). 
However it is not always clear how to optimally partition the collection of cliques into $k$ clusters. 
(2), the average load of clusters will be $\frac{n}{k}$. 
Thus if any clique is of size at least $\frac{n}{k}$, an optimal solution will place it entirely in its own cluster.
Finally (3), in an optimal solution, every cluster $C_j$ which contains more than $\frac{n}{k}$ vertices must meet the following property: for every clique $K_i$ assigned to cluster $C_j$, the size of $C_j\setminus K_i$ must be strictly smaller than $\frac{n}{k}$.
That is, if a cluster has size larger than $\frac{n}{k}$, it is only larger by one clique.

Any disjoint cliques instance admits a simple greedy algorithm.  Sort all cliques in decreasing order of size.
Assign each of the largest $k$ cliques to their own cluster.  For any remaining cliques, in decreasing order of size, greedily assign them to the cluster with the smallest load. This simple algorithm admits a 2-approximation (\Cref{lem:opt-theta}).  We also present a PTAS for disjoint cliques using the approach of~\cite{alon1998approximation} in \Cref{app:PTAS}.  

For any instance $\instance = (V,E)$ containing disjoint cliques, we define parameters $r(\instance)$, $\avg{\instance}$, and $\sml{\instance}$. (\Cref{def:s-of-instance,def:threshold}. See \Cref{fig:instance} for a visual representation.)
Assume $|K_i| = n_i$ and that cliques are ordered in decreasing size $n_1\geq\hdots\geq n_\ell$.
Then $r=r(\instance)$ is the index of the smallest clique which must be clustered individually by observation (2) (applied iteratively). 
We set $\avg{\instance} = \frac{1}{k-r} \sum_{i> r} n_i$ to be the average load of the remaining clusters.
And finally, we set $\sml{\instance}=\sum_{i>k} n_i$ to be the total number of vertices which do not lie in the $k$ largest cliques.

First, we show that $\cost(\mbox{Greedy})\leq \avg{\instance} \sml{\instance}$.
Since cliques remain intact, disagreement cost comes exclusively from co-clustering vertices from different cliques.
We can decompose $\cost(\mbox{Greedy})$ by examining all $(\ell-k)$ smallest cliques $K_i$ and summing $n_i$ times the load of its cluster when it is assigned. The greedy algorithm assigns $K_i$ to a currently least-loaded cluster among $C_{r+1}, \ldots, C_k$, whose total combined size (across all cliques $K_{r+1}, \ldots, K_\ell$ is $(k-r)\avg{\instance}$.  Averaging over the $(k-r)$ clusters, this load is no more than $\avg{\instance}$.
Thus, the cost is no more than $\avg{\instance} \sml{\instance}$ by definition of $\sml{\instance}$.

Second, we show that $\frac{1}{2} \avg{\instance}\sml{\instance}$ is a valid lower bound on the optimal $k$-clustering cost. We can assume that in an optimal solution, the first $r$ largest cliques are assigned to distinct clusters, and no other cliques are assigned to those clusters. We express $\cost(\opt)$ as $\frac{1}{2}(\sum_{j=r+1}^k c_j^2 - \sum_{j=r+1}^{\ell}n_j^2)$ where $c_j$ denotes the size of cluster $C_j$. By noting that $\sum_{j=r+1}^k c_j=(k-r)a(\instance)$ and using the Cauchy-Schwarz inequality, we obtain a lower bound on $\sum_{j=r+1}^k c_j^2$. Next, we upper bound $\sum_{j=r+1}^{\ell} n_j^2$ by $a(\instance)\sum_{j=r+1}^k n_j$.
Putting the terms together yields $\cost(\opt) \geq \frac{1}{2}a(\instance)s(\instance)$.

In the general case where our underlying instance is not necessarily disjoint cliques, we show that treating the clusters returned by an unconstrained correlation clustering algorithm $\mathcal{A}$ as ``cliques'' and placing these "disjoint cliques" into $k$ clusters using algorithm $\mathcal{B}$ (e.g., the PTAS of \Cref{app:PTAS}) also returns an approximately optimal solution, as long as the clusters from $\mathcal{A}$ were approximately optimal for the \CC\ problem and $\mathcal{B}$ is approximately optimal for the disjoint cliques instances.  We describe this algorithm and its proof in detail in \Cref{sec:offline}.

\begin{figure}[t]
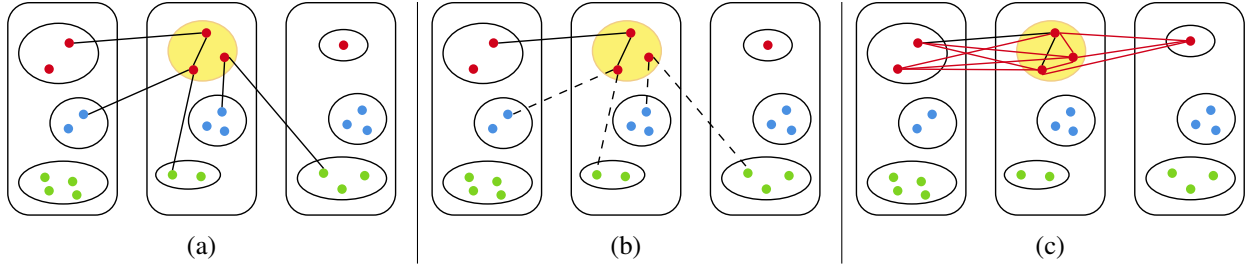

  \centering

  \begin{subfigure}{0.32\textwidth}
    \centering
    \resizebox{\linewidth}{!}{\input{fig_cc_1}}
    \caption{}
    \label{fig:cc-sub1}
\end{subfigure}
\hfill
\vrule width 0.4pt
\hfill
\begin{subfigure}{0.32\textwidth}
    \centering
    \resizebox{\linewidth}{!}{\input{fig_cc_2}}
    \caption{}
    \label{fig:cc-sub2}
\end{subfigure}
\hfill
\vrule width 0.4pt
\hfill
\begin{subfigure}{0.32\textwidth}
    \centering
    \resizebox{\linewidth}{!}{\input{fig_cc_3}}
    \caption{}
    \label{fig:cc-sub3}
\end{subfigure}

  \caption{Given an instance $\instance$, the $\alpha$-approximation
  solution $C$ to the unconstrained correlation clustering problem
  is made of three groups: red, blue, and green. The black borders
  show the optimal solution to the correlation $k$-clustering problem.
  To create instance $\instance'$ from $\instance$, we delete edges
  between different colors and add edges among them. This figure shows
  the process for only the red vertices of the middle cluster.}
  \label{fig:cc-to-kcc}
\end{figure}

\subsection{Lower Bound for Online Correlation $k$-Clustering}
\label{sec:overview.online.lower}
In Section~\ref{sec:lower-bound}, we establish that any online algorithm for correlation $k$-clustering under random vertex arrivals must incur a competitive ratio of $\Omega(\log k)$. 
In particular, since $k$ can be as large as $n$, this implies an $\Omega(\log n)$ lower bound when $k=\text{poly}(n)$, which it is in our hard instance below.

We prove this via a carefully constructed instance
consisting of $(k-1)$ disjoint cliques each of size $2\ln k$, together with two singleton vertices, for a total of $n = 2(k-1)\ln k + 2$ vertices. 
The optimal solution assigns each clique to a unique cluster and assigns the singletons together into a single cluster, incurring unit cost in total.
However, we show that any online algorithm incurs an expected cost of $\Omega(\log k)$ on this instance. 
Our probabilistic argument has two parts. 

First, we show that with probability $\Omega(1)$, both singleton vertices arrive by time $ck\ln k$, where $c>0$ is a small constant (Lemma \ref{lem:singletons_arrive}). 
Second, we show that with probability at least $\Omega(1)$ there exists at least one clique for which no vertices arrive by time $ck\ln k$ (Lemma \ref{lem:clique_missing}). 
It follows that with constant probability (and hence a constant fraction of arrival orders), there exists at least one clique whose first vertex arrives after both singletons have arrived. 
Thus, for a constant fraction of arrival orders, there exists a time step $t$ at which there are three vertices $v_1, v_2, v_3$, none of whose neighbors have arrived by time $t$, and where $v_1, v_2$ are singletons and $v_3$ is a vertex from some clique. 
By symmetry, there exist at least $\Omega(1)$ fraction of arrival orders, each of which has a unique time step $t$ at which there exist three vertices for which none of the neighbors have arrived. 
For each such arrival order, there is no way for an algorithm to distinguish between the singleton vertices and the first vertex from some clique at time $t$.  
By or during time $t$, the algorithm must choose to cluster vertices together which do not share edges. The only choice in which an online algorithm can incur optimal cost is if it places the two singleton vertices together. 
In all other cases, the cost incurred is at least $\Omega(\log k)$.
Since the algorithm cannot distinguish between singleton vertices and the vertex of at least one clique by the time it must make such a decision, we may conclude that with constant probability the algorithm makes a mistake in which uncorrelated vertices to place together. 
This incurs a cost of $\Omega(\log k)$, yielding the lower bound for random-order online \kCC.

Additionally in \Cref{sec:lower-bound}, we give a stronger $\Omega(\log^2 k)$ lower bound on algorithms which do not split disjoint cliques, using the same hard graph instance.
This relies on a more careful analysis, showing that with probability $\Omega(1)$ there exist at least two cliques for which no vertices have arrived by time $ck \ln k$ for a small constant $c>0$ (\Cref{lem:lb_2clique_arrive_late}). 
A similar argument allows us to conclude that with constant probability, any online algorithm will place two clique vertices into the same cluster. 
Since such an algorithm cannot later split these cliques, this mistake commits it to paying the disagreement cost of clustering both cliques fully together, which is $\Omega(\log^2 k)$.

\subsection{$\bigo(\polylog(n))$-Competitive Algorithm for Disjoint Cliques}
\label{sec:overview.online.upper.cliques}

As mentioned in Section~\ref{sec:overview.offline}, an asymptotically optimal bound on the cost of an optimal $k$-clustering for the disjoint cliques instance $\instance$ is given by $\Theta(\sml{\instance}\avg{\instance})$.  
Our main result for disjoint cliques in the \textit{online} setting is an algorithm \UP\ which incurs expected cost $\bigo(\sml{\instance}\avg{\instance}\log^3 n)$ on any disjoint cliques instance $\instance$, and is therefore $\bigo(\log^3 n)$-competitive (\Cref{thm-cliques-random}).

The \UP\ algorithm operates in two stages, described in detail in \Cref{sec:cliques}. 
The first stage corresponds to the \Pivot\ algorithm until the $k^{th}$ pivot arrives; let us call this (random) time step $\Tau_0$. At time $\Tau_0$, each of the $k$ clusters contains vertices from exactly one clique, which we call its \textit{leading clique}, and we let $S$ denote the clusters with $\bigo(\log n)$ vertices at time $\Tau_0$. For every vertex arriving at time $t>\Tau_0$, if it is designated as a pivot, it is assigned to a cluster chosen uniformly at random from $S$; otherwise it is assigned to the cluster to which its earliest pivot neighbor is assigned.

To analyze the cost of \UP, we first note that the cost incurred until time $\Tau_0$ is zero, and that no clique is ever split. 
Hence it suffices to consider only costs incurred by assigning two cliques to the same cluster. 
Let $U$ denote the set of \textit{uncovered} vertices at time $\Tau_0$, i.e., vertices of cliques none of whose vertices has arrived by time $\Tau_0$. 
If the first $k$ pivots come from the $k$ largest cliques, then $|U|= \sml{\instance}$. 
However as the lower bound instance demonstrates, with non-trivial probability this may not hold. 
Each clique $K\subseteq U$ is assigned to a uniformly random cluster of $S$, and the total cost is the sum of two types of costs:
\begin{enumerate}[label=\arabic*)]
    \item \textit{Type 1:} the cost between each clique $K\subseteq U$ and the leading clique of its cluster;
    \item \textit{Type 2:} the cost between two cliques $K, K'\subseteq U$ assigned to the same cluster.
\end{enumerate}

The main difficulty is that $|U|$ and the sizes of the cliques involved are strongly correlated with each other through the random time $\Tau_0$. 
If $\Tau_0$ is small, the leading cliques placed into $S$ and the cliques in $U$ may be large, which could blow up the overall cost. 
However, $\Tau_0$ is unlikely to be small.
In order to carefully control the correlations dependent on $\Tau_0$, we split the $n$ total timesteps into $\log n$ intervals $[2^i, 2^{i+1})$, and for each interval we bound the contribution of the event $\Tau_0\in [2^i, 2^{i+1})$ to the overall expected cost.
Summing each of the $\log n$ contributions together gives an overall bound on the total expected cost.
If $\Tau_0 \in [2^i, 2^{i+1})$, then we can establish some good events which occur with high probability.

If $\Tau_0\ge 2^i$, then leading cliques in $S$ and cliques in $U$ are both small with high probability. ~Using a straightforward concentration bound, we establish that for any time $t$, the number of vertices that arrive by time $t$ from a clique $K$ of size $|K|=\Omega(\frac{n\log n}{t})$ is $\Omega(\log n)$ with high probability, or whp\footnote{Throughout the paper, we say that an event $\mathcal{E}$ occurs with high probability (abbreviated as whp) if $\Pr[\mathcal{E}] \ge 1 - \frac{1}{n^c}$ for a suitably large constant $c>0$.} (\Cref{lem:few-vertices-by-tau}). Consequently, if $\Tau_0\ge 2^i$, then whp every clique contained in $U$, and every leading clique of a cluster in $S$, has size less than $B_i=\bigo(\frac{n\log n}{2^i})$. In particular, in this situation the Type 1 cost is at most $B_i|U|$. Moreover, every one of the $k$ largest cliques that is contained in $U$ was \textit{displaced} by one of the first $k$ pivots coming from a smaller clique $K_j$, $j>k$. Hence $|U|\le \sml{\instance} + B_i Y_i$, where $Y_i$ is the number of vertices of $\bigcup_{j>k}K_j$ among the first $2^{i+1}$ arrivals.

If $\Tau_0< 2^{i+1}$, then at least $k-r$ vertices from non-top-$r$ cliques arrived by time $2^{i+1}$.  
Producing $k$ pivots by time $2^{i+1}$ forces at least $k-r$ of them to come from the set $A = \bigcup_{j>r}K_j$ of non-top-$r$ cliques, which by definition has $\avg{\instance}(k-r)$ vertices. 
Letting $N_i$ be the number of vertices from $A$ which arrived in the first $2^{i+1}$ timesteps, then $\Tau_0< 2^{i+1}$ implies that $N_i\ge k-r$. 
Since $\E[N_i]\le\avg{\instance}(k-r)p_i$, where $p_i=\frac{2^{i+1}}{n}$, this is unlikely when $i$ is small.

\paragraph{Type 1 costs.}
Type 1 costs are bounded above by the size of $U$ times the size of leading cliques which get designated as part of $S$.
We showed that when $\Tau_0 \geq 2^i$, then whp leading cliques in $S$ have size no more than $B_i=\bigo\left(\frac{n\log n}{2^i}\right)$ and whp $|U|\leq \sml{\instance} + B_i Y_i$.
This allows us to construct a whp bound on the Type 1 costs of no more than $B_i(\sml{\instance} + B_i Y_i)$.
From there, we very carefully show how to use this relationship to bound the contribution of the event $\Tau_0\in [2^i, 2^{i+1})$ to the overall expected cost.
In all cases the contribution of each interval is $\bigo( \sml{\instance} \avg{\instance} \log^2 n)$, and summing over the $\bigo(\log n)$ intervals gives a total Type 1 cost of $\bigo( \sml{\instance} \avg{\instance} \log^3 n)$.

\paragraph{Type 2 costs.} Conditioned on the arrival order, the cliques in $U$ are assigned independently and uniformly at random to clusters in $S$. So the expected Type 2 cost equals $F/|S|$, where $F=\sum |K||K'|$ is taken over all pairs of distinct cliques $K,K'\subseteq U$. Splitting the pairs according to whether each clique is among the $k$ largest, and using that the displaced cliques have size less than $B_i$, we bound $F$ by $\sml{\instance}^2+\sml{\instance}B_iY_i+B_i^2Y_i(Y_i-1)$. Since $Y_i(Y_i-1)$ counts pairs of early arrivals, this yields $\E[F]=\bigo(\sml{\instance}^2\log^3 n)$. In \Cref{lem:logn_set_size}, we prove that $|S|\ge k-r$ whp. Using $\sml{\instance}\le\avg{\instance}(k-r)$, we can bound the total Type 2 cost by $\bigo(\sml{\instance}\avg{\instance}\log^3 n)$.

Combining the Type 1 and Type 2 costs, the expected cost of \UP\ is $\bigo(\sml{\instance}\avg{\instance}\log^3 n)$. Together with the lower bound of \Cref{lem:opt-theta}, this yields $\bigo(\log^3 n)$-competitiveness.

\subsection{$\bigo(\polylog(n))$-Competitive Algorithm for General Graphs}
Our online \kCC\ algorithm \BP\ is a simple adaptation of the \textsc{Pivot} algorithm.  When a new vertex $v$ arrives, if any of its previously clustered neighbors has been designated as a pivot, $v$ is placed into the same cluster as that of its earliest arriving neighbor pivot; otherwise $v$ is designated as a pivot and is placed into the cluster with least load.  While the analysis for disjoint cliques provides a blueprint for an approach for general graphs, there are a number of technical challenges.  We highlight two in this overview and briefly describe how we overcome them.

\noindent \textbf{Challenge 1: Lower bound on optimal cost.}  A starting point for a lower bound for the optimal \kCC\ cost for $\instance$ is the optimal \CC\ cost for $\instance$.  
However, as we have noted before this lower bound can be very weak.  
If the instance $\instance$ were a collection of disjoint cliques, we addressed this issue by deriving an asymptotically tight lower bound in terms of the parameters $\sml{\instance}$ and $\avg{\instance}$. 
For a general graph instance $\instance$, however, it is not clear how to define these parameters. 

To derive parameters analogous to $\sml{\instance}$ and $\avg{\instance}$ for disjoint cliques, we observe that an optimal \CC\ solution for an arbitrary graph instance $\instance$ yields a partition of the vertex set into clusters $Q_1, Q_2, \ldots, Q_\ell$ (for some $\ell$).  
While edges may be missing in any $Q_j$ and there may be edges crossing any $Q_j$, each $Q_j$ is an ``almost-clique'' in the sense that each vertex in $Q_j$ is adjacent to at least $\frac{1}{2}(|Q_j| - 1)$ vertices (otherwise, the vertex can be separated from $Q_j$, yielding a better \CC\ solution).  
We define a disjoint cliques instance $\instance_{\mathrm{clique}}$, which consists of a clique of size $|Q_j|$ for each $j$.  A key lower bound on optimal cost is $\Omega(\opt(\instance_{\mathrm{clique}})) + e_m + e_c$, where $e_m$ and $e_c$ are the missing edges within the almost-cliques and the edges crossing almost-cliques, respectively.  Note that $\opt(\instance_{\mathrm{clique}})$ can be expressed in terms of $\sml{\instance_{\mathrm{clique}}}$ and $\avg{\instance_{\mathrm{clique}}}$, which are both deterministic parameters, independent of the random arrival sequence $\sigma$.

For our analysis, it is also helpful to derive another lower bound based on an instance of disjoint cliques more closely associated with how \textsc{Pivot} operates. 
For a random vertex arrival sequence $\sigma$, let $\mathcal{J}_\sigma$ denote the disjoint cliques instance that contains, for each pivot cluster $C$ defined by \textsc{Pivot}, a clique of size $|C|$. In Lemma~\ref{lem:sml-avg-bound-pivot}, we show that the optimal \kCC\ cost is also $\Omega(\E[\sml{\mathcal{J}_\sigma} \avg{\mathcal{J}_\sigma}])$.

\noindent \textbf{Challenge 2: Impact of $U$, the set of vertices uncovered by the first $k$ pivots.} 
Unlike the case of disjoint cliques, any random vertex arrival sequence $\sigma$ can split an almost-clique $Q_j$ of the decomposition discussed above. 
We focus on the $k$ largest almost-cliques and the impact of the \emph{first pivot} $v$ that arrives from $Q_j$, $1 \le j \le k$, at some time $t$. 
By this time, some other vertices in $Q_j$ may have already arrived and clustered with pivots from other almost-cliques; the cost incurred due to the resulting disagreements can be attributed to the cost of Pivot (which is within a constant factor of the optimal).

The remaining cost of \BP\ can be upper bounded by the sum, over each vertex $v$ in $U$, of the following two quantities: 
(i) the number of other vertices from $U$ that also land in the same cluster as $v$, and (ii) the size of the first ``pivot cluster'' that is in the same cluster as $v$. 
For the analysis of \UP\ on disjoint cliques, there are at least $k-r(I)$ clusters with ``low load'', and the pivots from $U$ are placed uniformly at random among these clusters. We then bound the first cost using the number of early arrivals from fixed sets of vertices.  For \BP, however, we need to condition on the random time at which the $k$th pivot arrives.
Furthermore, \BP\ assigns pivots to the least loaded cluster, and bounding the projected load of the least loaded cluster at any time $t$ poses additional challenges.

We overcome the above obstacles through several analytical steps.  Let $\Tau$ denote the time step immediately preceding the arrival of the $(k+1)$st pivot (or $n$ if at most $k$ pivots arrive). First, we derive a general high probability bound on the projected load of the least loaded cluster at any time $t \ge \Tau$ in terms of $t$, $|U|$, and the size of the largest pivot cluster in an arbitrary set of clusters (Lemma~\ref{lem:proj_load_of_least_load}).  This is a key technical lemma, which builds on estimates we derive for the projected final load for a pivot cluster as well as the size of a pivot cluster at an arbitrary instant (see Section~\ref{sec:online.upper.analysis.projected}).  In the proof of Theorem~\ref{thm:BalancedPivotUpper}, we then pursue different approaches to bounding the total expected cost, depending on the size of $U$.  Let $S$ denote the set of clusters that have $\bigo(\log n)$ vertices at time $\Tau$.  When $|U|$ is $\widetilde{\Omega}(n|S|/\Tau)$, then since each pivot cluster in $U$ is of size $\widetilde\bigo(n/\Tau)$, we are able to argue that the cost of \BP\ can be well-approximated by the cost incurred by distributing the pivot clusters of $U$ among the clusters in $S$, which we then show is within a polylogarithmic factor of the lower bound $\E[\sml{\mathcal{J}_\sigma} \avg{\mathcal{J}_\sigma}]$.  When $|U|$ is $\widetilde\bigo(n|S|/\Tau)$, our analysis uses the projected load bound of Lemma~\ref{lem:proj_load_of_least_load} and follows an approach similar to the one for disjoint cliques; we split the analysis into two sub-cases depending on the (deterministic) parameters of the instance $\instance_{\mathrm{clique}}$.  The full argument is more involved and we refer to Section~\ref{sec:online.upper.analysis.theorem} for details.

\subsection{Preliminaries: Concentration Inequalities}

We will make repeated use of the following standard concentration inequalities throughout the paper.

\begin{fact}[Chernoff Bound]\label{lem:chernoff}
Let $X = \sum_{i=1}^n X_i$ where $X_i \in \{0,1\}$ are independent (or negatively associated) 
random variables, and let $\mu = \mathbb{E}[X]$. Then for any $\delta > 0$:
\begin{equation}
    \Pr[X \geq (1+\delta)\mu] \leq \exp\!\left(-\frac{\delta^2 \mu}{2+\delta}\right), \label{eq:chernoff-upper}
\end{equation}
and for any $\delta \in (0, 1)$:
\begin{equation}
    \Pr[X \leq (1-\delta)\mu] \leq \exp\!\left(-\frac{\delta^2 \mu}{2}\right). \label{eq:chernoff-lower}
\end{equation}
In particular, for any $a \geq 2e\mu$:
\[
    \Pr[X \geq a] \leq 2^{-a}.
\]
Note that when $\delta\geq 1$, we can upper bound \Cref{eq:chernoff-upper} with $\Pr[X\geq (1+\delta)\mu]\leq \exp\left(-\frac{\delta\mu}{3}\right)$. 
\end{fact}

For an introduction to negative association, its uses, and the extended proof of Chernoff Bounds in this case, see \cite{joag1983negative,dubhashi1996balls,neg-assoc-notes}. 
We only require use of negative association in the following fact, the proof of which can be found in Section 3.1 of \cite{joag1983negative}.

\begin{fact}[Hypergeometric Distributions are Negatively Associated]\label{lem:chernoff-hypergeometric}
    Let $N$ be a population of size $n$ containing $K$ marked elements. 
    Let $X$ be the number of marked elements in a sample of size $t$ drawn \emph{without replacement} (i.e., $X \sim \mathrm{Hypergeometric}(n, K, t)$). 
    Then $\E[X] = \frac{Kt}{n}$ and $X$ can be written as the sum of $t$ negatively associated variables $\sim\{0,1\}$, each equaling $1$ with probability $\frac{K}{n}$. Thus, \Cref{lem:chernoff} above (the Chernoff Bound) holds for $X$.
\end{fact}

\begin{fact}\label{fact: samplingworep}
For non-negative integers $n,a, b$ such that $0\leq a\leq n$ and $0\leq b \leq n-a$, then 
\[ \left(1-\frac{a}{n-b+1}\right)^b\leq\frac{\binom{n-a}{b}}{\binom{n}{b}}\leq \left(1-\frac{a}{n}\right)^b\leq e^{-\frac{ab}{n}} . \]
\end{fact}

\section{Bounding the Optimal for Collection of Cliques}
\label{sec:cliques.optimal}
An important ingredient of our analysis is an asymptotically tight characterization of an optimal $k$-clustering when the graph instance is a collection of disjoint cliques.  In such an instance $\instance$, the input graph consists of $\ell$ disjoint cliques $K = \{K_1, K_2, \dots, K_\ell\}$, where clique $K_i$ contains $n_i$ vertices and $\sum_{i=1}^\ell n_i = n$.  Let the cliques be indexed in decreasing order of size: \( n_1 \geq n_2 \geq \cdots \geq n_\ell \).  We first show that an optimal clustering keeps each clique intact.

\begin{lemma}
\label{lem:cliques-intact}
For any $\ell$ disjoint cliques instance $\mathcal{I}$ of correlation $k$-clustering, there exists an optimal solution, there exists an optimal solution in which each clique $K_i$ is entirely contained within a single cluster.
\end{lemma}
\begin{proof}
Let $G$ be our input graph consisting of $\ell$ disjoint cliques, and let $C = \{C_1, C_2, \ldots, C_k\}$ be an optimal $k$-clustering. 
Suppose for contradiction that there exists a clique $K_i$ that is split across multiple clusters, including clusters $C_1$ and $C_2$. 
That is, $K_i$ contains disjoint sets $S_1 = C_1 \cap K_i$ and $S_2 = C_2 \cap K_i$ with $S_1\subseteq C_1$ and $S_2\subseteq C_2$.
Let $|S_1| = s_1$ and $|S_2| = s_2$, and note that both $s_1,s_2>0$ and $s_1 + s_2 \leq |K_i|$.

Consider the cost suffered by $C$ due to vertices in $S_1$ and $S_2$.
Since $K_i$ is a clique, there are $s_1 s_2$ edges which cross from $C_1$ to $C_2$, which contributes $s_1 s_2$ to the disagreement cost.
Additionally, $S_1$ and $S_2$ will disagree with any other vertices in their respective clusters.
Let $X_1$ be the vertices in $C_1 \setminus S_1$ and $X_2$ be the vertices in $C_2 \setminus S_2$ with $|X_1| = x_1$ and $|X_2| = x_2$.
Then placing $S_1$ into $C_1$ and $S_2$ into $C_2$ also contributes $(s_1 x_1 + s_2 x_2)$ to the disagreement cost.

Without loss of generality, let $x_1\leq x_2$.
Now consider the alternative clustering $C'$ where we move all vertices in $S_2$ to $C_1$.
$C'$ does not split $S_1$ and $S_2$, so the changed cost contributions suffered by $C'$ due to vertices in $S_1$ and $S_2$ is from disagreement between $S_1$ and $S_2$ and the rest of $C_1$.
This contributes $(s_1 + s_2) x_1$ to the total disagreement cost.  The decrease in cost when switching from $C$ to $C'$ is exactly 

\[
s_1 s_2 + s_1 x_1 + s_2 x_2 - (s_1 + s_2) x_1 = s_1 s_2 + s_2  x_2 - s_2  x_1 = s_1  s_2 + s_2  (x_2 - x_1) > 0,
\]
where the last step holds since $x_1\leq x_2$. 

By repeatedly applying this argument, we can transform any optimal clustering into one where no cliques are split, without increasing the cost.
Thus, there exists an optimal solution where every clique $K_i$ is entirely contained within a single cluster.
\end{proof}

We now define three important parameters of any given disjoint cliques instance $\instance$, which help lower bound the optimal cost.  
\begin{definition}\label{def:s-of-instance}
    Define $\sml{\instance}$ as the total number of vertices in the graph that do not lie in the $k$ largest cliques: $\sml{\instance} = \sum_{i=k+1}^{\ell} n_i$. (See Figure \ref{fig:instance}.)
\end{definition}
To motivate the next two parameters, we make some observations about an optimal $k$-clustering.  With $n$ vertices and $k$ clusters, the average number of vertices assigned to each cluster is $\frac{n}{k}$. 
If the largest clique $K_1$ is of size $n_1\geq \frac{n}{k}$, then any optimal solution must place $K_1$ into its own cluster; otherwise, we can move any clique placed in the same cluster as $K_1$ to the cluster with the least total size and decrease the cost.   Thus, if $n_1 \ge \frac{n}{k}$, then the optimal cost on $\instance$ is equal to the cost of an optimal $(k-1)$-clustering on a smaller instance $\{ K_2,\hdots, K_\ell \}$ with $(n-n_1)$ vertices.  We define $r(\instance)$ to be the index at which this process of pulling out large cliques that must be clustered separately stops.
\begin{definition} \label{def:threshold}
For instance $\instance$ with greater than $k$ cliques, let \( r(\instance) \in \{0, \ldots, k-1\}\) be the smallest index such that
\[
n_{r(\instance)+1} < \frac{n - \sum_{i=1}^{r(\instance)} n_i}{k - r(\instance)} =: \avg{\instance} .
\]
Essentially, $r(\instance)$ is the number of cliques of $\instance$ that are so large that optimally they should be placed into their own cluster, regardless of what the rest of the instance looks like.  If $\instance$ has at most $k$ cliques, then we set $\avg{\instance}$ to 0.

The value $\avg{\instance}$ represents the average size required to evenly distribute the remaining vertices (those not in the largest \( r(\instance) \) cliques) into the remaining \( k - r(\instance) \) clusters. (See Figure \ref{fig:instance}.)
Note that since these cliques are not large enough to obviously be placed into their own cluster, in the optimal solution they may be placed into a cluster with other cliques depending on the rest of the instance.
\end{definition}

\begin{lemma}\label{lem:opt-theta}
The cost of the optimal correlation $k$-clustering on an $\ell$ disjoint cliques instance $\instance$ satisfies $\cost(\opt) \geq \frac{\sml{\mathcal{I}} \cdot \avg{\mathcal{I}}}{2}$.
Moreover, there exists a polynomial-time greedy algorithm for computing a $2$-approximate solution and a polynomial time approximation scheme for the problem.
\end{lemma}
\begin{proof}
If $\ell \le k$, assign every clique to a separate cluster. Then the cost and $s(\instance)a(\instance)$ are both zero. Henceforth, we assume $\ell > k$, so $r(\instance)$ and $k - r(\instance) > 0$ are defined.

By Lemma \ref{lem:cliques-intact}, in any optimal clustering, each clique remains intact. The disagreement cost comes exclusively from co-clustering vertices from different cliques.
Let $C = \{C_1, C_2, \ldots, C_k\}$ be an optimal clustering, and let $c_j = |C_j|$ be the size of cluster $j$. 

For simplicity, let $r=r(\instance)$. By definition of $r$, cliques $K_1,...,K_r$ are large enough that any optimal clustering will give them their own cluster. 
Without loss of generality, assume $C_j = K_j$ for all such $j\in\{1,\hdots,r\}$.  

Now, consider a cluster $C_j$ for $j\in \{r+1,\ldots,k\}$, and let $K_{j1},...,K_{jq}$ (with sizes $n_{j1},\ldots,n_{jq}$ respectively) denote the cliques assigned to it in an optimal solution. The total disagreement cost incurred for this cluster $C_j$ is given by, 
\begin{align*}
    \binom{c_j}{2}-\sum_{i=1}^q\binom{n_{ji}}{2} =\frac{1}{2}\left(c_j(c_j-1) - \sum_{i=1}^q n_{ji}(n_{ji}-1)\right) = \frac{1}{2}\left(c_j^2-\sum_{i=1}^q n_{ji}^2\right) 
\end{align*}
where the final equality holds since $c_j=\sum_{i=1}^q n_{ji}$. Summing the disagreement costs across all clusters $C_{r+1},,,.,C_k$, yields that, 
\begin{align*}
    \cost(\opt) = \frac{1}{2}\sum_{j=r+1}^k \left(c_j^2-\sum_{i=1}^q n_{ji}^2 \right) = \frac{1}{2}\left(\sum_{j=r+1}^k c_j^2 - \sum_{j=r+1}^{\ell} n_j^2\right)
\end{align*}

Noting that $\sum_{j=r+1}^k c_j = \sum_{j=r+1}^{\ell} n_j = (k-r)a(\instance)$, we have by Cauchy Schwartz that: $(\sum_{j=r+1}^{k}c_j)^2 \leq (k-r)\sum_{j=r+1}^k c_j^2$, and so $\sum_{j=r+1}^k c_j^2 \geq \frac{a(\instance)^2(k-r)^2}{k-r}=a(\instance)^2(k-r)$. Hence, 

\begin{align*}
    \cost(\opt)\geq \frac{1}{2}\left(a(\instance)^2(k-r)- \sum_{j=r+1}^{\ell} n_j^2\right)
\end{align*}

Next, we establish an upper bound on $\sum_{j=r+1}^{\ell}n_j^2$ and show that
$\sum_{j=r+1}^{\ell}n_j^2\leq a(\instance)\sum_{j=r+1}^k n_j$. Recall that $n_1\geq n_2 \geq ...\geq n_{\ell}$. To this end, we first note that $\sum_{j=r+1}^{\ell}\frac{n_j}{a(\instance)}=k-r$ by definition. Let $w_j=\frac{n_j}{a(\instance)}$. Since $n_j<a(\instance)$ for all $j>r$ it follows that $0\leq w_j<1$. 

Since clique sizes are nonincreasing, the weighted sum $\sum_{j=r+1}^{\ell}w_jn_j$ is maximized by assigning $w_j=1$ to the $k-r$ largest remaining cliques, and hence $\sum_{j=r+1}^{\ell}w_jn_j \leq \sum_{j=r+1}^k n_j$. By substituting the definition of $w_j$ we obtain that, 
\begin{align*}
    \sum_{j=r+1}^{\ell} n_j^2 &\leq a(\instance)\sum_{j=r+1}^k n_j  \\
    &=a(\instance)(\sum_{j=r+1}^{\ell} n_j - \sum_{j=k+1}^{\ell} n_j) \\
    &=a(\instance)(a(\instance)(k-r)-s(\instance)).
\end{align*}
Substituting this into the cost for OPT, we obtain,
\begin{align*}
    \cost(\opt)\geq \frac{1}{2}\left(a(\instance)^2(k-r)- \sum_{j=r+1}^{\ell} n_j^2\right) \geq \frac{1}{2}(a(\instance)^2(k-r)-a(\instance)(a(\instance)(k-r)-s(\instance)))=\frac{1}{2}a(\instance)s(\instance).
\end{align*}

To establish the 2-approximation, we present a greedy algorithm that computes a $k$-clustering $C$ with $\text{\cost}(C) \le a(\mathcal{I})s(\mathcal{I}) \le 2\text{Cost}(\text{OPT})$. Consider the following clustering strategy. First, assign each of the largest $k$ cliques to its own cluster, which incurs zero cost for this decision. The first $r$ clusters $C_1,C_2,...,C_r$ are made unavailable. Next, for $j\in \{k+1,...,\ell\}$ we assign $K_j$ to a currently least-loaded cluster among $C_{r+1},...,C_k$.

We attribute the total disagreement cost incurred by this algorithm in the following manner. When a clique $K_j$ is assigned to a cluster $C_i\in \{C_{r+1},...,C_k\}$, we account for an amount $n_j\cdot c_{i'}$ where $c_i'$ denotes the number of vertices assigned to $C_i$ immediately before the assignment of clique $K_j$. This measures the total disagreement cost incurred between vertices in $K_j$ and vertices in all other cliques assigned to $C_i$, prior to the assignment of $K_j$. It is easy to observe that the sum over all such disagreement costs yields the total cost for the algorithm. Moreover, for every clique $K_j$ and the cluster $C_i$ to which it is assigned, $c_i'<a(\instance)$ since $C_i$ is the least loaded cluster among $C_{r+1},...,C_k$ at that instant. Thus, the total cost of this clustering is given by, 

\[ \cost(C) \leq \sum_{i=k+1}^\ell n_i \avg{\instance} = \avg{\instance} \sml{\instance} \leq 2\cost(\opt) . \] 

The PTAS for the problem closely follows the approach of Alon, Azar, Woeginger, and Yadid \cite{alon1998approximation} which gives a PTAS for the load balancing problem in the $\ell_2$ norm.  We refer to \Cref{app:PTAS} for details.

\end{proof}

\section{Constant-Approximation for Offline Correlation $k$-Clustering}\label{sec:offline}

\begin{lemma}\label{thm:cc-to-kcc}
Given an $\alpha$-approximation correlation clustering solution $C$ for a general graph instance $\instance$, there exists a $(1 + 2\alpha)$-approximation solution to the correlation $k$-clustering problem which does not split any clusters of $C$.
\end{lemma}

\begin{proof}
Let $\instance$ be any instance of the correlation $k$-clustering problem, and let $C$ be an $\alpha$-approximate clustering solution to the unconstrained correlation clustering problem on $\instance$. 
We will call the clusters of $C$ ``groups'' in order to not confuse them with the ``clusters'' associated with correlation $k$-clustering. 
Also, let $\opt_{CC}$ denote the optimal unconstrained correlation clustering solution on $\instance$, and let $\opt_k$ denote the optimal solution for correlation $k$-clustering on $\instance$. 

For a clustering solution $C'$ to an instance $\mathcal{I}$, we denote by $\cost(C', \mathcal{I})$ the disagreement cost of clustering $C'$ on instance $\mathcal{I}$, which can be expressed as a sum of: \begin{enumerate}
    \item $\cost_A(C', \instance)$: the cost associated with non-edges between vertices clustered together.
    \item $\cost_D(C', \mathcal{I})$: the cost associated with edges between vertices not clustered together.
\end{enumerate}

First we build a new instance $\instance'$ from $\instance$ by turning each group of $C$ into a clique, as illustrated in \Cref{fig:cc-to-kcc} by deleting all inter-group edges (i.e., edges between different clusters of $C$), and adding all possible intra-group edges (edges inside each cluster of $C$).
This transforms our original instance $\instance$ into an instance $\instance'$ consisting of disjoint cliques, where each clique corresponds exactly to a cluster in $C$.

Since exactly $\cost_A(C, \instance)$ edges are added and $\cost_D(C, \instance)$ edges are deleted when transforming $\instance$ to $\instance'$, the cost of $\opt_k$ on the modified instance $\instance'$ can be related to its cost on the original instance $\instance$ as follows.
\begin{align*}
\cost(\opt_k, \instance') & \leq \cost(\opt_k, \instance) + \cost_A(C, \instance) + \cost_D(C, \instance) \\
& = \cost(\opt_k, \instance) + \cost(C, \instance).
\end{align*}
 The inequality follows because $\instance$ and $\instance'$ differ on exactly \(\operatorname{Cost}(C,\instance)\) vertex-pair labels, and changing one pair label changes the cost of any fixed clustering by at most one.

Note that $\instance'$ is a collection of cliques, corresponding to the groups in $C$. 
From \Cref{lem:cliques-intact}, we know that there exists an optimal solution $C_{\instance'}$ that does not split any cliques. 
Thus,
\begin{equation*}
\cost(C_{\instance'}, \instance') \leq \cost(\opt_k, \instance') .
\end{equation*}

When applying $C_{\instance'}$ to the original instance $\instance$, we incur additional costs due to the transformation from $\instance$ to $\instance'$.  We derive
\begin{align*}
    \cost(C_{\instance'}, \instance) & \leq \cost(C_{\instance'}, \instance') + \cost(C, \instance) \\
    & \leq \cost(\opt_k, \instance') + \cost(C, \instance)\\
    & \leq \cost(\opt_k, \instance) + 2\cost(C, \instance) \\
    & \leq \cost(\opt_k,\instance) + 2\alpha \cost(\opt_{CC},\instance) \\
    & \leq \cost(\opt_k,\instance) + 2\alpha \cost(\opt_k,\instance) \\
    &= (1+2\alpha) \cost(\opt_k,\instance) .
\end{align*}
In the above derivations, we use the facts that $\cost(C, \instance) \leq \alpha \cdot \cost(\opt_{CC})$ (since $C$ is an $\alpha$-approximation for unconstrained correlation clustering) and $\cost(\opt_{CC}) \leq \cost(\opt_k, \instance)$.

\end{proof}

\offline*

\begin{proof}
    Given an instance $\instance$, begin by computing $C = \mathcal{A}(\instance)$, which is an $\alpha$-approximate solution to unconstrained correlation clustering.  Next compute $\instance'$ as in the proof of \Cref{thm:cc-to-kcc}. 
    This is the instance in which all groups (or, clusters) of $C$ have been turned into disjoint cliques.
    Since $\instance'$ is an $\ell$-cliques instance, use the PTAS from \Cref{lem:opt-theta} to compute a $(1 + \eps)$-approximate correlation $k$-clustering $C'$ on $\instance'$.
    Recall that $C_{\instance'}$ denotes the optimal $k$-clustering on $\instance'$, and $\opt_k$ denotes the optimal $k$-clustering on $\instance$.
    Then
    \begin{align*}
        \cost(C',\instance) & \leq \cost(C',\instance') + \cost(C,\instance) \\
        & \leq (1 + \eps)\cost(C_{\instance'}, \instance') + \cost(C,\instance) \\
        & \leq (1 + \eps)\cost(\opt_k,\instance) + (2 + \eps)\cost(C,\instance) \\
        & \leq (1+2\alpha + \bigo(\eps)) \cost(\opt_k,\instance) .
    \end{align*}
    
\end{proof}
\section{ Lower Bound for Online Correlation $k$-Clustering}
\label{sec:lower-bound}
In this section, we present a lower bound of $\Omega(\log k)$ for online correlation $k$-clustering in the random arrival model, proving Theorem \ref{thm:onlinelowerbound}. We accomplish this by constructing an instance $\mathcal{I}$ on $n = 2(k-1)\ln k + 2$ vertices which consists of $k-1$ cliques, each of size $2\ln k$ and two singletons (see \Cref{fig:hard-instance-n}).  We show that any algorithm must incur an expected cost of at least $\Omega( \log k)$, while the optimal algorithm incurs a unit cost on this instance. In particular, we show that any algorithm incurs this cost after $ck\ln k$ vertices have arrived.  (For simplicity, we omit floor/ceiling rounding and assume that $2(k-1)\ln k + 2$, $2\ln k$, and $ck\ln k$ are integers, as it does not affect the asymptotic bounds established in this section.)
\begin{figure}[htbp]
\centering
\begin{subfigure}{0.45\textwidth}
    \centering
            \tikzset{every picture/.style={line width=0.75pt}}    
        \begin{tikzpicture}[x=0.75pt,y=0.75pt,yscale=-1,xscale=1] 
\draw  [fill=gray!60  ,fill opacity=1 ] (143,99) .. controls (143,95.69) and (145.69,93) .. (149,93) -- (167,93) .. controls (170.31,93) and (173,95.69) .. (173,99) -- (173,205) .. controls (173,208.31) and (170.31,211) .. (167,211) -- (149,211) .. controls (145.69,211) and (143,208.31) .. (143,205) -- cycle ;
\draw  [fill=gray!60  ,fill opacity=1 ] (105,99) .. controls (105,95.69) and (107.69,93) .. (111,93) -- (129,93) .. controls (132.31,93) and (135,95.69) .. (135,99) -- (135,205) .. controls (135,208.31) and (132.31,211) .. (129,211) -- (111,211) .. controls (107.69,211) and (105,208.31) .. (105,205) -- cycle ;
\draw  [fill=gray!60  ,fill opacity=1 ] (254,99) .. controls (254,95.69) and (256.69,93) .. (260,93) -- (278,93) .. controls (281.31,93) and (284,95.69) .. (284,99) -- (284,206) .. controls (284,209.31) and (281.31,212) .. (278,212) -- (260,212) .. controls (256.69,212) and (254,209.31) .. (254,206) -- cycle ;
\draw  [fill=gray!60  ,fill opacity=1 ] (295,204.5) .. controls (295,200.36) and (298.36,197) .. (302.5,197) .. controls (306.64,197) and (310,200.36) .. (310,204.5) .. controls (310,208.64) and (306.64,212) .. (302.5,212) .. controls (298.36,212) and (295,208.64) .. (295,204.5) -- cycle ;
\draw  [fill=gray!60  ,fill opacity=1 ] (214,99) .. controls (214,95.69) and (216.69,93) .. (220,93) -- (238,93) .. controls (241.31,93) and (244,95.69) .. (244,99) -- (244,206) .. controls (244,209.31) and (241.31,212) .. (238,212) -- (220,212) .. controls (216.69,212) and (214,209.31) .. (214,206) -- cycle ;
\draw  [fill=gray!60  ,fill opacity=1 ] (316,204.5) .. controls (316,200.36) and (319.36,197) .. (323.5,197) .. controls (327.64,197) and (331,200.36) .. (331,204.5) .. controls (331,208.64) and (327.64,212) .. (323.5,212) .. controls (319.36,212) and (316,208.64) .. (316,204.5) -- cycle ;
\draw   (108,228) .. controls (108,232.67) and (110.33,235) .. (115,235) -- (185,235) .. controls (191.67,235) and (195,237.33) .. (195,242) .. controls (195,237.33) and (198.33,235) .. (205,235)(202,235) -- (275,235) .. controls (279.67,235) and (282,232.67) .. (282,228) ;
\draw   (90,94) .. controls (85.33,94) and (83,96.33) .. (83,101) -- (83,143) .. controls (83,149.67) and (80.67,153) .. (76,153) .. controls (80.67,153) and (83,156.33) .. (83,163)(83,160) -- (83,205) .. controls (83,209.67) and (85.33,212) .. (90,212) ;

\draw (30,141) node [anchor=north west][inner sep=0.75pt]   [align=left] {$2 \ln k$};
\draw (171,249) node [anchor=north west][inner sep=0.75pt]   [align=left] {$k-1$};
\draw (184,165) node [anchor=north west][inner sep=0.75pt]   [align=left] {$\dots$};

\end{tikzpicture}
\vspace{2pt}
    \caption{ The instance representing cliques of size $2 \ln k$ and $2$ singletons.} 
    \label{fig:hard-instance-n}
    \end{subfigure}
    \hfill
    \begin{subfigure}{0.45\textwidth}
            \centering
            \tikzset{every picture/.style={line width=0.75pt}}       
            \begin{tikzpicture}[x=0.75pt,y=0.75pt,yscale=-1,xscale=1]

\draw  [fill=gray!60  ,fill opacity=1 ] (216,116) .. controls (216,112.69) and (218.69,110) .. (222,110) -- (240,110) .. controls (243.31,110) and (246,112.69) .. (246,116) -- (246,222) .. controls (246,225.31) and (243.31,228) .. (240,228) -- (222,228) .. controls (218.69,228) and (216,225.31) .. (216,222) -- cycle ;
\draw  [fill=gray!60  ,fill opacity=1 ] (178,116) .. controls (178,112.69) and (180.69,110) .. (184,110) -- (202,110) .. controls (205.31,110) and (208,112.69) .. (208,116) -- (208,222) .. controls (208,225.31) and (205.31,228) .. (202,228) -- (184,228) .. controls (180.69,228) and (178,225.31) .. (178,222) -- cycle ;
\draw  [fill=gray!60  ,fill opacity=1 ] (326,115) .. controls (326,111.69) and (328.69,109) .. (332,109) -- (350,109) .. controls (353.31,109) and (356,111.69) .. (356,115) -- (356,222) .. controls (356,225.31) and (353.31,228) .. (350,228) -- (332,228) .. controls (328.69,228) and (326,225.31) .. (326,222) -- cycle ;
\draw  [fill=gray!60  ,fill opacity=1 ] (367,220.5) .. controls (367,216.36) and (370.36,213) .. (374.5,213) .. controls (378.64,213) and (382,216.36) .. (382,220.5) .. controls (382,224.64) and (378.64,228) .. (374.5,228) .. controls (370.36,228) and (367,224.64) .. (367,220.5) -- cycle ;
\draw  [fill=gray!60  ,fill opacity=1 ] (288,115) .. controls (288,111.69) and (290.69,109) .. (294,109) -- (312,109) .. controls (315.31,109) and (318,111.69) .. (318,115) -- (318,222) .. controls (318,225.31) and (315.31,228) .. (312,228) -- (294,228) .. controls (290.69,228) and (288,225.31) .. (288,222) -- cycle ;
\draw  [fill=gray!60  ,fill opacity=1 ] (367,197.5) .. controls (367,193.36) and (370.36,190) .. (374.5,190) .. controls (378.64,190) and (382,193.36) .. (382,197.5) .. controls (382,201.64) and (378.64,205) .. (374.5,205) .. controls (370.36,205) and (367,201.64) .. (367,197.5) -- cycle ;

\draw (253,186) node [anchor=north west][inner sep=0.75pt]   [align=left] {$\dots$};

\end{tikzpicture}
\vspace{30pt}
    \caption{\vspace{20pt}An optimal clustering for the instance.}
    \label{fig:hard-instance-opt-clustering-n}
    \end{subfigure}
            \caption{}
    \end{figure}

One possible clustering is obtained by assigning all $k-1$ cliques to unique clusters, and assigning singleton vertices to the final cluster, which incurs a unit cost (see \Cref{fig:hard-instance-opt-clustering-n}). 
To exhibit a lower bound, we prove that after $ck\ln k$ vertices have arrived for some constant $c$, then with at least constant probability: i) no vertices have arrived from at least one clique, ii) both singleton vertices have arrived. 
As a result, the first vertex to arrive from the clique from which no vertices have arrived by time $c k\ln k$, is indistinguishable from the singleton vertices, and hence with constant probability, the algorithm must with constant probability assign a clique and a singleton vertex to the same cluster, incurring a $\Omega(\log k)$ cost (see Figure \ref{fig:both'}). In the following lemmas, we bound the probability of the aforementioned events.

\begin{lemma}[Both singletons arrive]
\label{lem:singletons_arrive}
With probability at least $\frac{c^2}{24}$ where $c<1$ is a constant, both singleton vertices arrive by time $ck\ln k$, where $k\geq \max\{2/c,3\}$.
\end{lemma}
\begin{proof}
The probability that both singleton vertices arrive by time $ck \ln k$ is, 
\begin{align*}
     \left(\frac{ck \ln k}{n}\right) \cdot \left(\frac{ck \ln k -1}{n - 1}\right) \geq \left(\frac{ck \ln k}{n}\right) \cdot \left(\frac{ck \ln k -1}{n}\right)
     = \left(\frac{ck \ln k}{n}\right)^2 - \frac{ck \ln k}{n^2} .
\end{align*}
Note that when $k\geq \frac{2}{c}$, $ck\ln k \geq 2$, yielding, 
\begin{align*}
    \left(\frac{ck \ln k}{n}\right)^2 - \frac{ck \ln k}{n^2} & \geq \frac{1}{2} \left(\frac{ck \ln k}{n}\right)^2 = \frac{1}{2} \left(\frac{ck \ln k}{2(k-1)\ln k + 2}\right)^2 \geq \frac{1}{8}\left(\frac{ck \ln k}{k \ln k + 1}\right)^2 \\
    &= \frac{c^2}{8}\left(1 - \frac{1}{k \ln k+1}\right)^2 = \frac{c^2}{8}\left(1-\frac{2}{k\ln k+1}+\frac{1}{(k\ln k+1)^2}\right)\\
    &\geq \frac{c^2}{8} - \frac{c^2}{4(k\ln k+1)} \geq \frac{c^2}{8}-\frac{c^2}{12} = \frac{c^2}{24}
\end{align*}

\end{proof}

\begin{lemma}\label{lem:vars_are_na}
    For clique $K_j$, let $X_j$ denote the indicator variable for if clique $K_j$ did not arrive by time $ck\ln k$. 
    That is, $X_j=1$ if clique $K_j$ has no vertex arrive by time $ck\ln k$, and $X_j=0$ otherwise.
    Then the variables $X_1,\hdots,X_{k-1}$ are negatively associated.
\end{lemma}

\begin{proof}
To show this, we will use the following two properties of negatively associated distributions.

\begin{lemma}[Permutation Distributions are Negatively Associated \cite{joag1983negative}]\label{lem:perm_distr_na}
    Let $x_1\leq\hdots\leq x_n$ be $n$ values and let $Y_1,\hdots,Y_n$ be random variables drawn uniformly without replacement from the set $\{x_1\leq\hdots\leq x_n\}$. 
    That is, $Y_1,\hdots,Y_n$ is a uniformly random permutation of the value set $\{x_1\leq\hdots\leq x_n\}$.
    Then $Y_1,\hdots,Y_n$ are negatively associated.
\end{lemma}

\begin{lemma}[Closure of Negatively Associated Random Variables \cite{joag1983negative}]\label{lem:na_closure}
    Suppose the random variables $Y_1,\hdots,Y_n$ are negatively associated. 
    Also suppose that $f_1,\hdots,f_k:\mathbb{R}^n\rightarrow\mathbb{R}$ are all monotonically increasing (or all monotonically decreasing) functions, with each function $f_i$ depending on a disjoint subset of indices from $[n]$. 
    Then the random variables $f_1(\vec{Y}),\hdots,f_k(\vec{Y})$ are also negatively associated.
\end{lemma}

Showing that the variables $X_1,\hdots,X_{k-1}$ are negatively associated, where recall $X_j$ denotes the indicator variable for if clique $K_j$ did not arrive by time $ck\ln k$, is a simple application of the above two properties.

First, let $x_1,\hdots,x_n$ be the set of values with exactly $(n-ck\ln k)$ copies of $0$ and $ck\ln k$ copies of $1$. 
Let $Y_1,\hdots,Y_n$ be the permutation distribution on the value set $x_1,\hdots,x_n$. 
Then by \Cref{lem:perm_distr_na}, $Y_1,\hdots,Y_n$ is negatively associated.
Note that $Y_u$ can be interpreted as the indicator variable for if vertex $u$ arrived before time $ck\ln k$ or not. 
$Y_u=1$ if vertex $u$ arrived before time $ck\ln k$, and $Y_u=0$ if vertex $u$ did not arrive before time $ck\ln k$.

Next, let $f_1,\hdots,f_{k-1}:\{0,1\}^n\rightarrow\{0,1\}$ be defined by
\[ 
    f_j(\vec{Y}) = \begin{cases} 
      0 & \mbox{if some vertex $u$ in } K_j \mbox{ has } Y_u=1 \\
      1 & \mbox{otherwise} 
   \end{cases} ~ .
\]
Each function $f_j$ is monotonically decreasing. 
Additionally, each function $f_1,\hdots,f_{k-1}$ depends on disjoint subsets of indices from $[n]$. 
Finally, note that for each $j$, $f_j(\vec{Y}) = X_j$. 
Therefore by \Cref{lem:na_closure}, the variables $X_1,\hdots,X_{k-1}$ are negatively associated.

\end{proof}

\begin{lemma}[At least one large clique missing]
\label{lem:clique_missing}
With probability at least $1 - e^{-\frac{\sqrt{k}}{16}}$, there exists at least one large clique from which no vertices have arrived by time $ck\ln k$ where $k\geq 16$ and $0<c<2/5$. 
\end{lemma}

\begin{proof}
For each clique $K_j$, let $X_j$ denote the indicator random variable which is 1 if no vertex from $K_j$ arrives by time $ck\ln k$ and 0 otherwise. Note that $n=2(k-1)\ln k +2$.
By \Cref{lem:vars_are_na}, $X_1,\hdots,X_{k-1}$ are negatively associated. Therefore, we can use Chernoff bounds (\cref{lem:chernoff}).

We can bound the expectation of each $X_j$,
\begin{align*}
\mathbb{E}[X_j] 
= \frac{\binom{n - 2\ln k}{ck\ln k}}{\binom{n}{ck\ln k}} &\geq \left(1-\frac{2\ln k}{n-ck\ln k+1}\right)^{ck\ln k} = \left(1-\frac{2\ln k}{((2-c)k-2)\ln k+3}\right)^{ck\ln k}\\
&\geq \left(1-\frac{2\ln k}{((2-c)k-2)\ln k}\right)^{ck\ln k} = \left(1-\frac{2}{(2-c)k-2}\right)^{ck\ln k}
\end{align*}
where the first inequality follows from \Cref{fact: samplingworep}. 
Noting that for $0\leq x <1$, $1-x\geqq e^{-\frac{x}{1-x}}$ and setting $x=\frac{2}{(2-c)k-2}$ which is upper bounded by $1$ whenever $c<2$, we have that $\frac{x}{1-x}= \frac{2}{(2-c)k-4}$. 
Thus, 
\[
    \mathbb{E}[X_j] \geq e^{-\frac{2ck\ln k}{(2-c)k-4}}= k^{\frac{-2ck}{(2-c)k-4}} = k^{\frac{-2c}{2-c-4/k}}.
\]
Now let $X=\sum_{j=1}^{k-1} X_j$ denote the random variable corresponding to the number of cliques from which no vertex arrives by time $ck\ln k$. 
By linearity of expectation, 
\begin{equation}\label{eq:ex_num_clique_not_arrive} 
\E[X] = \sum_{j=1}^{k-1} \mathbb{E}[X_j] \geq (k-1) \E[X_j]\geq (k-1)k^{\frac{-2c}{2-c-4/k}} \geq \frac{1}{2}\sqrt{k} .
\end{equation}
whenever $k\geq 8$ and $c<2/5$. 
We now apply Chernoff bounds to achieve our desired result.

Assuming $k\geq 16$, then $\frac{1}{2}\sqrt{k}\geq 2$, and we can set $\delta = \frac{1}{2}$.
\[
    \Pr[ X < 1 ] \leq \Pr\left[ X < \frac{1}{2} \E[X] \right] \leq e^{-\frac{\sqrt{k}}{16}} .
\]

\end{proof}
\begin{figure}[htbp]
\centering
\begin{subfigure}{0.45\textwidth}
    \centering

\tikzset{every picture/.style={line width=0.75pt}}         

\begin{tikzpicture}[x=0.7pt,y=0.7pt,yscale=-1,xscale=1]
    \draw  [fill=gray!60  ,fill opacity=1 ] (163,157) .. controls (163,153.69) and (165.69,151) .. (169,151) -- (187,151) .. controls (190.31,151) and (193,153.69) .. (193,157) -- (193,263) .. controls (193,266.31) and (190.31,269) .. (187,269) -- (169,269) .. controls (165.69,269) and (163,266.31) .. (163,263) -- cycle ;
    \draw  [fill=gray!60  ,fill opacity=1 ] (125,157) .. controls (125,153.69) and (127.69,151) .. (131,151) -- (149,151) .. controls (152.31,151) and (155,153.69) .. (155,157) -- (155,263) .. controls (155,266.31) and (152.31,269) .. (149,269) -- (131,269) .. controls (127.69,269) and (125,266.31) .. (125,263) -- cycle ;
    \draw  [fill=gray!60  ,fill opacity=1 ] (273,156) .. controls (273,152.69) and (275.69,150) .. (279,150) -- (297,150) .. controls (300.31,150) and (303,152.69) .. (303,156) -- (303,263) .. controls (303,266.31) and (300.31,269) .. (297,269) -- (279,269) .. controls (275.69,269) and (273,266.31) .. (273,263) -- cycle ;
    \draw  [fill=gray!60  ,fill opacity=1 ] (314,261.5) .. controls (314,257.36) and (317.36,254) .. (321.5,254) .. controls (325.64,254) and (329,257.36) .. (329,261.5) .. controls (329,265.64) and (325.64,269) .. (321.5,269) .. controls (317.36,269) and (314,265.64) .. (314,261.5) -- cycle ;
    \draw  [fill=gray!60  ,fill opacity=1 ] (235,156) .. controls (235,152.69) and (237.69,150) .. (241,150) -- (259,150) .. controls (262.31,150) and (265,152.69) .. (265,156) -- (265,263) .. controls (265,266.31) and (262.31,269) .. (259,269) -- (241,269) .. controls (237.69,269) and (235,266.31) .. (235,263) -- cycle ;
    \draw  [fill=gray!60  ,fill opacity=1 ] (132,137.5) .. controls (132,133.36) and (135.36,130) .. (139.5,130) .. controls (143.64,130) and (147,133.36) .. (147,137.5) .. controls (147,141.64) and (143.64,145) .. (139.5,145) .. controls (135.36,145) and (132,141.64) .. (132,137.5) -- cycle ;
    
    \draw (203,225) node [anchor=north west][inner sep=0.75pt]   [align=left] {$\dots$};

\end{tikzpicture}
    \caption{}
    \label{fig:sub1}
\end{subfigure}
\hfill
\begin{subfigure}{0.45\textwidth}
    \centering

\tikzset{every picture/.style={line width=0.75pt}}       

\begin{tikzpicture}[x=0.7pt,y=0.7pt,yscale=-1,xscale=1]

\draw  [fill=gray!60  ,fill opacity=1 ] (183,176) .. controls (183,172.69) and (185.69,170) .. (189,170) -- (207,170) .. controls (210.31,170) and (213,172.69) .. (213,176) -- (213,283) .. controls (213,286.31) and (210.31,289) .. (207,289) -- (189,289) .. controls (185.69,289) and (183,286.31) .. (183,283) -- cycle ;
\draw  [fill=gray!60  ,fill opacity=1 ] (145,174) .. controls (145,170.69) and (147.69,168) .. (151,168) -- (169,168) .. controls (172.31,168) and (175,170.69) .. (175,174) -- (175,283) .. controls (175,286.31) and (172.31,289) .. (169,289) -- (151,289) .. controls (147.69,289) and (145,286.31) .. (145,283) -- cycle ;
\draw  [fill=gray!60  ,fill opacity=1 ] (297,281.5) .. controls (297,277.36) and (300.36,274) .. (304.5,274) .. controls (308.64,274) and (312,277.36) .. (312,281.5) .. controls (312,285.64) and (308.64,289) .. (304.5,289) .. controls (300.36,289) and (297,285.64) .. (297,281.5) -- cycle ;
\draw  [fill=gray!60  ,fill opacity=1 ] (256,176) .. controls (256,172.69) and (258.69,170) .. (262,170) -- (280,170) .. controls (283.31,170) and (286,172.69) .. (286,176) -- (286,284) .. controls (286,287.31) and (283.31,290) .. (280,290) -- (262,290) .. controls (258.69,290) and (256,287.31) .. (256,284) -- cycle ;
\draw  [fill=gray!60  ,fill opacity=1 ] (317,281.5) .. controls (317,277.36) and (320.36,274) .. (324.5,274) .. controls (328.64,274) and (332,277.36) .. (332,281.5) .. controls (332,285.64) and (328.64,289) .. (324.5,289) .. controls (320.36,289) and (317,285.64) .. (317,281.5) -- cycle ;
\draw  [fill=gray!60  ,fill opacity=1 ] (146,48) .. controls (146,44.69) and (148.69,42) .. (152,42) -- (170,42) .. controls (173.31,42) and (176,44.69) .. (176,48) -- (176,152) .. controls (176,155.31) and (173.31,158) .. (170,158) -- (152,158) .. controls (148.69,158) and (146,155.31) .. (146,152) -- cycle ;

\draw (228,248) node [anchor=north west][inner sep=0.75pt]   [align=left] {...};

\end{tikzpicture}
    \caption{}
    \label{fig:sub2-n}
\end{subfigure}
\caption{(a) If one singleton is clustered together with a clique, the cost is at least $2 \ln k$.
(b) If two large cliques are clustered together, the cost is at least $4 \ln^2 k$.}
\label{fig:both'}
\end{figure}
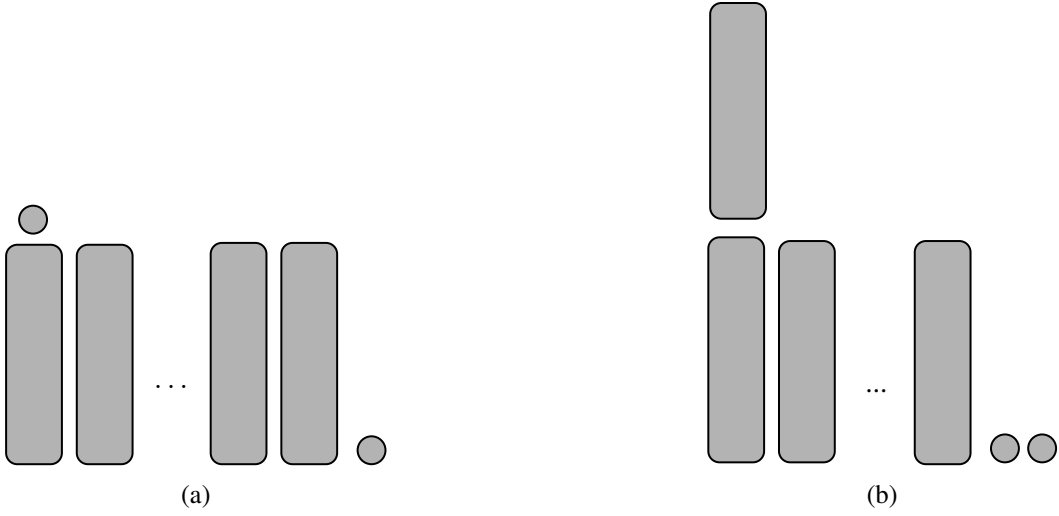
   

\lowerbound*

\begin{proof}
We consider the behavior of any online algorithm on the instance consisting of $k-1$ cliques of size $2\ln k$ and $2$ singletons. 
Let $c=2/7$ and assume $k\geq 16 > \frac{2}{c}$.
From \Cref{lem:singletons_arrive}, the probability that both singletons arrive by time $ck\ln k$ is at least $\frac{1}{300}$.
Additionally by \Cref{lem:clique_missing}, the probability that at least one clique has not yet arrived by time  $ck\ln k$ is at least $1 - e^{-\frac{\sqrt{k}}{16}}>\frac{1}{5}$ since $k\geq 16$.

Let $\mathcal{E}$ be the intersection of both events --- both singletons arrive by time $ck\ln k$ (event ${\cal E}_1$) and at least one clique has not yet arrived by time  $ck\ln k$ (event ${\cal E}_2$).  
We prove that $\Pr[\mathcal{E}_2] \ge \Pr[\mathcal{E}_2 | \mathcal{E}_1]$. 
Conditioned on $s \in \{0,1,2\}$ singletons arriving in the first $ck\ln k$ steps, the set of clique vertices arriving in the first $ck\ln k$ steps is a uniform sample of size $(ck \ln k - s)$ drawn without replacement from $(n-2)$ clique vertices; hence, the probability that at least one clique has not yet arrived by time  $ck\ln k$ is maximized when $s = 2$.  
Therefore $\Pr[\mathcal{E}_2] \ge \Pr[\mathcal{E}_2 | \mathcal{E}_1]$, and hence $\Pr[\mathcal{E}]$ is at least $\Pr[{\cal E}_1]\Pr[{\cal E}_2]\geq \frac{1}{1500}$, which is a constant.   

Let $t$ be the time step when the last, or $(k+1)$-st representative from a clique or singleton arrives.
Let $\mathcal{S}_3$ denote the unlabeled observable event that at time $t$, there are $m \ge 3$ isolated vertices which do not yet have any neighbors from the set of arrived vertices.  Since $\mathcal{E}\subseteq\mathcal{S}_3$, so $\Pr[\mathcal{S}_3]\geq \frac{1}{1500}$.  Given $\mathcal{S}_3$, permuting the $m$ isolated vertices among their $m$ arrival timestamps while keeping all other $n - m$ arrivals fixed preserves the entire sequence of unlabeled revealed graphs until time $t$.  That is, on event $\mathcal{S}_3$, the algorithm has no way of distinguishing which of these at least three vertices are true singletons versus which are from a clique.

Since there are $(k+1)$ unconnected components at this time, the algorithm must choose two to cluster together. 
If the two components it clusters together include one that has at least two vertices (and thus it knows to be a clique), then the algorithm commits to paying cost $\Omega(\log n)$.  So assume it avoids making this (clearly sub-optimal) decision.

Otherwise, the two components it clusters together both have exactly one vertex.
The algorithm ``succeeds'' if it clusters both singletons together.  It fails if it makes any other choice.
Since there are at least 3 indistinguishable vertices, the algorithm succeeds with probability at most $\frac{1}{3}$, and fails with probability at least $\frac{2}{3}$.  We now argue that when the algorithm fails, then no matter what decisions the algorithm makes in the future, it will incur a cost at least $\Omega(\log n)$.  Suppose by time $t$, an isolated vertex $v$ belonging to clique $K_j$ is placed in the same cluster as a vertex $u \notin K_j$.  Then, for each of the remaining $2 \ln k - 1$ vertices $w \in K_j \setminus \{v\}$, assigning $w$ to $v$'s cluster incurs a non-edge disagreement cost with $u$, while assigning $w$ to a different cluster incurs an edge disagreement cost with $v$.  Therefore, any online algorithm must incur a cost of $\Omega(\log n)$ with constant probability on this graph instance.

Finally, we note that our lower bound instance uses $k = \Omega(n/\log n)$, and hence we have also derived an $\Omega(\log n)$ lower bound.
\end{proof}

\subsection{Lower Bound for Online Algorithms which do not Split Cliques} \label{sec:lb_logsq}

To conclude this section, we show a stronger lower bound of $\Omega(\log^2 n)$-competitiveness for online algorithms which never split disjoint cliques, using the same hard instance. 
Notably, this lower bound applies to our online algorithms \BP\ and \UP\ and thus shows that our algorithm's performance is no better than polylogarithmic-competitive.
Additionally, any online algorithm that uses \Pivot\ as a building block will also face the same barrier.

This lower bound indicates a potential difference between how to optimally solve \kCC{} in the offline versus the online case.
In the offline case, the optimal solution never splits disjoint cliques.
However, it is not out of the question that a $\bigo(\log n)$-competitive algorithm could exist, assuming the algorithm is able to smartly split the cliques to rectify past mistakes in its clustering. 
Further pursuit of a $\bigo(\log n)$-competitive algorithm is left as an open problem. 

\lblogsq*

We use the same lower-bound graph instance as before. 
There are $(k-1)$ cliques each of size $2\ln k$, and two singletons.

To prove this stronger lower bound, we need to show that with constant probability, any online algorithm will accidentally cluster the first two vertices from two separate cliques into the same cluster. 
Since the algorithm never splits cliques, this will force the algorithm to pay a cost of $\Omega(\log^2 k)$.

To start, we consider the following stronger Chernoff result.

\begin{lemma}[Two Cliques Arrive Late]\label{lem:lb_2clique_arrive_late}
    Assuming $c<\frac{2}{5}$ and $k\geq 64$, then with probability at least $1-e^{-\frac{\sqrt{k}}{16}}$, at least two cliques arrive after time $ck\ln k$.
\end{lemma}

\begin{proof}
For each clique $K_j$, let $X_j$ denote the indicator variable for if clique $K_j$ did not arrive by time $ck\ln k$. 

Recalling \Cref{eq:ex_num_clique_not_arrive}, assuming $k\geq 8$ and $c<2/5$, then
\[ \E\left[\sum_{j=1}^{k-1} X_j\right] \geq \frac{1}{2}\sqrt{k} . \]

Since by \Cref{lem:vars_are_na} the variables $X_1,\hdots,X_{k-1}$ are negatively associated, we can apply Chernoff Bounds (\Cref{lem:chernoff}). 
Assuming $k\geq 64$, then $\frac{1}{2}\sqrt{k}\geq 4$, and we can set $\delta = \frac{1}{2}$.
\[
    \Pr[ X < 2 ] \leq \Pr\left[ X < \frac{1}{2} \E[X] \right] \leq e^{-\frac{\sqrt{k}}{16}} .
\]

\end{proof} 

\begin{proof}[Proof of \Cref{thm:lb_logsq}]

Consider the behavior of any online algorithm on the instance consisting of $(k-1)$ cliques of size $2\ln k$ and $2$ singletons.
Let $c=\frac{2}{7}$ and $k\geq 64 > \frac{2}{c}$.

From \Cref{lem:singletons_arrive}, the probability that both singletons arrive by time $ck\ln k$ is at least $\frac{1}{300}$.
Additionally by \Cref{lem:lb_2clique_arrive_late}, the probability that at least two cliques have not yet arrived by time  $ck\ln k$ is at least $1 - e^{-\frac{\sqrt{k}}{16}}>\frac{1}{3}$ since $k\geq 64$.

Analogous to the proof of Theorem~\ref{thm:onlinelowerbound}, let $\mathcal{E}$ be the intersection of both events---both singletons arrive by time $ck\ln k$ (event ${\cal E}_1$) and at least two cliques has not yet arrived by time  $ck\ln k$ (event ${\cal E}_2$).  We prove that $\Pr[\mathcal{E}_2] \ge \Pr[\mathcal{E}_2 | \mathcal{E}_1]$. Conditioned on $s \in \{0,1,2\}$ singletons arriving in the first $ck\ln k$ steps, the set of clique vertices arriving in the first $ck\ln k$ steps is a uniform sample of size $ck \ln k - s$ drawn without replacement from $n-2$ clique vertices; hence, the probability that at least two cliques have not yet arrived by time  $ck\ln k$ is maximized when $s = 2$.  Therefore, $\Pr[\mathcal{E}_2] \ge \Pr[\mathcal{E}_2 | \mathcal{E}_1]$, and $\Pr[\mathcal{E}]$ is at least $\Pr[{\cal E}_1]\Pr[{\cal E}_2]$, and hence at least $\frac{1}{1500}$, which is a constant.   

The probability that both singletons arrive by time $ck\ln k$ conditional on the event that at least two cliques do not arrive by time $ck\ln k$ is greater than or equal to the unconditional probability that both singletons arrive by time $ck\ln k$.
Let $\mathcal{E}$ be the event that both events happen.
Then the probability of $\mathcal{E}$ is at least their product, or $\Pr[\mathcal{E}]\geq \frac{1}{900}$, which is a constant. 

Conditional upon an arrival pattern with this outcome, consider the first timestep $\geq ck\ln k$ at which exactly two large cliques have not yet arrived.  For the remainder of the arrival pattern, the probability that we receive exactly one vertex from each remaining large clique before we receive a second vertex from either is at least $\frac{1}{2}$.

Let $\mathcal{S}_4$ denote the unlabeled observable event that at the time $t$ when the $(k+1)$-st representative from a clique or singleton arrives there are $m \ge 4$ isolated vertices which do not yet have any neighbors from the set of arrived vertices.  Then clearly, $\Pr[\mathcal{S}_4]\geq \frac{1}{2}\cdot\Pr[\mathcal{E}] \ge \frac{1}{1800}$.  On event $\mathcal{S}_4$, permuting the $m$ isolated vertices among their $m$ arrival times while keeping all other $n - m$ arrival positions fixed preserves the entire sequence of unlabeled graphs revealed up to time $t$.
That is, under event $\mathcal{S}_4$, the algorithm has no way of distinguishing which of these at least four vertices are true singletons versus which are from a clique.

The algorithm must select two components to cluster together.  It fails if it clusters two cliques together.
If it clusters two components it knows to be a clique, it knowingly fails and pays cost $\Omega(\log^2 k)$, so assume the algorithm does not make this (clearly sub-optimal) decision.
If it clusters a component it knows to be a clique with one of the $\geq 4$ isolated vertices, then (conditional on event $\mathcal{S}_4$) it fails with probability at least $\frac{1}{2}$.
Otherwise, if it clusters two isolated vertices together, then the probability that it selects two vertices which are both from different large cliques to cluster together is at least $\frac{1}{6}$.

Since the algorithm cannot later split disjoint cliques, when it makes these mistakes it must pay cost $\Omega(\log^2 k)$.
Since this occurs with constant probability, then the expected cost of any online algorithm which never splits disjoint cliques must be at least $\Omega(\log^2 k)$.

Finally, we note that our lower bound instance uses $k = \Omega(n/\log n)$, and hence we have also derived an $\Omega(\log^2 n)$ lower bound.

\end{proof}
\section{Polylogarithmic Competitiveness for Disjoint Cliques}
\label{sec:cliques}
In this section, we consider a special case of the problem where the graph being clustered consists of a collection of disjoint cliques. We give an algorithm \UP\ which yields $\bigo(\log^3 n)$-competitiveness for correlation $k$-clustering against random arrivals. Some of the technical ideas used in the general graphs result are easier to demonstrate in this special case. Also, our analysis for disjoint cliques highlights some of the challenges that we need to overcome in the general case, which requires a much more sophisticated probabilistic analysis.  

Let us recall how the \Pivot\ algorithm works for correlation clustering. For any vertex $v$, if no neighbor of $v$ is a pivot vertex when $v$ arrives, it is designated as a \textit{pivot} vertex, and a new cluster is initialized containing $v$. Else, it is assigned to the cluster to which its earliest arrived pivot neighbor is assigned.

\paragraph{Algorithm \UP.} Our algorithm \UP\ works as follows. The behavior of our algorithm is identical to the \Pivot\ algorithm until the timestep at which the $k$th pivot arrives. Let us refer to this timestep as $\Tau_0$. At time $\Tau_0$, let $S$ denote the set of clusters which have at most $c_0\log n$ vertices assigned to them, for a suitably chosen constant $c_0>0$.  (Here and below, $\log$ denotes $\log_2$.)  Note that $S$ is nonempty since the $k$th pivot forms its own cluster of size 1 when it arrives.

For any vertex $v$ that arrives at time $t\geq \Tau_0+1$, \UP\ does the following: if $v$ is a pivot vertex, it is assigned to a cluster that is picked uniformly at random from the set $S$; otherwise, $v$ is assigned to the cluster to which its earliest arrived pivot neighbor is assigned. 

\cliquesthm*

We first observe that if the number of cliques is at most $k$ or $n = 1$, then the cost of \UP\ is deterministically zero.  So, in the remainder of this proof, we restrict to the case where $n \ge 2$ and the number of cliques exceeds $k$.  We give an upper bound of $\bigo(\sml{\instance}\avg{\instance}\log^3 n)$ on the expected cost of \UP\ which, in light of Lemma~\ref{lem:opt-theta}, implies $\bigo(\log^3 n)$-competitiveness.

To analyze the cost of \UP, we first note that the algorithm does not incur any disagreement cost between vertex pairs of the same clique since for every clique, all of its vertices are assigned to the same cluster. Furthermore, we note that the algorithm does not incur any disagreement cost until time $\Tau_0$.  Subsequently, the only disagreement cost arises from pairwise disagreements between vertices in distinct cliques assigned to the same cluster. That is, for any cluster $S_i \in S$, and cliques $K_{j_1}, K_{j_2},\dots,K_{j_b}$ (ordered with respect to their pivot arrivals) that are assigned to $S_i$, the total disagreement cost attributed to cluster $S_i$ is given by $\sum_{c,d: c<d}n_{j_c}n_{j_d}$ where $n_{j_c} = |K_{j_c}|$ for any $c\in[b]$. This can be alternatively expressed as $\sum_{d>1} n_{j_1}n_{j_d} + \sum_{c,d: 1<c<d} n_{j_c}n_{j_d}$. In other words, the total disagreement cost is the sum of:

\begin{enumerate}[label=\alph*)]
    \item the disagreement cost between a clique $K_j$, all of whose vertices arrive after time $\Tau_0$, that is assigned to $S_i$, and the \emph{leading clique} of $S_i$, i.e., the unique clique whose pivot arrived by time $\Tau_0$ and was assigned to $S_i$;
    \item the disagreement cost between two cliques, all of whose vertices arrive after time $\Tau_0$, that are assigned to the same cluster $S_i$.
\end{enumerate}

 We say a vertex $v$ is \emph{covered} at time $t$ if $v$ belongs to a clique from which at least one vertex has arrived by time $t$; otherwise $v$ is \emph{uncovered} at time $t$. Let $U$ denote the set of vertices uncovered at time $\Tau_0$, and let $W = |U|$.

    To bound these two costs, we first show that a large clique has many vertices arriving early (Lemma~\ref{lem:few-vertices-by-tau}). Consequently, if $\Tau_0$ is large, then every clique that is uncovered at time $\Tau_0$, or that is the leading clique of a cluster in $S$, is small. We then show that $S$ contains at least $k - r$ clusters with high probability (Lemma~\ref{lem:logn_set_size}). The proof of Theorem~\ref{thm-cliques-random} combines these two lemmas with elementary first- and second-moment bounds on the number of early arrivals from fixed sets of vertices.

\begin{lemma}[Lower bound on number of vertices arriving from large cliques]
    \label{lem:few-vertices-by-tau}
    For any time $t$ and constant $c_1>0$, and any clique of size at least $\frac{c_1 n \log n}{t}$, the number of vertices in the clique that arrive by time $t$ is more than $\frac{1}{2}\cdot c_1\log n$ with probability at least $1-n^{-c_1/8}$.
\end{lemma}
\begin{proof}
Consider any clique $K_j$ of size $|K_j| \geq \frac{c_1 n \log n}{t} $. Let $X$ be the number of vertices that arrive from $K_j$ by time $t$. 
Then $\E[X]\geq c_1\log n$ and $X \sim \mathrm{Hypergeometric}(n,\, |K_j|,\, t)$. 
By Facts~\ref{lem:chernoff} and \ref{lem:chernoff-hypergeometric} (setting $\delta = \frac{1}{2}$),
\[
  \Pr\left[X \leq \tfrac{1}{2}\cdot c_1\log n \right]
  \;\leq\; \exp\!\left(-\frac{c_1 \log n}{8}\right)
  \;\leq\; n^{-c_1/8}.
\]
\end{proof}

\begin{lemma}[Sufficiently many clusters have small load at the time of arrival of the $k$th pivot]
\label{lem:logn_set_size}
Let $c>4$ be any constant and let $S$ be the set of clusters with load at most $4c\ln n$ at the time the $k$th pivot arrives. 
Then $|S|$ is at least $(k-r)$ with probability at least $1-n^{1-c}-n^{-c/8}-n^{1-c/4}-n^{1-c/3}$.
\end{lemma}

Here and below, we write $r = r(\instance)$ (Definition~\ref{def:threshold}). Let $A = \bigcup_{j=r+1}^{\ell} K_{j}$, and recall that $\Tau_0$ denotes the time step at which the $k$th pivot arrives. Note that the number of pivots that arrive from $A$ by time $\Tau_0$ is at least $k - r$. 

\begin{proof}
We establish the claim through the following two steps. 
\begin{enumerate}
    \item We will first show that $\Tau_0\leq T'$ with high probability, where $T' = \left\lceil\frac{c n \ln n}{a(\instance)}\right\rceil$.
    \item For any clique $K_j$ in $A$, the number of vertices from $K_j$ that arrive by time $T'$ is $\bigo(\log n)$, also with high probability.
\end{enumerate}

If $T' \ge n$, then $\frac{cn\ln n}{a(\instance)} > n-1$, so $a(\instance) < \frac{cn\ln n}{n-1} \le 2c\ln n$ for $n \ge 2$. Every clique in $A$ has fewer than $a(\instance)$ vertices by Definition~\ref{def:threshold}, so every cluster whose leading clique lies in $A$ has load less than $2c\ln n < 4c\ln n$ at time $\Tau_0$. Since at least $k-r$ of the pivots that arrive by time $\Tau_0$ lie in $A$, we have $|S| \ge k-r$ with probability $1$. Hence we assume from now on that $T' \le n - 1$, so that $T'$ is a valid sample size in the hypergeometric arguments below.

These two steps together will imply that every pivot cluster that arrives from $A$ must have load $\bigo(\log n)$ at time $\Tau_0$ with high probability, yielding the lemma.

\newcommand{\poly}{\mbox{poly}}
We now show that $\Tau_0$ is at most $T' = \left\lceil\frac{c n \ln n}{a(\instance)}\right\rceil$. 
Each clique $K_j$ with $j \le r$ has size $y \ge a(\instance)$. Indeed, by the minimality of $r$ in Definition~\ref{def:threshold}, $n_r \ge \frac{n_r + (k-r)\,a(\instance)}{k-r+1}$, which rearranges to $n_r \ge a(\instance)$, and $n_j \ge n_r$ for $j \le r$.  If $T' > n-y$, then the desired claim trivially holds.  Otherwise, the probability that such a clique doesn't have any vertex arrive by time $T'$ is no more than 
\[
    \frac{\binom{n - y}{T'}}{\binom{n}{T'}} = \prod_{i=0}^{T'-1}\frac{n-y-i}{n-i} \leq \left(1 - \frac{y}{n}\right)^{T'} 
    \leq \left(1 - \frac{a(\instance)}{n}\right)^{T'} 
  \leq \exp \left( -\frac{a(\instance)}{n} \cdot T' \right)
  \leq \exp \left( -c \ln n \right) = \frac{1}{n^{c}}.
\]
By a union bound over the $r$ cliques $K_1, \dots, K_r$, we obtain that a pivot from each of the $r$ large cliques has arrived by time $T'$ with probability at least $1 - n^{1-c}$.  (If $r = 0$, then there is no clique $K_j$ and the claim is trivially true.)

Now consider the set $A$. 
By definition, $|A| = (k-r)a(\instance)$. 
Let $X$ be the random variable indicating the number of vertices that arrive from $A$ by time $T'$.
Then 
\[ \E[X] = \frac{|A|\cdot T'}{n} \geq \frac{(k-r)a(\instance)}{n} \cdot \frac{c n\ln n}{a(\instance)} = c (k-r) \ln n. \]
 
Since $X\sim\mathrm{Hypergeometric}(n, |A|, T')$, then we can use Fact \ref{lem:chernoff-hypergeometric} (setting $\delta = \frac{1}{2}$). 
\begin{align*}
    \Pr[X\leq \frac{1}{2}c(k-r)\ln n] & \leq \exp\left(-\frac{c(k-r)\ln n}{8}\right) \leq n^{-c/8} .
\end{align*}

Furthermore, each clique in $A$ has at most $a(\instance)$ vertices by definition. 
Then there exists a way to partition some subset of $\{K_{r+1},\hdots,K_\ell\}$ into $(k-r)$ groups, each with total number of vertices at least $a(\instance)/2$.
We describe the process below.

First assign each clique of size at least $a(\instance)/2$ to be in a group by itself.  If the number of groups created exceeds $k-r$, then select an arbitrary $k-r$ of them.  If we have not already created $(k-r)$ groups, 
then build the remaining groups by placing cliques together greedily until total size is at least $a(\instance)/2$ and at most $a(\instance)$. 
Since remaining cliques are all of size no more than $a(\instance)/2$, a greedy strategy will work.

We now argue that at least one vertex (and hence at least one pivot) from each group is drawn during a random draw of $q = \lceil \frac{1}{2}c(k-r) \ln n \rceil$ vertices with high probability. 
Fix a group of size $g \ge a(\instance)/2$.  First, if $(k-r)a(\instance) - g < q$, the claim holds deterministically.  Otherwise, the probability that no vertex is drawn from the group is upper bounded by
\[
\frac{\binom{(k-r)a(\instance) - g}{q}}{\binom{(k-r)a(\instance)}{q}} \le \left(1 - \frac{g}{(k-r)a(\instance)}\right)^{q}
\le e^{-\frac{1}{2}c (\ln n)/2} = n^{-c/4}.\]

By a union bound over the at most $k-r \le n$ groups, at least one vertex from every group, and hence at least $(k-r)$ pivots from $A$, arrive by time $T'$ with probability at least $1-n^{-c/8}-n^{1-c/4}$.
Together with the $r$ large cliques, which all have a pivot by time $T'$ with probability at least $1-n^{1-c}$, at least $k$ pivots arrive by time $T'$, and hence $\Tau_0\leq T'$.
Therefore, we have shown that $\Tau_0\leq T'$ with probability at least $1-n^{1-c}-n^{-c/8}-n^{1-c/4}$.

Now, fix some clique $K_j$ for $j>r$. Let $Y$ denote the number of vertices from $K_j$ which arrive by time $T'$.
Since all cliques in $A$ are of size no more than $a(\instance)$, then clearly $\E[Y] \leq a(\instance) \cdot\frac{T'}{n} \leq 2c \ln n$.
Using Fact~\ref{lem:chernoff-hypergeometric} (with $\delta = 1$; the upper tail in Fact~\ref{lem:chernoff} remains valid when $\mu$ is replaced by any upper bound on $\E[Y]$), we can then show

\[
  \Pr\left[Y \geq 4 c \ln n\right] \leq \exp\left(-\frac{c\ln n}{3}\right) = n^{-c/3} .
\]
By a union bound over all $(\ell-r) \leq n$ cliques in $A$, then with probability at least $1 - n^{1-c/3}$,
every clique in $A$ contributes at most $4c\ln n$ vertices by time $T'$.

Then when $\Tau_0 \leq T'$, it follows that each pivot that arrives from $A$ by time $\Tau_0$ has at most $4c \ln n$ vertices in its cluster by time $\Tau_0$.
Since at least $(k-r)$ pivots from $A$ arrive by time $\Tau_0$, it then follows that $S$ is of size at least $k-r$ with total probability at least $1-n^{1-c}-n^{-c/8}-n^{1-c/4}-n^{1-c/3}$.

\end{proof}

\cliquesthm*

\begin{proof}[Proof of Theorem~\ref{thm-cliques-random}]
If the instance has at most $k$ cliques, then \UP\ never has more than $k$ pivots, places every clique in its own cluster, and incurs zero cost; in this case $\sml{\instance}=0$. Hence we assume $\ell > k$. Then $\sml{\instance} \ge 1$, and $\avg{\instance} > n_{r+1} \ge 1$ by Definition~\ref{def:threshold}, so $\sml{\instance}\avg{\instance} \ge 1$.

Let $\Tau_0$ be the time when the $k$th pivot arrives, let $U$ be the set of vertices uncovered at time $\Tau_0$, and let $W = |U|$. Let $S$ be the set of clusters with load at most $c_0 \log n$ at time $\Tau_0$. Recall that $A = \bigcup_{j>r} K_j$, so that $|A| = \avg{\instance}(k-r)$ by Definition~\ref{def:threshold}. Since $k > r$, we have $\bigcup_{j>k} K_j \subseteq A$.

\paragraph{Structure of the cost.}
Up to time $\Tau_0$, \UP\ behaves as \Pivot, so each of the $k$ clusters contains exactly one clique at time $\Tau_0$; we call it the \emph{leading clique} of the cluster. Every clique not contained in $U$ has a vertex arriving by time $\Tau_0$; this vertex is a pivot, and all later vertices of the clique join the same cluster. Every clique $K \subseteq U$ has no vertex arriving by time $\Tau_0$. Its first vertex arrives after $\Tau_0$, becomes a pivot, and is assigned to a cluster chosen uniformly at random from $S$; though $S$ depends on the arrival order $\sigma$, the random coin used to choose the cluster for the pivot is independent of $\sigma$ and of all other such choices; all later vertices of $K$ join the same cluster. In particular, no clique is ever split, and two leading cliques never share a cluster. Hence $\mathrm{Cost} = \text{Cost}_1 + \text{Cost}_2$, where
\[
\text{Cost}_1 = \sum_{K \subseteq U} |K| \cdot |L(K)|, \qquad
\text{Cost}_2 = \sum_{\substack{a < b:\; K_a, K_b \subseteq U \\ K_a, K_b \text{ in the same cluster}}} n_a n_b,
\]
and $L(K)$ denotes the leading clique of the cluster to which $K$ is assigned.

\paragraph{Time intervals and fixed-set counts.}
For $i \in \{0, 1, \dots, \lfloor \log n \rfloor\}$, let $\mathcal{T}_i$ denote the event $\Tau_0 \in [2^i, 2^{i+1})$. Since $1 \le \Tau_0 \le n$, exactly one of these events occurs. Define
\[
p_i = \frac{2^{i+1}}{n}, \qquad B_i = \frac{c_1 n \log n}{2^i}, \qquad \text{so that } B_i\, p_i = 2 c_1 \log n,
\]
where $c_1 = \max\{40,\, 2c_0 + 1\}$. Let $Y_i$ and $N_i$ denote the number of vertices of $\bigcup_{j>k} K_j$ and of $A$, respectively, among the first $\min\{2^{i+1}, n\}$ arrivals (i.e., that arrive by time $2^{i+1}$). 

The sets $\bigcup_{j>k} K_j$ and $A$ are determined by the instance, and the number of arrivals considered is fixed. So $Y_i$ and $N_i$ are hypergeometric, with no conditioning on $\Tau_0$. Every vertex is among the first $\min\{2^{i+1}, n\}$ arrivals with probability at most $p_i$, and every pair of distinct vertices is with probability at most $p_i^2$. Hence
\begin{equation}
\label{eq:moments-1}
\E[Y_i] \le \sml{\instance}\, p_i, \qquad \E[N_i] \le \avg{\instance}(k-r)\, p_i, \qquad \E[Y_i(Y_i - 1)] \le \sml{\instance}^2 p_i^2.
\end{equation}
Moreover, $Y_i(N_i-1)$ counts ordered pairs $(u,v)$ of distinct vertices among the first $\min\{2^{i+1}, n\}$ arrivals with $u \in \bigcup_{j>k}K_j$ and $v \in A$. Indeed, each such $u$ also lies in $A$, so there are exactly $N_i - 1$ arrived vertices of $A$ other than $u$. There are at most $\sml{\instance}\cdot\avg{\instance}(k-r)$ such pairs, each among the first $\min\{2^{i+1}, n\}$ arrivals with probability at most $p_i^2$. Therefore
\begin{equation}
\label{eq:moments-2}
\E[Y_i(N_i - 1)] \;\le\; \sml{\instance} \cdot \avg{\instance}(k-r) \cdot p_i^2 .
\end{equation}

\paragraph{A good event.}
Let $\mathcal{G}_1$ be the event that for every $i \in \{0, \dots, \lfloor \log n \rfloor\}$, every clique of size at least $B_i$ has more than $c_0 \log n$ vertices arriving by time $2^i$. By Lemma~\ref{lem:few-vertices-by-tau} with $t = 2^i$, a fixed clique of size at least $B_i$ has more than $\frac{1}{2} c_1 \log n > c_0 \log n$ vertices arriving by time $2^i$ with probability at least $1 - n^{-c_1/8}$. A union bound over at most $n$ cliques and $\lfloor \log n \rfloor + 1 \le n$ values of $i$ gives
\[
\Pr[\bar{\mathcal{G}}_1] \le n^{2 - c_1/8} \le n^{-3}.
\]

\begin{claim}
\label{clm:good-event}
Fix $i$. On the event $\mathcal{G}_1 \cap \mathcal{T}_i$, the following hold.
\begin{enumerate}[label=(\alph*)]
    \item Every leading clique of a cluster in $S$ has size less than $B_i$.
    \item Every clique contained in $U$ has size less than $B_i$.
    \item Let $m$ denote the number of cliques among $K_1, \dots, K_k$ that are contained in $U$. Then $m \le Y_i$, and consequently $W \le \sml{\instance} + B_i\, Y_i$.
    \item $N_i \ge k - r$.
\end{enumerate}
\end{claim}
\begin{proof}
(a) A leading clique of a cluster in $S$ has at most $c_0 \log n$ vertices arriving by time $\Tau_0 \ge 2^i$, hence at most $c_0 \log n$ vertices arriving by time $2^i$. On $\mathcal{G}_1$, it therefore has size less than $B_i$.

(b) A clique contained in $U$ has no vertex arriving by time $\Tau_0 \ge 2^i$. On $\mathcal{G}_1$, it therefore has size less than $B_i$.

(c) The $k$ pivots that arrive by time $\Tau_0$ belong to $k$ distinct cliques, and a clique among $K_1, \dots, K_k$ is contained in $U$ if and only if none of these $k$ pivots belongs to it. Hence $m$ equals the number of these $k$ pivots that lie in $\bigcup_{j>k} K_j$. Each such pivot arrives by time $\Tau_0 < 2^{i+1}$, so $m \le Y_i$. The cliques contained in $U$ are cliques among $K_{k+1}, \dots, K_\ell$, of total size at most $\sml{\instance}$, together with $m$ cliques among $K_1, \dots, K_k$, each of size less than $B_i$ by (b). Hence $W \le \sml{\instance} + B_i\, m \le \sml{\instance} + B_i\, Y_i$.

(d) At most $r$ of the $k$ pivots that arrive by time $\Tau_0$ belong to $K_1, \dots, K_r$, so at least $k - r$ of them are vertices of $A$. All of them arrive by time $\Tau_0 < 2^{i+1}$.
\end{proof}

Since the cost of any $k$-clustering is at most $n^2$, the contribution of $\bar{\mathcal{G}}_1$ to the expected cost is at most $n^2 \Pr[\bar{\mathcal{G}}_1] \le 1/n$.
\paragraph{$\text{Cost}_1$: Disagreements between a clique in $U$ and the leading clique of its cluster.}
For a nonnegative random variable $X$ and an event $\mathcal{E}$, $\E[X \cdot \mathbf{1}_{\mathcal{E}}]$ is the contribution of $\mathcal{E}$ to $\E[X]$. Every clique $K \subseteq U$ is assigned to a cluster in $S$, so by Claim~\ref{clm:good-event}(a), on $\mathcal{G}_1 \cap \mathcal{T}_i$,
\begin{equation}
\label{eq:cost1-pointwise}
\text{Cost}_1 = \sum_{K \subseteq U} |K| \cdot |L(K)| \le B_i\, W.
\end{equation}
We bound $\E[W \cdot \, \mathbf{1}_{\mathcal{G}_1 \cap \mathcal{T}_i}]$ for each fixed $i$, distinguishing two cases.

\noindent\textit{Case $k - r \ge 2$.} On $\mathcal{G}_1 \cap \mathcal{T}_i$, Claim~\ref{clm:good-event}(c) and (d) give $W \le \sml{\instance} + B_i Y_i$ and $N_i \ge k-r$. In particular, $\mathcal{G}_1 \cap \mathcal{T}_i$ is contained in the event $N_i \ge k-r$, and on this event $\frac{N_i - 1}{k-r-1} \ge 1$, so $Y_i \le \frac{Y_i(N_i-1)}{k-r-1}$. Hence
\[
\E[W \cdot \, \mathbf{1}_{\mathcal{G}_1 \cap \mathcal{T}_i}]
\;\le\; \sml{\instance}\, \Pr[N_i \ge k-r] \;+\; B_i\, \frac{\E[Y_i (N_i - 1)]}{k-r-1}.
\]
By Markov's inequality and~\eqref{eq:moments-1}, $\Pr[N_i \ge k-r] \le \E[N_i]/(k-r) \le \avg{\instance}\, p_i$. By~\eqref{eq:moments-2} and $\frac{k-r}{k-r-1} \le 2$, $\E[Y_i(N_i-1)]/(k-r-1) \le 2\, \sml{\instance}\avg{\instance}\, p_i^2$. Combined with~\eqref{eq:cost1-pointwise}, this gives
\[
\E[\text{Cost}_1\cdot \, \mathbf{1}_{\mathcal{G}_1 \cap \mathcal{T}_i}]
\;\le\; B_i\, \sml{\instance}\avg{\instance}\, p_i + 2 B_i^2\, \sml{\instance}\avg{\instance}\, p_i^2
\;=\; \left(2c_1 \log n + 8 c_1^2 \log^2 n\right) \sml{\instance}\avg{\instance}.
\]

\noindent\textit{Case $k - r = 1$.} In this case $|A| = \avg{\instance}$, and since $A = K_k \cup \bigcup_{j>k} K_j$, we have $n_k = \avg{\instance} - \sml{\instance}$. Let $\mathcal{E}$ be the event that each of $K_1, \dots, K_r$ contains one of the $k$ pivots that arrive by time $\Tau_0$. (If $r = 0$, then $\mathcal{E}$ always holds.)

On $\mathcal{E}$, exactly one of the $k = r+1$ pivots belongs to $A$; let $K'$ be its clique. No other clique in $A$ has a vertex arriving by time $\Tau_0$, since the first such vertex arriving would itself be a pivot, while $K_1, \dots, K_r$ are covered. Hence $W = \avg{\instance} - |K'|$. If $K' = K_k$, then $W = \sml{\instance}$, and on $\mathcal{T}_i$ a vertex of $A$ arrives by time $2^{i+1}$, i.e., $N_i \ge 1$. If $K' \subseteq \bigcup_{j>k} K_j$, then $W \le \avg{\instance}$, and on $\mathcal{T}_i$ a vertex of $\bigcup_{j>k} K_j$ arrives by time $2^{i+1}$, i.e., $Y_i \ge 1$. Thus
\[
\E[W \cdot \, \mathbf{1}_{\mathcal{E} \cap \mathcal{T}_i}]\;\le\; \sml{\instance}\, \Pr[N_i \ge 1] + \avg{\instance}\, \Pr[Y_i \ge 1].
\]
By Markov's inequality and~\eqref{eq:moments-1}, the right-hand side is at most $\sml{\instance} \cdot \avg{\instance}\, p_i + \avg{\instance} \cdot \sml{\instance}\, p_i = 2\sml{\instance}\avg{\instance}\, p_i$. Combined with~\eqref{eq:cost1-pointwise} (restricting further to $\mathcal{G}_1$ only decreases the contribution),
\[
\E[\text{Cost}_1 \cdot \, \mathbf{1}_{\mathcal{G}_1 \cap \mathcal{E} \cap \mathcal{T}_i}] \;\le\; 2 B_i\, p_i\, \sml{\instance}\avg{\instance} \;=\; 4 c_1\, \sml{\instance}\avg{\instance} \log n.
\]

On $\bar{\mathcal{E}}$, some clique among $K_1, \dots, K_r$ contains none of the $k = r+1$ pivots that arrive by time $\Tau_0$. So at most $r - 1$ of them belong to $K_1, \dots, K_r$, and at least two are vertices of $A$; on $\mathcal{T}_i$, this gives $N_i \ge 2$, so $N_i - 1 \ge 1$ and $Y_i \le Y_i(N_i-1)$. Together with Claim~\ref{clm:good-event}(c),
\[
\E[W\cdot \, \mathbf{1}_{\mathcal{G}_1 \cap \bar{\mathcal{E}} \cap \mathcal{T}_i}]
\;\le\; \sml{\instance}\, \Pr[N_i \ge 1] + B_i\, \E[Y_i (N_i - 1)].
\]
By~\eqref{eq:moments-1} and~\eqref{eq:moments-2} with $k - r = 1$, the right-hand side is at most $\sml{\instance}\avg{\instance}\, p_i + B_i\, \sml{\instance}\avg{\instance}\, p_i^2$. Combined with~\eqref{eq:cost1-pointwise},
\[
\E[\text{Cost}_1 \cdot \, \mathbf{1}_{\mathcal{G}_1 \cap \bar{\mathcal{E}} \cap \mathcal{T}_i}]
\;\le\; \left(2 c_1 \log n + 4 c_1^2 \log^2 n\right) \sml{\instance}\avg{\instance}.
\]

In both cases, $\E[\text{Cost}_1 \cdot \, \mathbf{1}_{\mathcal{G}_1 \cap \mathcal{T}_i}] \le 14 c_1^2\, \sml{\instance}\avg{\instance} \log^2 n$ for every $i$. Summing over the $\lfloor \log n \rfloor + 1$ values of $i$ and adding the contribution of $\bar{\mathcal{G}}_1$,
\[
\E[\text{Cost}_1] \;\le\; \sum_{i=0}^{\lfloor \log n \rfloor} \E[\text{Cost}_1\cdot \, \mathbf{1}_{\mathcal{G}_1 \cap \mathcal{T}_i}] + \frac{1}{n}
\;=\; \bigo(\sml{\instance}\avg{\instance} \log^3 n).
\]

\paragraph{$\text{Cost}_2$: Disagreements between two cliques in $U$ residing in the same cluster.}
Conditioned on $\sigma$, two distinct cliques $K_a, K_b \subseteq U$ are assigned to the same cluster with probability exactly $1/|S|$. Hence
\[
\E[\text{Cost}_2 \mid \sigma] \;=\; \frac{F}{|S|}, \qquad \text{where } F := \sum_{a<b:\; K_a, K_b \subseteq U} n_a n_b
\]
is determined by $\sigma$. Fix $i$, and let $m$ be as in Claim~\ref{clm:good-event}(c). We split the pairs in $F$ by whether each of the two cliques lies among $K_1, \dots, K_k$. Pairs of cliques among $K_{k+1}, \dots, K_\ell$ contribute at most $\sml{\instance}^2$. Pairs with one clique of each kind contribute at most $\sml{\instance} \cdot B_i\, m$, by Claim~\ref{clm:good-event}(b). Pairs of two distinct cliques among $K_1, \dots, K_k$ contribute at most $B_i^2\, m(m-1)$, again by Claim~\ref{clm:good-event}(b). Since $0 \le m \le Y_i$ are integers, $m(m-1) \le Y_i(Y_i - 1)$. Hence, on $\mathcal{G}_1 \cap \mathcal{T}_i$,
\[
F \;\le\; \sml{\instance}^2 + \sml{\instance}\, B_i\, Y_i + B_i^2\, Y_i(Y_i - 1).
\]
All terms are nonnegative, so we may bound the contribution of $\mathcal{G}_1 \cap \mathcal{T}_i$ to the last two terms by their full expectations. By~\eqref{eq:moments-1},
\begin{eqnarray*}
\E[F \cdot \, \mathbf{1}_{\mathcal{G}_1 \cap \mathcal{T}_i}]
& \le & \sml{\instance}^2 \Pr[\mathcal{T}_i] + \sml{\instance}\, B_i\, \E[Y_i] + B_i^2\, \E[Y_i(Y_i-1)]\\
& \le & \sml{\instance}^2 \Pr[\mathcal{T}_i] + \left(2c_1 \log n + 4 c_1^2 \log^2 n\right) \sml{\instance}^2.
\end{eqnarray*}
Summing over $i$, and using $\sum_i \Pr[\mathcal{T}_i] = 1$, gives $\E[F\cdot \, \mathbf{1}_{\mathcal{G}_1}] = \bigo(\sml{\instance}^2 \log^3 n)$.

Let $\mathcal{G}_2$ be the event that $|S| \ge k-r$. By Lemma~\ref{lem:logn_set_size} with $c = 32$ (i.e., choosing $c_0$ so that $c_0 \log n = 4c \ln n$), $\Pr[\bar{\mathcal{G}}_2] \le n^{-31} + n^{-4} + n^{-7} + n^{-29/3} \le n^{-3}$ for $n \ge 2$. On $\mathcal{G}_2$ we have $F/|S| \le F/(k-r)$, and since the cost is at most $n^2$,
\begin{eqnarray*}
\E[\text{Cost}_2]
& \le & \E\left[\frac{F}{|S|};\, \mathcal{G}_1 \cap \mathcal{G}_2\right] + n^2\left(\Pr[\bar{\mathcal{G}}_1] + \Pr[\bar{\mathcal{G}}_2]\right)
\;\le\; \frac{\E[F\cdot \, \mathbf{1}_{\mathcal{G}_1}]}{k-r} + \frac{2}{n}\\
& = & \bigo\left(\frac{\sml{\instance}^2 \log^3 n}{k-r}\right)
\;=\; \bigo(\sml{\instance}\avg{\instance} \log^3 n),
\end{eqnarray*}
where the last step uses $\sml{\instance} \le \avg{\instance}(k-r)$, which holds since $\avg{\instance}(k-r) = \sum_{j > r} n_j \ge \sum_{j > k} n_j = \sml{\instance}$.

\paragraph{Combining both costs.}
We obtain
\[
\E[\mathrm{Cost}] = \E[\text{Cost}_1] + \E[\text{Cost}_2] = \bigo(\sml{\instance}\avg{\instance} \log^3 n).
\]
By Lemma~\ref{lem:opt-theta}, the optimal cost is $\Omega(\sml{\instance}\avg{\instance})$, so \UP\ is $\bigo(\log^3 n)$-competitive.
\end{proof}

\section{Polylogarithmic Competitiveness for General Graphs}
\label{sec:online}
In this section, we present a simple adaptation of the \textsc{Pivot} algorithm for correlation clustering to obtain a competitive algorithm for correlation $k$-clustering under random arrivals. 
In particular the algorithm, which we call \textsc{BalancedPivot}, achieves cost that is at most a poly-logarithmic factor of optimal cost. 
We present the \textsc{BalancedPivot} algorithm in Section~\ref{sec:online.upper.algorithm}.  In Section~\ref{sec:online.upper.prelim}, we introduce key definitions and notation needed for our analysis.  Finally, Section~\ref{sec:online.upper.analysis} presents our analysis of \textsc{BalancedPivot} and proof of the main theorem.

\subsection{The \mbox{\BP} Algorithm}
\label{sec:online.upper.algorithm}
In the \textsc{BalancedPivot} algorithm, each new vertex is compared with its previously arrived neighbors to determine if any has been selected as a pivot. If so, the vertex is assigned to the same cluster as its first pivot neighbor. If no previously arrived neighbor is a pivot, the vertex is designated as a pivot and placed in a currently least loaded cluster. In this way, it keeps the load of clusters balanced.

\begin{algorithm}
\caption{\textsc{\textsc{BalancedPivot}}$(\instance, k)$}
\begin{algorithmic}[1]
\Require Graph $\instance = (V, E)$, arrival order $\sigma$, maximum clusters $k$
\Ensure Assignment of each vertex to a cluster
\State Initialize $k$ empty clusters $C_1, \ldots, C_k$
\State Initialize empty set of pivots $P$
\For{each arriving vertex $v$ in order $\sigma$}
    \State $N_v \gets$ set of neighbors of $v$ that have already arrived
    \State $P_v \gets$ set of pivots in $N_v$ 
    \If{$P_v \neq \emptyset$}
        \State Let $u$ be the earliest arrived pivot in $P_v$
        \State Assign $v$ to the same cluster as $u$
    \Else

        \State Assign $v$ to the cluster with minimal size
        \State Add $v$ to the set of pivots $P$
    
    \EndIf
\EndFor
\end{algorithmic}
\end{algorithm}

\subsection{Preliminaries: Definitions and notation}
\label{sec:online.upper.prelim}
A major challenge in analyzing any algorithm for correlation $k$-clustering is to obtain suitable lower bounds on the optimal cost.
Section~\ref{sec:cliques.optimal} provides a framework for deriving an asymptotically tight bound on optimal cost for the special case where the graph is a collection of disjoint cliques. 
Our analysis of \textsc{BalancedPivot} proceeds by adapting this framework to the optimal unconstrained correlation clustering $Q=Q_1,\hdots,Q_\ell$.

\begin{definition}\label{def:general-sec-notation}
    We order the $Q_1,\hdots,Q_\ell$ by size: $n_1 \geq n_2 \geq \cdots \geq n_\ell$ where $n_i = |Q_i|$.

    Let $\instance_{clique}$ be the graph instance obtained by turning all $Q_j$ into cliques and deleting all crossing edges between $Q_j$'s.
\end{definition}

\begin{definition}\label{def:general-rand-things}
    Given a graph instance $\instance$, the optimal unconstrained correlation clustering $Q$, and a random arrival ordering $\sigma$, we define the following random variables dependent on $\sigma$.  
    First, as in the overview, we define $\Tau$ as the timestep before the arrival of the $(k+1)$th pivot.  
    (If there are at most $k$ pivots in the random arrival sequence, then $\Tau$ is set to $n$.)  We also define $T_0$ as the arrival time of the $k$th pivot, with $T_0 = n$ if it never appears.  
    Let $U$ denote the set of vertices that are uncovered at time $\Tau$; note that this is the same set as the set of vertices uncovered at time $\Tau_0$ since no new pivots arrive between $T_0$ and $\Tau$.
    
    Consider running the \textsc{Pivot} algorithm on $\instance$, and consider the state of the algorithm mid-way through.
    At any time step, all vertices that have arrived have been clustered, and are considered covered.
    All vertices that have yet to arrive are either adjacent to, and thus covered by, some vertex already marked as a pivot, or they are not adjacent to any current pivot vertices, in which case we consider them not covered.
    For all $j$, either there is some first pivot from $Q_j$ that arrives at some time $t_j$, or there is no pivot that arrives from $Q_j$.  If there is no pivot, let $P_j = \emptyset$;
    otherwise, let $P_j$ be the random set indicating what vertices from $Q_j$ are not covered by time $t_j$.
    Intuitively, this definition will allow us to connect $P_j$ with both the expected pivot cost and the number of non-edges between vertex pairs in $Q_j$.
    Note that $P_j$ is a random subset of $Q_j$ dependent on $\sigma$.
\end{definition}

\begin{definition}[Projected Load]\label{def:projected-load}
The \emph{projected load} of a pivot cluster is the total number of vertices that will eventually join that pivot cluster by the end of the algorithm. Formally, for a pivot $p$ that founds a cluster at any time, the projected load is $|\{v \in V : v \text{ joins the pivot cluster founded by } p\}|$.
\end{definition}

\subsection{Analysis of \BP}\label{sec:online.upper.analysis}

The main result of this paper is the following theorem.

\mainbalancedpivot*

To prove \Cref{thm:BalancedPivotUpper}, in \Cref{sec:online.upper.analysis.OPT} we first present a lower bound on the cost of an optimal clustering, using a combination of the lower bounds we derived in \Cref{sec:cliques.optimal} for the special case of disjoint cliques, and the way the \Pivot\ algorithm operates.

A challenge in analyzing \BP\ is that when a new pivot is placed, it is placed in a cluster with the least load \emph{at that time}, and it is not easy to estimate the size of the pivot clusters associated with the pivots that have been already placed there. 
Yet such an estimate is needed to bound the cost incurred due to the disagreements among the new pivot cluster and the pivot clusters that have already been placed in the cluster. 
\Cref{sec:online.upper.analysis.projected} presents a collection of lemmas that help bound the projected load of pivot clusters and least loaded clusters.  

The third component of the proof is an upper bound on the size of the set of vertices that are left uncovered at time $\Tau$ which we call $U$.  
For disjoint cliques, the proof of \Cref{thm-cliques-random} bounds $|U|$ in terms of $\sml{\instance}$ by comparing it with the number of early arrivals from the fixed set of non-top-$k$ cliques, $\bigcup_{j>k}K_j$. 
For general graphs however, an analogous definition of $\sml{\instance}$ based on an almost-clique decomposition of the graph does not quite suffice since almost-cliques can be split by the Pivot algorithm into multiple pivot clusters. 
This is compounded by the fact that when analyzing \BP\ (as opposed to \UP) we also have to take into account potential correlation between the event that a vertex is not covered by time $\Tau$ and the event that another vertex is covered (and placed in the same cluster).
\Cref{sec:online.upper.analysis.uncovered} provides the tools to bound $|U|$, conditioned on how the almost-cliques get broken apart by the first $k$ pivots.

Finally, \Cref{sec:online.upper.analysis.theorem} combines the above components to prove Theorem~\ref{thm:BalancedPivotUpper}.

\subsubsection{Lower Bound on Optimal cost}
\label{sec:online.upper.analysis.OPT}

\begin{lemma}[Lower Bound on Optimal Cost using Pivot]
\label{lem:sml-avg-bound-pivot}
    Given a graph instance $\instance$ and a random arrival ordering $\sigma$, let $\mathcal{J}_\sigma = J_1,\hdots,J_{\ell_\sigma}$ be the instance obtained by running the pivot algorithm with arrival pattern $\sigma$, then converting all pivot clusters into disjoint cliques.
    Then 
    \[ \Omega\left( \E[\sml{\mathcal{J}_\sigma} \avg{\mathcal{J}_\sigma}] \right) \leq \cost(\opt_k(\instance)) . \]
\end{lemma}

\begin{proof}
To start, fix $\sigma$. Since $\mathcal{J}_\sigma$ is a clique instance, 
\Cref{lem:opt-theta} applies. 
\[ \frac{1}{2}\cdot \sml{\mathcal{J}_\sigma} \avg{\mathcal{J}_\sigma} \leq \cost(\opt_k(\mathcal{J}_\sigma)) \leq \cost(\opt_k(\instance), \mathcal{J}_\sigma) , \]
where $\cost(\opt_k(\instance), \mathcal{J}_\sigma)$ denotes the cost of the clustering $\opt_k(\instance)$ on instance $\mathcal{J}_\sigma$.
We can upper bound $\cost(\opt_k(\instance), \mathcal{J}_\sigma)$ by the sum of two parts:
the cost of $\opt_k(\instance)$ on the original instance $\instance$, and the number of edges added or subtracted when converting from instance $\instance$ to instance $\mathcal{J}_\sigma$. 
Note the second value is exactly the cost of the pivot algorithm on $\instance$ when using arrival pattern $\sigma$.
\[ \cost(\opt_k(\instance), \mathcal{J}_\sigma) \leq \cost(\opt_k(\instance),\instance) + \cost(\mbox{Pivot with arrival pattern } \sigma, \instance) . \]
Because this is true for all $\sigma$, it is also true in expectation. And because Pivot is an expected 3-approximate algorithm, we can derive our final bound.
\begin{align*}
    \frac{1}{2}\cdot \E\left[ \sml{\mathcal{J}_\sigma} \avg{\mathcal{J}_\sigma} \right] &\leq \E\left[ \cost(\opt_k(\instance), \mathcal{J}_\sigma) \right] \\
    & \leq \cost(\opt_k(\instance),\instance) + \E\left[ \cost(\mbox{Pivot with arrival pattern } \sigma, \instance) \right] \\
    &\leq \cost(\opt_k(\instance),\instance) + 3\cdot \cost(\opt_k(\instance),\instance) \\
    & = 4\cdot \cost(\opt_k(\instance)).
\end{align*}
\end{proof}

\subsubsection{Bounds on Projected Load}
For the remainder of our analysis, we assume that $n \ge 2$.  Before bounding the projected load for pivot clusters, we give some notation. For any vertex $v$ which is chosen as a pivot, let $p_t(v)$ denote the size of $v$'s pivot cluster at time $t$. In particular, note that $p_n(v)$ denotes the (final) size of $v$'s pivot cluster after all $n$ vertices have arrived.  The following lemma gives a bound on $p_n(v)$ in terms of $p_t(v)$ for any time $t$.

\label{sec:online.upper.analysis.projected}
\begin{lemma}[Projected Pivot Cluster Size is Predictable]\label{lem:projected_load}
    For any constant $c>0$, the probability that there is a timestep $t$ and a vertex $v$ which arrives at some time $t_v\le t$ for which both $v$ is chosen as a pivot and $p_n(v)\geq \frac{4np_t(v)}{t}+ \frac{2cn\ln n}{t}+1$ is no more than $n^{2-\frac{c}{8}}$.
\end{lemma}
Note that when $c>16$, the failure probability becomes an inverse polynomial in $n$. When $c>32$ is set, the failure probability is less than $n^{-2}$.

Before presenting a proof of Lemma \ref{lem:projected_load}, we give the following lemma which bounds the size of a vertex's pivot cluster in terms of its arrival times.

For a vertex $v$, let $t_v$ be its arrival time. For $t<t_v$, let $P_t$ denote the set of pivot vertices that have been selected among the first $t$ vertex arrivals. 
We let 
\[ R_t(v)=\{u\in N(v):u \mbox{ has not arrived by time } t \mbox{ and } N(u)\cap P_t =\emptyset\} \]
denote the set of neighbors of $v$ which have not arrived by time $t$ and are not incident to any pivot vertex in $P_t$. Thus, $R_t(v)$ denotes the uncovered neighbors of $v$ at time $t$.

First, we note that $R_t(v)$ is non-increasing with $t$. Second, if $v$ becomes a pivot vertex at time $t_v$, then $p_n(v)= 1+|R_{t_{v-1}}(v)|$. 
We first present the following Lemma.

\begin{lemma}\label{lem: survival-prob}
    For any $m\geq 0$ and any vertex $v$,
    \[\Pr[v \text{ is a pivot and } p_n(v)\geq m+1 ~|~ t_v=j] \leq \exp(-\frac{m(j-1)}{n-1}) ~. \]
\end{lemma}

\begin{proof}
The case $n = 1$ is trivial, so we assume $n \ge 2$ in the remainder of this proof.  We will use the following relationship. 
\[ \Pr[v \text{ is a pivot and } p_n(v)\geq m+1 \mid t_v=j] = \Pr[v \text{ is a pivot and } |R_{j-1}(v)| \geq m \mid t_v=j] . \]
Consider the arrival process before time $t_v=j$. 
Note that if $|R_0(v)|<m$, then the Lemma is trivially true.
More so the event that at any time $t<j$ we see $|R_t(v)|< m$ does not contribute to the probability we wish to bound. 
Thus we ignore this event entirely.  Furthermore, if $m > n -j$, then the event $p_n(v) \ge m + 1$ is impossible, so we assume $m \le n-j$.

We prove this Lemma by bounding the probability that at each timestep $t=0,\hdots,j-2$, the next vertex arrival is not a neighbor of $v$ that becomes a pivot.
Note that at any time $t$, the neighbors of $v$ that will become a pivot if they arrive before $v$ is exactly the set $R_t(v)$ by definition.

Suppose we are at time $t=0$ in the arrival process. 
The probability that the next vertex to arrive is a neighbor of $v$ that will become a pivot is exactly $\frac{|R_0(v)|}{n-1}$. 
By our earlier assumption, this is at least $\frac{m}{n-1}$.
Thus, the probability that the next vertex to arrive is {\em not} a neighbor of $v$ that will become a pivot is no more than $1-\frac{m}{n-1}$.

Now consider an arbitrary timestep $t<j-1$.
The probability that the next vertex to arrive is a neighbor of $v$ that will become a pivot is exactly $\frac{|R_t(v)|}{n-1-t} \geq \frac{m}{n-1-t}$.
Thus, the probability that the next vertex to arrive is {\em not} a neighbor of $v$ that will become a pivot is no more than $1-\frac{m}{n-1-t}$.

Since $|R_t(v)|$ is nonincreasing as $t$ increases, the event that $v$ becomes a pivot with at least $|R_{j-1}(v)|\geq m$, requires at least $m$ uncovered neighbors to remain throughout the first $t_v-1=j-1$ arrivals. 
Multiplying out the conditional probabilities gives
\begin{align*}
    \Pr[v \text{ is a pivot and } p_n(v)\geq m+1| t_v=j] 
    &\leq \prod_{t=0}^{j-2}\left(1-\frac{m}{n-1-t}\right)\\
    &\leq \exp\left(-m\sum_{t=0}^{j-2}\frac{1}{n-1-t}\right) \\
    &\leq \exp\left(-\frac{m(j-1)}{n-1}\right)
\end{align*}
which completes the proof. 
\end{proof}

\begin{proof}[Proof of Lemma \ref{lem:projected_load}]
The case $n = 1$ is trivial, so we assume $n \ge 2$ in the remainder of this proof. 
Fix a vertex $v$, time $t$, and condition on time $t_v\le t$. Let $E_{v,t}$ denote the event that $v$ becomes a pivot and $p_n(v)\geq \frac{4np_t(v)}{t}+\frac{2cn\ln n}{t}+1$. We first establish that $\Pr[E_{v,t}|\,t_v=j]\leq e^{-\frac{c\ln n}{8}}$. To this end, we consider two cases. 

\begin{enumerate}
    \item $t\leq 2j$: Since $p_t(v)\geq 1$ in the event that $E_{v,t}$ occurs, we have that $p_n(v)-1\geq \frac{4n}{t}+\frac{2cn\ln n}{t}=\frac{(4+2c\ln n)n}{t}$. Since $t\leq 2j$, we have that $p_n(v)-1\geq \frac{(2+c\ln n)n}{j}$. For $j=1$, $t\in\{1,2\}$ and the bad event is impossible because its size threshold exceeds $n$. For $j\geq 2$, we apply Lemma \ref{lem: survival-prob}, with $m=\frac{(2+c\ln n)n}{j}$ to obtain, 
    \begin{align*}
        \Pr[E_{v,t}|\, t_v=j] &\leq \exp\left(-\frac{(2+c\ln n)n}{j}\cdot \frac{j-1}{n-1}\right) \\
        &\leq \exp\left(-\frac{2+c\ln n}{2}\right) \\
        &\leq \exp\left(-\frac{c\ln n}{8}\right) = n^{-c/8} .   
    \end{align*}
    \item $t>2j$: In this case, rather than directly bounding the probability of event $E_{v,t}$, we instead will show that with high probability $p_t(v)$ must be large enough in comparison to $p_n(v)$.
    
    We condition on the entire arrival sequence through time $t_v=j$. Consider any such sequence under which $v$ becomes a pivot. 
    We know that $p_n(v)=|R_{j-1}(v)|+1$, since all vertices from $R_{j-1}(v)$ will join $v$'s pivot cluster by definition. 
    Conditioned on the arrival sequence through time $t_v=j$, the remaining vertices arrive in a uniformly random order, and consequently, $p_t(v)=1+X$, where $X$ is a random variable sampled from a hyper-geometric distribution, $X \sim Hyp(n-j, |R_{j-1}(v)|, t-j)$. 
    
    If $|R_{j-1}(v)|\leq \frac{2cn\ln n}{t}$, then note that $p_n(v)=|R_{j-1}(v)|+1\leq \frac{2cn\ln n}{t}+1$ and so the probability of event $E_{v,t}$ is $0$. 
    
    Else, $|R_{j-1}(v)|>\frac{2cn\ln n}{t}$. 
    Let $\mu=\E[X]=|R_{j-1}(v)|\cdot\frac{t-j}{n-j}$. Since $t>2j$, then $t-j>\frac{1}{2}t$ and thus we can write $\mu\geq |R_{j-1}(v)|\cdot\frac{t}{2n}$. 
    Applying the lower tail Chernoff bound for hyper-geometric random variables with $\delta=\frac{1}{2}$ and using the relationship $|R_{j-1}(v)|>\frac{2cn\ln n}{t}$, we obtain 
    \begin{align*}
        \Pr\left[X<\frac{|R_{j-1}(v)|t}{4n}\right] &\leq \Pr\left[X<\frac{1}{2}\mu\right] \leq e^{-\mu/8} \\
         & \leq \exp\left(-\frac{|R_{j-1}(v)|t}{16n}\right) \\
         & \leq e^{-\frac{c\ln n}{8}} = n^{-c/8} .
    \end{align*}
Thus with probability at least $1-n^{-c/8}$, $p_t(v)=1+X\geq 1+ \frac{|R_{j-1}(v)|t}{4n}$, and hence $|R_{j-1}(v)|\leq \frac{4np_t(v)}{t}$. 
Since $p_n(v)=|R_{j-1}(v)|+1$, we have that $p_n(v)\leq \frac{4np_t(v)}{t}+1$, and thus, $\Pr[E_{v,t}|\,t_v=j]\leq n^{-c/8}$. 
\end{enumerate}
Both cases establish that $\Pr[E_{v,t}|\,t_v=j]\leq n^{-c/8}$ for every fixed $v,t$ and every possible arrival time $j\le t$. We then derive
\[
\Pr[E_{v,t} \cap \{t_v \le t\}] = \sum_{j=1}^t \Pr[t_v = j] \Pr[E_{v,t}|\,t_v = j]\le n^{-c/8}.
\]
Taking a union bound over all $n$ vertices and all $n$ possible time steps yields the Lemma statement.
\end{proof}

\begin{lemma}[Pivot Cluster Size at an Arbitrary Instant]
\label{lem:pivot cluster size at time t}
    Let $v$ be a pivot that arrives at time $t_v$.  Then, for any time $t \ge t_v$, the probability that $p_t(v)$ exceeds $1 + \frac{2p_n(v)t}{n} + c \ln n$ is at most $n^{-c/3}$ for any constant $c>0$.
\end{lemma}
\begin{proof}
Fix a vertex $v$, and its arrival time $t_v$ and the complete arrival history until time $t_v$ at which $v$ becomes a pivot. Note that this conditioning determines the eventual pivot cluster $C_v$ of $v$, whose size is denoted by $p_n(v)$.  We also note that if $t_v = n$, then $p_t(v) =1$, and the desired estimate is deterministic. So, in the remainder, we assume $t_v < n$.

Conditioned on this history, $p_t(v)=1+X$, where $X\sim Hyp(n-t_v,p_n(v)-1, t-t_v)$. The conditional expectation of $X$ is given by, 
\begin{align*}
\mu_X=\frac{(p_n(v)-1)(t-t_v)}{n-t_v} \leq \frac{(p_n(v)-1)t}{n}\leq \frac{p_n(v)t}{n}.
\end{align*}
since $t \le n$.  

If $\mu_X = 0$, the bound holds vacuously. Otherwise, set $\delta=\frac{1}{\mu_X}\left({\frac{2p_n(v)t}{n}+c\ln n}\right)-1\geq 1$ by the above upper bound on $\mu_X$. Hence, we have,
\begin{align*}
    \Pr[p_t(v)>1+\frac{2p_n(v)t}{n}+c\ln n]&\le\Pr[X\geq\frac{2p_n(v)t}{n}+c\ln n] \\
    &\leq \Pr[X\geq (1+\delta)\mu_X] \\
    &\leq \exp\left(-\frac{\delta\mu_X}{3}\right) \\
    &=\exp\left(-\frac{1}{3}\left(\frac{2p_n(v)t}{n}+c\ln n -\mu_X\right)\right) \\
    &\leq \exp\left(-\frac{1}{3}\left( \frac{p_n(v)t}{n}+c\ln n\right)\right) \\
    &\leq \exp\left(-\frac{c\ln n}{3}\right)\\
    &= n^{-c/3}  
\end{align*}
Since this bound holds for every possible arrival sequence through $t_v$, it also holds without conditioning on a particular arrival sequence, completing the proof.
\end{proof}

Define the \emph{projected load of a cluster} $B$ at time $t$ to be the sum of the projected loads of all the pivot clusters that are in $B$ at time $t$.  For an arbitrary set $S$ of clusters, let $\mpl{S}$ denote the maximum projected load among all clusters in $S$ at time $\Tau$. Also, let $\mathrm{load}^{\mathrm{min}}_{\mathrm{proj}}(t)$ be the projected load of a currently least loaded cluster at time $t$.  
\begin{lemma}[Least Loaded Cluster Has Bounded Projected Load]\label{lem:proj_load_of_least_load}
     With high probability, for every nonempty subset \(S\) of the algorithm clusters present at \(T\), and every integer \(T\le t \le n\),
    \[
    \mathrm{load}^{\mathrm{min}}_{\mathrm{proj}}(t) = \bigo\left(\mu(S)\ln n +\frac{n\ln^2 n}{\Tau} + \frac{|U|\ln n}{|S|}\right).
    \]
   \label{lem:projected_load_least_loaded}
\end{lemma}
\begin{proof}
If $T = n$, then $U = \emptyset$, and at time $n$, the projected loads equal the actual loads.  Hence the claimed bound holds deterministically.  We therefore assume $\Tau < n$ for the remainder of the proof.  Let ${\cal G}_{\mathrm{load}}$ be the intersection of the following three events, where $c$ will be chosen suitably large.
\begin{enumerate}
    \item There does not exist any time $t$ and any vertex $v$ that arrives at or before $t$ such that $v$ is chosen as a pivot and $p_n(v) \ge 4np_t(v)/t + 2cn \ln n/t + 1$.
    \item There does not exist any vertex $v$ and any time $t$ such that $v$ arrives as a pivot at or before time $t$ and $p_t(v)$ exceeds $1 + 2p_n(v)t/n + c \ln n$.
    \item Let $X_t$ be the random variable denoting the number of uncovered vertices that arrive from $U$ by time $t$.  Then, \(X_t\le 2|U|t/n+c\ln n\) holds for every \(t\ge T_0\).
\end{enumerate}

\noindent\textbf{Bounding failure probability of ${\cal G}_{\mathrm{load}}$.}  We first argue that ${\cal G}_{\mathrm{load}}$ holds with probability $1 - O(n^{-3})$, for $c$ chosen suitably large.  For the first sub-event, we utilize Lemma \ref{lem:projected_load} which upper bounds the projected size of any pivot cluster as a function of the size of the pivot cluster at an arbitrary time step $t$. More precisely, with probability at least $1- n^{2-\frac{c}{8}}$, any pivot cluster with size $x$ at time $t$ has a projected size that is upper bounded by $\frac{4nx}{t} + \frac{2cn\ln n}{t} + 1$.   For the second sub-event, we invoke Lemma~\ref{lem:pivot cluster size at time t}, which establishes that with probability at least $1 - n^{-c/3}$, for any pivot $v$, $p_t(v)$ is at most $1 + \frac{2p_n(v)t}{n} + c\ln n$.  Since the number of pivots and the number of time steps are each bounded by $n$, we obtain that the failure probability for the second sub-event is at most $n^{2-c/3}$.

It remains to bound the failure probability of the third sub-event.  Let ${\cal P}_s$ denote the full prefix of events until $T_0$ with $T_0 = s$.  Conditional on ${\cal P}_s$, $X_t$ follows a hypergeometric distribution, $X_t\sim Hyp(n-s, |U|, t-s)$. 
If $s = n$, then the desired claim is deterministic.  Otherwise, $X_t$ can be written as a sum of $(t-s)$ random variables $\sim \{0,1\}$ each with probability $\frac{|U|}{n-s}$.
Then $\lambda_t := \E[X_t \mid {\cal P}_s] = \frac{|U|(t-s)}{n-s} \leq \frac{|U|t}{n}$ since $s \le t \le n$. 

If $\lambda_t$ is 0, then $X_t = 0$ almost surely conditional on ${\cal P}_s$, and the desired bound holds deterministically.  Otherwise, since $X_t \mid {\cal P}_s$ follows a hypergeometric distribution, we can invoke the Chernoff bound (Fact \ref{lem:chernoff}).  Set $\delta = (2|U|t/n + c\ln n)/\lambda_t - 1$; $\delta \lambda_t = 2|U|t/n - \lambda_t + c \ln n \ge c\ln n$ and $\delta \ge 1$ since $\lambda_t \le |U|t/n$.
\[
    \Pr\left[ X_t \geq 2|U|t/n + c\ln n ~|~ {\cal P}_s \right] = \Pr[ X_t \geq (1+\delta) \lambda_t ~|~ {\cal P}_s ] \le e^{-\frac{\delta \lambda_t}{3}} \le e^{-\frac{c \ln n}{3}} = n^{-c/3}.
\]

For each ${\cal P}_s$, taking a union bound over the at most $n$ choices for $t$ gives a failure probability of at most $n^{1-c/3}$.  Averaging over all ${\cal P}_s$ yields the same unconditional bound.  We now choose $c$ sufficiently large so that $n^{2-c/8}+n^{2-c/3}+n^{1-c/3}$ is at most $n^{-3}$, implying that ${\cal G}_{\mathrm{load}}$ holds with probability $1 - n^{-3}$.  

\smallskip
\noindent \textbf{Bounding $\mathrm{load}^{\mathrm{min}}_{\mathrm{proj}}(t)$.} 
We are now ready to upper bound the projected load of a least-loaded cluster.  Let $x_B(t)$ be the current load of output cluster $B$, let $b_B(t)$ be its number of already-assigned pivots, and define
\[
P_B(t)=\sum_{\substack{v\text{ pivot assigned to }B\\ t_v\le t}}p_n(v),
\qquad \mu(S)=\max_{B\in S}P_B(T).
\]

Assume that ${\cal G}_{\mathrm{load}}$ holds.  Summing the estimate from the first sub-event of ${\cal G}_{\mathrm{load}}$ over all pivots already assigned to an output cluster $B$ gives
\begin{align}
P_B(t)
&\le \frac{4n}{t}x_B(t)
+b_B(t)\left(\frac{2cn\ln n}{t}+1\right)\nonumber\\
&\le C_c\frac{n\ln n}{t}x_B(t),
\label{eqn:PB}
\end{align}
where $b_B(t)\le x_B(t)$, and $C_c$ is a fixed constant. For example, $C_c=2c+8$ is sufficient for $n\ge2$ and $1\le t\le n$. 

When $T<n$, every output cluster contains exactly one leading pivot cluster at time $T$. For every nonempty comparison set $S$ and every $t\ge T$, the estimate from the second sub-event of ${\cal G}_{\mathrm{load}}$ bounds the current size of each of its leading pivot clusters by $1+2\mu(S)t/n+c\ln n$.  All vertices subsequently assigned to these output clusters from \emph{other} pivot clusters belong to $U$; at most $X_t$ such vertices have arrived overall. Averaging over $S$, and then comparing with a globally least-currently-loaded output cluster $B_*(t)$, gives
\begin{eqnarray}
x_{B_*(t)}(t)
\le 1+2\mu(S)t/n+c\ln n+\frac{X_t}{|S|}.
\label{eqn:xB*}
\end{eqnarray}
Combine \Cref{eqn:PB}, \Cref{eqn:xB*}, and $X_t\le2|U|t/n+c\ln n$ to obtain
\begin{align*}
P_{B_*(t)}(t)
&\le C_c\frac{n\ln n}{t}
 \left(1+\frac{2\mu(S)t}{n}+c\ln n+
       \frac{2|U|t/n+c\ln n}{|S|}\right)\notag\\
&\le C'_c\left(
 \mu(S)\ln n+\frac{n\ln^2 n}{t}+
 \frac{|U|\ln n}{|S|}\right)\notag\\
&\le C'_c\left(
 \mu(S)\ln n+\frac{n\ln^2 n}{T}+
 \frac{|U|\ln n}{|S|}\right).
\end{align*}
Here $|S|\ge1$ and $t\ge T$; the constants absorb $\ln n\le (\ln^2 n)/\ln2$ for $n\ge2$.

This is the desired conclusion, \emph{simultaneously} for every $S$ and every relevant $t$.  For $T=n$, the only relevant time is $t=n$. The projected and actual loads are the same, and
\[
P_{B_*(n)}(n)=\min_B x_B(n)\le\max_{B\in S}x_B(n)=\mu(S).
\]
Thus that case is deterministic.

\end{proof}

\begin{lemma}[Pivot Clusters with Many Early Arrivals Remain Large]
\label{lem:piv-cluster-projected-load}
    Let $c>6$ be any real and fix a time $t$. Then the probability that any pivot cluster has at least $c \ln n+1$ vertices at time $t$, yet has a final size of less than $\frac{cn\ln n}{2t}$ is no more than $n^{1-c/6}$.
\end{lemma}

\begin{proof}
    We fix a vertex $v$ and condition on the complete arrival history through time $t_v\leq t$  at which $v$ becomes a pivot vertex.  The case $t_v = n$ is trivial since then $t = n$ and $c \ln n + 1$ is at least $\frac{cn\ln n}{2t}$, making the desired claim deterministic.

    Let $N'(v)$ be the set of neighbors of $v$ that eventually join $v$'s pivot cluster. Conditioning on the arrival history through time $t_v$ completely determines $N'(v)$. The number of vertices $Y$, arriving from $N'(v)$ between times $t_v+1$ and $t$, satisfies $Y\sim Hyp(n-t_v, |N'(v)|, t-t_v)$, and thus $\mu_Y=E[Y]=\frac{|N'(v)|(t-t_v)}{n-t_v}\leq \frac{|N'(v)|t}{n}.$ 
    
    Suppose that the final size of $v$'s pivot cluster, i.e. $|N'(v)|+1 < \frac{cn\ln n}{2t}$, and so, $|N'(v)| < \frac{cn\ln n}{2t}$. This implies that $\mu_Y<\frac{c\ln n}{2}$. Since $N'(v)$ excludes the pivot $v$ itself, we have $p_t(v)=Y+1$. Hence, if $p_t(v)\ge c\ln n+1$, then $Y\ge c\ln n$.

    If $\mu_Y = 0$, then $Y = 0$ almost surely, and the desired claim trivially holds.  Otherwise, by a hypergeometric Chernoff bound (the first part of Fact \ref{lem:chernoff}), setting $\delta = \frac{c \ln n}{\mu_Y} - 1\geq 1$ we obtain, 
    \begin{align*}
\Pr[Y \ge c \ln n] &= \Pr[Y\geq (1+\delta) \mu_Y] \leq e^{-\delta\mu_Y/3} = e^{-\frac{c\ln n-\mu_Y}{3}} \leq e^{-\frac{c\ln n}{6}} = n^{-c/6}.
    \end{align*}
Taking a union bound over all $n$ possible pivot vertices yields the Lemma statement.
    
\end{proof}

\subsubsection{Upper Bound on Uncovered Vertices from Almost-Cliques with no Pivot by Time $\Tau$}
\label{sec:online.upper.analysis.uncovered}

We will next analyze the set of uncovered vertices at time $\Tau$ from $Q_j$, $1 \le j \le k$, such that no pivot from $Q_j$ arrived by time $\Tau$. 
Recall that $Q_1, \ldots, Q_\ell$ denote the set of clusters (in order of non-increasing size) determined by an optimal unconstrained correlation clustering.  We refer to each $Q_j$ as an almost-clique since every vertex in $Q_j$ has edges to at least half of the other vertices in $Q_j$. 

\begin{lemma}[Almost-Cliques Without Early Pivots Have Few Uncovered Vertices]\label{lem: uncovered}
Fix an almost clique $Q_j$.  
Let $Q^{u}_j \subseteq Q_j$ denote the set of uncovered vertices from $Q_j$ at time $2^i$. 
Then for any constant $c>0$, with probability at least $1-\frac{1}{n^c}$, either a pivot arrives from $Q_j$ by time $2^i$ or $|Q^{u}_j| \leq \frac{cn\ln n}{2^i}$.

\end{lemma}

\begin{proof}
    Fix an integer $\kappa$. 
    Define $B_t$ to be the event that the first pivot from $Q_j$ arrives after $t$ and the number of uncovered vertices from $Q_j$ by time $t$ is at least $\kappa$. If $\kappa > n$, then $\Pr[B_t] = 0$ for all $t$.  In the remainder, we assume $\kappa \le n$.  Conditioned on $B_t$, the next arrival is an uncovered vertex of $Q_j$ with probability at least $\kappa/(n-t) \ge \kappa/n$. Such an arrival becomes a pivot from $Q_j$, causing $B_{t+1}$ to fail. In addition, event $B_{t+1}$ cannot occur unless $B_t$ also occurs. 
    Therefore, 
    \[ \Pr[B_{t+1}] = \Pr[B_{t+1} \mid B_t]\cdot \Pr[B_t] \le (1 - \kappa/n)\cdot \Pr[B_t] . \]  
    We can chain the above inequality.
    Setting $\kappa = \lceil \frac{cn\ln n}{2^i}\rceil$ and iterating over $t$ from 0 through $2^i-1$ yields that the probability that the first pivot from $Q_j$ arrives after $2^i$ and $|Q_j^u|$ is at least $\frac{cn\ln n}{2^i}$ is at most
    \[
       \Pr[B_{2^i}] \le \left(1 - \frac{\kappa}{n}\right)^{2^i} \Pr [B_0] \le e^{-\kappa 2^i/n} \le n^{-c}. 
    \]
\end{proof}

\subsubsection{Proof of \Cref{thm:BalancedPivotUpper}} 
\label{sec:online.upper.analysis.theorem}

\mainbalancedpivot*

The general framework of this proof relies heavily on 1) partitioning vertex arrival patterns $\sigma$ into groups with certain properties, then showing that the contribution to the expected cost of \textsc{BalancedPivot} of each group is within a $\polylog(n)$-factor of the corresponding contribution to our lower bounds from \Cref{sec:online.upper.analysis.OPT}, and 2) showing that with high probability at least $1-\frac{1}{n^3}$, arrival pattern $\sigma$ has nice properties. Since the worst-case cost of any $k$-clustering is $\bigo(n^2)$, the contribution of arrival patterns $\sigma$ without the nice properties is negligible to the expected cost.

\begin{proof}
We use the term \emph{pivot cluster} to refer to a cluster formed by the \textsc{Pivot} algorithm underlying \textsc{\textsc{BalancedPivot}}. 
Note that by construction, a pivot cluster is never split by \textsc{\textsc{BalancedPivot}}; however, multiple pivot clusters may be placed in a single cluster. Therefore, the cost of the algorithm can be divided into three categories: (i) non-edges within pivot clusters; (ii) edges between pivot clusters assigned to different clusters; (iii) non-edges between pivot clusters assigned to the same cluster.

The first two costs, taken together, are at most the cost of Pivot and hence no more than $3\cdot \cost(\opt_k(\instance))$ in expectation. We focus on deriving an upper bound on the third type of cost. 
For any given clustering $\mathcal{C} = C_1,\hdots, C_k$ returned by \textsc{\textsc{BalancedPivot}},
we upper bound cost of type (iii) by
\[
\sum_{i=1}^{k} \sum_{\substack{X, X' \subseteq C_i \\ X \neq X'\\\mathrm{unordered}}} |X| \cdot |X'|,
\]
where the inner sum is taken over all pairs of distinct pivot clusters $X, X'$ in cluster $C_i$. 
We will prove that the expected total cost of this type is $O(\opt_k(\instance)\log^4(n))$.

\smallskip    
\noindent {\bf Grouping of random arrival patterns.} The case $\Tau = n$ is trivial since there are at most $k$ pivots and the merging cost is zero; henceforth we assume $\Tau < n$.  We divide all random arrival patterns into $\lceil \log_2 n \rceil$ groups depending on whether $\Tau$ is in $[2^i, 2^{i+1})$ for $0 \le i \le \lfloor \log_2 n \rfloor$. 
For each integer $i\ge 0$ with $2^i<n$, write
\[
    L_i=2^i,\qquad H_i=\min\{2^{i+1},n\},\qquad
    \theta_i=\frac{n\ln n}{L_i}.
\]
Fix $i$ such that $\Tau \in [L_i, H_i)$.  We refer to group $i$ as the collection of all random arrival patterns with $\Tau \in [L_i, H_i)$.

\smallskip
\noindent {\bf Pivot cluster size estimation.}
Let $\beta$ and $\gamma$ be constants to be defined later. 
Let $S_1$ denote the set of early pivot clusters that have at most $\beta \ln n$ vertices at time $\Tau$ and let $S_2$ denote the set of pivot clusters that arrived by time $\Tau$ but are not in $S_1$; that is, each pivot cluster in $S_2$ has more than $\beta\ln n$ vertices.

We would like for the following to be true: (a) pivot clusters which do not have a pivot arrive by time $L_i$ have projected load no more than $2\gamma \theta_i + 1 \le c_{\mathrm{late}}\theta_i$; (b) pivot clusters with load at most $\beta\ln n$ by time $\Tau$, have projected load no more than $4\beta \theta_i + 4\gamma \theta_i +1$, and (c) pivot clusters with load at least $\beta\ln n$ by time $H_i$, have projected load at least $\beta \theta_i/8 = c_{\mathrm{heavy}} \theta_i$. 

For property~(a), we use \Cref{lem: survival-prob}: conditional on $t_v = j > L_i$, setting $m = 2\gamma \theta_i$ yields the failure probability for a fixed vertex to be at most
\[
e^{-\frac{2\gamma \theta_i L_i}{n-1}} \le n^{-2\gamma}.
\]
Averaging over $j$ and applying a union bound over vertices, we obtain that $\Pr[\mbox{(a) fails}]$ is at most $n^{1 - 2\gamma}$.  
For property~(b), we apply \Cref{lem:projected_load}, $\Pr[\mbox{(b) fails}]$ is at most 
$n^{2-\gamma /8}$.
For property~(c), we apply \Cref{lem:piv-cluster-projected-load} with parameter $c = \beta/2$ at time $H_i$.  For large enough $\beta$ and $n \ge 2$, a count at least $\beta \ln n$ is at least $(\beta/2)\ln n + 1$.  Consequently, 
the final pivot cluster size bound in \Cref{lem:piv-cluster-projected-load} is at least $\beta n \ln n/(4H_i) \ge \beta \theta_i/8$.  So, we obtain that $\Pr[\mbox{(c) fails}]\leq \frac{n\log_2 n}{n^{\beta/12}}$. Let ${\cal G}_i^{\mathrm{size}}$ denote the event that properties (a), (b), and (c) hold.  It thus follows that 
$\Pr[{\cal G}_i^{\mathrm{size}}]$ is at least $1-n^{1-2\gamma}-n^{2-\gamma /8}-\frac{n\log_2 n}{n^{\beta/12}} = 1 - O(n^{-3})$ for constants chosen suitably large.

Setting $2\gamma + 1 \leq \beta/8$, we obtain $\beta \theta_i/8 > 2\gamma \theta_i +1$. One explicit consistent choice of constants is
\[
\gamma = 64, \quad \beta = 2048, \quad c_{\mathrm{late}} = 129, \quad c_{\mathrm{heavy}} = 256, c_2 = 512.
\]
We thus obtain that $c_{\mathrm{heavy}} > c_{\mathrm{late}}$; so, we get a total separation between the projected load of pivot clusters in $S_2$ (those which are large by time $\Tau$) and the projected load of pivot clusters which have not yet arrived by time $\Tau$.

We now divide all random arrival patterns in group $i$ into two categories based on the size of $U$, the set of uncovered vertices at time $\Tau$.  
In the first case, our analysis compares the algorithm's cost with the lower bound established in \Cref{lem:sml-avg-bound-pivot}.
In the second case, our analysis compares the algorithm's cost with the lower bound using the counts of the missing edges $e_m$ and cross-edges $e_c$ in an optimal unconstrained correlation clustering solution and the parameters associated with instance $\instance_{\mathrm{clique}}$.

\begin{table}[h!]
\centering
\begin{tabular}{| c | p{0.75\linewidth} |} 
 \hline
 Var/Set Name & Description \\ 
 \hline\hline
 $\Tau$ & Timestep before the $(k+1)$-th pivot arrives. Depends on $\sigma$. \\
 \hline
 $U$ & Vertices which don't arrive by $\Tau$ and are not adjacent to any pivots that arrive by $\Tau$. Depends on $\sigma,\Tau$. \\
 \hline
 $S_1$ & Pivot clusters which arrive by time $\Tau$ but have at most $\beta\ln n$ vertex arrivals by time $\Tau$. Depends on $\sigma,\Tau$. \\
 \hline
 $S_2$ & Pivot clusters which have more than $\beta\ln n$ vertex arrivals by time $\Tau$. Depends on $\sigma,\Tau$. \\
\hline

 $\mathcal{J}_\sigma$ & ``Pivot''-based clique instance used in lower bound \Cref{lem:sml-avg-bound-pivot}. Depends on $\sigma$. \\
 \hline
 $s(\mathcal{J}_\sigma)$ & Number of vertices in $\mathcal{J}_\sigma$ which do not lie in the $k$ largest cliques. Depends on $\sigma$ and $\mathcal{J}_\sigma$. \\
 \hline
 $r(\mathcal{J}_\sigma)$ & Number of cliques in $\mathcal{J}_\sigma$ so large they should optimally be placed into their own cluster, see \Cref{def:threshold}. Depends on $\sigma$ and $\mathcal{J}_\sigma$. \\
 \hline
 $a(\mathcal{J}_\sigma)$ & Assuming the $r(\mathcal{J}_\sigma)$  largest cliques in $\mathcal{J}_\sigma$  are placed into their own clusters, the average load of the remaining clusters, see \Cref{def:threshold}. Depends on $\sigma$, $\mathcal{J}_\sigma$, and  $r(\mathcal{J}_\sigma)$. \\
 \hline
\end{tabular}
\caption{Quick reference guide to variable and set names used in proof of \Cref{thm:BalancedPivotUpper}.}
\label{tab:1}
\end{table}

\medskip
\noindent\textbf{Case 1: Large $U$.}  This is the case with $|U| > c_2 |S_1| \theta_i$, where
\begin{equation}
    c_2\ge\max\{2c_{\mathrm{late}},c_{\mathrm{heavy}}\},
    \label{eq:large-u-wf-constant}
\end{equation}
and we use the same $c_2$ in the complementary Small-$U$ condition.
Recall that we are considering arrival orders for which $L_i\le \Tau <H_i$.
Orders with at most $k$ pivots have no merging cost and are omitted;
for these orders we set $\Tau=n$.
Let $\mathcal P_\sigma$ be the family of final pivot clusters, and let
$\tau(X)$ be the arrival time of the pivot of $X\in\mathcal P_\sigma$.
Define
\[
    \mathcal P_{\le \Tau}=\{X\in\mathcal P_\sigma:\tau(X)\le \Tau\},
    \qquad
    \mathcal P_{>\Tau}=\{X\in\mathcal P_\sigma:\tau(X)>\Tau\}.
\]
Exactly $k$ pivot clusters belong to $\mathcal P_{\le \Tau}$, one in each
algorithm cluster at time $\Tau$.  The late pivot clusters form a disjoint
partition of $U$, so
\begin{equation}
    |U| =\sum_{X\in\mathcal P_{>\Tau}}|X|.
    \label{eq:large-u-wf-late-mass}
\end{equation}
Indeed, a vertex covered at $\Tau$ belongs to a pivot cluster founded by $\Tau$,
whereas an uncovered vertex can only join a pivot cluster founded later.

Recall that $S_1$ consists of the early pivot clusters with at most
$\beta\ln n$ arrivals by $\Tau$, and $S_2=\mathcal P_{\le \Tau}\setminus S_1$.
Put $m=|S_1|$, so $|S_2|=k-m$.  Recall that the preceding pivot-cluster size estimates yield that for 
fixed constants $0<c_{\mathrm{late}}<c_{\mathrm{heavy}}$ and an
unconditional good event $\mathcal G_i^{\mathrm{size}}$ such that
$\Pr[(\mathcal G_i^{\mathrm{size}})^c]=O(n^{-3})$, we have the following claims under the event 
$\{L_i\le \Tau<H_i\}\cap\mathcal G_i^{\mathrm{size}}$, 
\begin{align}
    |X|&\le c_{\mathrm{late}}\theta_i
       &&\text{for every }X\in\mathcal P_{>\Tau},
       \label{eq:large-u-wf-late-size}\\
    |X|&\ge c_{\mathrm{heavy}}\theta_i
       &&\text{for every }X\in S_2.
       \label{eq:large-u-wf-heavy-size}
\end{align}
The constants in the preceding estimates are chosen so that this strict
separation holds.  Define
\[
    \mathcal L_i
    =\{L_i\le \Tau<H_i\}\cap\{|U| >c_2m\theta_i\}.
\]
For the deterministic argument below, fix an arrival order
$\sigma\in\mathcal L_i\cap\mathcal G_i^{\mathrm{size}}$.
In particular, $\Tau<n$, $|U|>0$, and $\mathcal P_\sigma$ has more than $k$
nonempty pivot clusters.

\medskip
\noindent\emph{The water-filling characterization of $a({\cal J}_{\sigma})$.} A major part of the proof is to derive a suitable lower bound on $a({\cal J}_{\sigma})$.  Toward this end, we first provide a water-filling characterization of $a({\cal J}_{\sigma})$ following the process of Definition~\ref{def:threshold}.  Let $w_1\ge\cdots\ge w_q>0$ be the final pivot-cluster sizes, with $q>k$,
and put
\[
    F(x)=\sum_{h=1}^q\min\{w_h,x\},\qquad x>0.
\]
The function
\[
    \frac{F(x)}{x}=\sum_{h=1}^q\min\{w_h/x,1\}
\]
is continuous and nonincreasing.  It equals $q>k$ for $0<x\le w_q$,
is strictly decreasing for $x>w_q$, and tends to zero as $x$ tends to
infinity.  Hence there is a unique $a_*>0$ such that $F(a_*)=ka_*$.
Let $r_* = |\{h:w_h\ge a_*\}|$.

Necessarily $r_*<k$: otherwise the first $k$ terms contribute $ka_*$
and at least one further term is positive, contradicting $F(a_*)=ka_*$.
Thus
\[
    \sum_{h>r_*}w_h=(k-r_*)a_*,\qquad w_{r_*+1}<a_*.
\]
For every $0\le j<r_*$,
\[
    \frac{\sum_{h>j}w_h}{k-j}
    =\frac{\sum_{h=j+1}^{r_*}w_h+(k-r_*)a_*}{k-j}
    \le w_{j+1}.
\]
Therefore the peeling rule of Definition~\ref{def:threshold} cannot stop before $r_*$,
and it stops at $r_*$.  Consequently, with
$a=a({\cal J}_\sigma)$ and $r=r({\cal J}_\sigma)$,
\begin{equation}
    F(a)=ka,\qquad
    r=|\{X\in\mathcal P_\sigma:|X|\ge a\}|<k,
    \qquad
    \sum_{\substack{X\in\mathcal P_\sigma\\|X|<a}}|X|=(k-r)a.
    \label{eq:large-u-wf-characterization}
\end{equation}
Clusters of size exactly $a$ are counted among the peeled clusters,
consistent with the strict stopping inequality in Definition~3.3.
In particular, the characterization implies the threshold test
\begin{equation}
    F(b)>kb\quad\Longrightarrow\quad a>b
    \qquad(b>0).
    \label{eq:large-u-wf-threshold-test}
\end{equation}

\medskip
\noindent\emph{A lower bound on $a({\cal J}_\sigma)$.}
Set $b=c_{\mathrm{heavy}}\theta_i$.
Every cluster in $S_2$ contributes $b$ to $F(b)$.
Every late cluster has size at most $c_{\mathrm{late}}\theta_i<b$, so the
late clusters contribute their full mass $|U|$.
The contribution of $S_1$ is nonnegative.  Therefore
\[
    F(b)\ge(k-m)b+|U|
         >(k-m)b+c_2m\theta_i
         \ge kb.
\]
The strict inequality also holds when $m=0$, because then $|U|>0$.
By~\eqref{eq:large-u-wf-threshold-test},
\begin{equation}
    a({\cal J}_\sigma)>c_{\mathrm{heavy}}\theta_i
    =\Omega\!\left(\frac{n\ln n}{2^i}\right).
    \label{eq:large-u-wf-a-lower}
\end{equation}
Moreover, every late pivot cluster has size strictly less than $a$.
Thus every peeled cluster of ${\cal J}_\sigma$ is an early pivot cluster.

\medskip
\noindent\emph{The comparison set $A$.}
Define $A$ to be the set of algorithm clusters whose unique pivot cluster
at time $\Tau$ has final size strictly less than $a$.
This definition concerns only the $k$ algorithm clusters present at $\Tau$;
it does not include the late pivot clusters as additional algorithm
clusters.
Since all $r$ peeled pivot clusters are early,
\begin{equation}
    |A|=k-r\ge1,\qquad \mu(A)<a.
    \label{eq:large-u-wf-A-size}
\end{equation}
Here $\mu(A)$ is the maximum projected load in $A$ at time $T$, when
each algorithm cluster contains exactly one pivot cluster.
Writing $X_B$ for that pivot cluster in $B\in A$, we obtain the following from the last identity
in~\eqref{eq:large-u-wf-characterization}
\begin{equation}
    \sum_{B\in A}|X_B|+|U|=a|A|, \qquad \frac{|U|}{|A|}\le a.
        \label{eq:large-u-wf-average}
\end{equation}

\noindent\emph{Comparing $|U|$ with $s({\cal J}_\sigma)$.}
Every cluster in $S_2$ is larger than every late cluster.
Since there are exactly $k$ early clusters, all $k-m$ members of $S_2$
belong to the $k$ largest pivot clusters of ${\cal J}_\sigma$.
Consequently, at most $m$ late pivot clusters belong to those $k$ largest
clusters.  Their combined size is at most
$m c_{\mathrm{late}}\theta_i$, and hence
\begin{equation}
    s({\cal J}_\sigma)\ge |U|-mc_{\mathrm{late}}\theta_i>\frac{|U|}{2}.
    \label{eq:large-u-wf-s-comparison}
\end{equation}
The last inequality follows from $|U| >c_2m\theta_i$ and
$c_2\ge2c_{\mathrm{late}}$.
When $m=0$, all early clusters are among the $k$ largest and
$s({\cal J}_\sigma)=|U|$.

\medskip
\noindent\emph{Projected load and merging cost.}
Use the simultaneous, unconditional high-probability form of the
projected-load estimate in Lemma~\ref{lem:proj_load_of_least_load}.
Specifically, let $\mathcal G_i^{\mathrm{load}}$ have complement probability
$O(n^{-3})$, and suppose that on this event, for every nonempty set
$\mathcal S$ of the algorithm clusters present at $\Tau$ and all
$\Tau\le t<n$,
\begin{equation}
    \operatorname{load}_{\mathrm{proj}}^{\min}(t)
    \le C_{\mathrm{load}}\left(
       \mu(\mathcal S)\ln n
       +\frac{n\ln^2 n}{\Tau}
       +\frac{|U| \ln n}{|\mathcal S|}
    \right).
    \label{eq:large-u-wf-load-input}
\end{equation}
The left side is the projected load of a least currently loaded
algorithm cluster, after $t$ arrivals: it counts the full sizes of pivot
clusters already assigned to that algorithm cluster by $t$, not the
sizes of pivot clusters assigned there in the future.
The estimate is simultaneous in $\mathcal S$ and $t$, so it may be
applied pointwise to the set $A$, even though $A$ depends on $\sigma$.
No conditioning on $A$ is used.

Fix $\sigma\in\mathcal L_i\cap\mathcal G_i^{\mathrm{size}}
\cap\mathcal G_i^{\mathrm{load}}$.
Since $T\ge L_i$,~\eqref{eq:large-u-wf-a-lower} gives
\[
    \frac{n\ln^2 n}{\Tau}
    \le\theta_i\ln n
    \le\frac{a}{c_{\mathrm{heavy}}}\ln n.
\]
Applying~\eqref{eq:large-u-wf-load-input} with $\mathcal S=A$ and
using~\eqref{eq:large-u-wf-A-size}--\eqref{eq:large-u-wf-average}, we obtain
\begin{equation}
    \operatorname{load}_{\mathrm{proj}}^{\min}(t)
    \le C_{\mathrm{load}}
       \left(2+c_{\mathrm{heavy}}^{-1}\right)a\ln n
    =: C_*a\ln n
    \qquad(\Tau\le t<n).
    \label{eq:large-u-wf-load-conclusion}
\end{equation}
When the pivot of a late cluster $X$ arrives at time $\tau(X)>\Tau$, it is
assigned to a least-loaded algorithm cluster in the state at
$t=\tau(X)-1\ge \Tau$.
The number of disagreements between $X$ and earlier pivot clusters in
that algorithm cluster is at most $|X|$ times its projected load in that
state, and therefore at most $C_*|X|a\ln n$ by \eqref{eq:large-u-wf-load-conclusion}.
Every merging disagreement is counted once by charging it to the later
of its two pivot clusters; distinct early pivot clusters occupy distinct
algorithm clusters and incur no such cost with one another.
Summing over the late clusters and using
\eqref{eq:large-u-wf-late-mass} and
\eqref{eq:large-u-wf-s-comparison},
\begin{equation}
    \mathrm{Cost}_{\mathrm{merge}}
    \le C_*|U|a\ln n
    \le2C_*s({\cal J}_\sigma)a({\cal J}_\sigma)\ln n.
    \label{eq:large-u-wf-cost}
\end{equation}

\medskip
\noindent\emph{Taking expectations.}
Let $\mathcal G_i=\mathcal G_i^{\mathrm{size}}
\cap\mathcal G_i^{\mathrm{load}}$.
Its failure probability is $O(n^{-3})$ unconditionally, and the actual
merging cost is always at most $n^2$.
Thus~\eqref{eq:large-u-wf-cost} implies
\begin{align}
    \mathbb E\!\left[
       \mathrm{Cost}_{\mathrm{merge}}\mathbf1_{\mathcal L_i}
    \right]
    &\le2C_*\ln n\,
       \mathbb E\!\left[
          s({\cal J}_\sigma)a({\cal J}_\sigma)\mathbf1_{\mathcal L_i}
       \right]
       +n^2\Pr(\mathcal G_i^c)\notag\\
    &\le O(\ln n)\,
       \mathbb E\!\left[
          s({\cal J}_\sigma)a({\cal J}_\sigma)\mathbf1_{\mathcal L_i}
       \right]+O(1/n).
       \label{eq:large-u-wf-expectation}
\end{align}
Set $s({\cal J}_\sigma)a({\cal J}_\sigma)=0$ on orders having at most $k$ pivots.
The events $\mathcal L_i$ are disjoint, so summing over the
$O(\log n)$ time intervals yields
\begin{align*}
    \mathbb E\!\left[
       \mathrm{Cost}_{\mathrm{merge}}
       \mathbf1_{\bigcup_i\mathcal L_i}
    \right]
    &\le O(\ln n)\,
       \mathbb E[s({\cal J}_\sigma)a({\cal J}_\sigma)]
       +O(\log n/n)\\
    &\le O(\log n)\,\operatorname{OPT}_k(I)
       +O(\log n/n),
\end{align*}
where the last step uses Lemma~\ref{lem:sml-avg-bound-pivot}.
If $\operatorname{OPT}_k(\instance)>0$, integrality allows the additive term to
be absorbed into the multiplicative bound.
If $\operatorname{OPT}_k(\instance)=0$, the input is a union of at most $k$
disjoint cliques, \textsc{Pivot} produces at most $k$ pivots, and the merging cost
is identically zero.
This completes the Large-$U$ case.

\begin{figure}
    \centering
    \tikzset{every picture/.style={line width=0.75pt}}        

\begin{tikzpicture}[x=0.75pt,y=0.75pt,yscale=-1,xscale=1]

\draw  [fill={rgb, 255:red, 17; green, 62; blue, 123 }  ,fill opacity=1 ] (261,110.8) .. controls (261,107.6) and (263.6,105) .. (266.8,105) -- (284.2,105) .. controls (287.4,105) and (290,107.6) .. (290,110.8) -- (290,228.2) .. controls (290,231.4) and (287.4,234) .. (284.2,234) -- (266.8,234) .. controls (263.6,234) and (261,231.4) .. (261,228.2) -- cycle ;
\draw  [fill={rgb, 255:red, 17; green, 62; blue, 123 }  ,fill opacity=1 ] (221,120.8) .. controls (221,117.6) and (223.6,115) .. (226.8,115) -- (244.2,115) .. controls (247.4,115) and (250,117.6) .. (250,120.8) -- (250,228.2) .. controls (250,231.4) and (247.4,234) .. (244.2,234) -- (226.8,234) .. controls (223.6,234) and (221,231.4) .. (221,228.2) -- cycle ;
\draw  [fill={rgb, 255:red, 17; green, 62; blue, 123 }  ,fill opacity=1 ] (321,149.8) .. controls (321,146.6) and (323.6,144) .. (326.8,144) -- (344.2,144) .. controls (347.4,144) and (350,146.6) .. (350,149.8) -- (350,228.2) .. controls (350,231.4) and (347.4,234) .. (344.2,234) -- (326.8,234) .. controls (323.6,234) and (321,231.4) .. (321,228.2) -- cycle ;
\draw  [fill={rgb, 255:red, 17; green, 62; blue, 123 }  ,fill opacity=1 ] (371,129.8) .. controls (371,126.6) and (373.6,124) .. (376.8,124) -- (394.2,124) .. controls (397.4,124) and (400,126.6) .. (400,129.8) -- (400,228.2) .. controls (400,231.4) and (397.4,234) .. (394.2,234) -- (376.8,234) .. controls (373.6,234) and (371,231.4) .. (371,228.2) -- cycle ;
\draw  [fill={rgb, 255:red, 17; green, 62; blue, 123 }  ,fill opacity=1 ] (431,139.8) .. controls (431,136.6) and (433.6,134) .. (436.8,134) -- (454.2,134) .. controls (457.4,134) and (460,136.6) .. (460,139.8) -- (460,228.2) .. controls (460,231.4) and (457.4,234) .. (454.2,234) -- (436.8,234) .. controls (433.6,234) and (431,231.4) .. (431,228.2) -- cycle ;
\draw  [color={rgb, 255:red, 128; green, 128; blue, 128 }  ,draw opacity=1 ][fill={rgb, 255:red, 155; green, 155; blue, 155 }  ,fill opacity=1 ] (481,180.8) .. controls (481,177.6) and (483.6,175) .. (486.8,175) -- (504.2,175) .. controls (507.4,175) and (510,177.6) .. (510,180.8) -- (510,228.2) .. controls (510,231.4) and (507.4,234) .. (504.2,234) -- (486.8,234) .. controls (483.6,234) and (481,231.4) .. (481,228.2) -- cycle ;
\draw  [color={rgb, 255:red, 128; green, 128; blue, 128 }  ,draw opacity=1 ][fill={rgb, 255:red, 155; green, 155; blue, 155 }  ,fill opacity=1 ] (521,180.8) .. controls (521,177.6) and (523.6,175) .. (526.8,175) -- (544.2,175) .. controls (547.4,175) and (550,177.6) .. (550,180.8) -- (550,228.2) .. controls (550,231.4) and (547.4,234) .. (544.2,234) -- (526.8,234) .. controls (523.6,234) and (521,231.4) .. (521,228.2) -- cycle ;
\draw  [color={rgb, 255:red, 128; green, 128; blue, 128 }  ,draw opacity=1 ][fill={rgb, 255:red, 155; green, 155; blue, 155 }  ,fill opacity=1 ] (580.92,180.85) .. controls (580.91,177.64) and (583.5,175.04) .. (586.7,175.03) -- (604.1,174.97) .. controls (607.31,174.96) and (609.91,177.55) .. (609.92,180.75) -- (610.08,228.15) .. controls (610.09,231.36) and (607.5,233.96) .. (604.3,233.97) -- (586.9,234.03) .. controls (583.69,234.04) and (581.09,231.45) .. (581.08,228.25) -- cycle ;
\draw    (466,94) -- (467,234) ;
\draw   (481,244) .. controls (481,248.67) and (483.33,251) .. (488,251) -- (535.5,251) .. controls (542.17,251) and (545.5,253.33) .. (545.5,258) .. controls (545.5,253.33) and (548.83,251) .. (555.5,251)(552.5,251) -- (603,251) .. controls (607.67,251) and (610,248.67) .. (610,244) ;
\draw   (363,246) .. controls (363,250.67) and (365.33,253) .. (370,253) -- (405,253) .. controls (411.67,253) and (415,255.33) .. (415,260) .. controls (415,255.33) and (418.33,253) .. (425,253)(422,253) -- (460,253) .. controls (464.67,253) and (467,250.67) .. (467,246) ;
\draw   (221,245) .. controls (221,249.67) and (223.33,252) .. (228,252) -- (275,252) .. controls (281.67,252) and (285,254.33) .. (285,259) .. controls (285,254.33) and (288.33,252) .. (295,252)(292,252) -- (342,252) .. controls (346.67,252) and (349,249.67) .. (349,245) ;
\draw    (219,290) -- (613,289.5) ;
\draw [shift={(615,289.5)}, rotate = 179.93] [color={rgb, 255:red, 0; green, 0; blue, 0 }  ][line width=0.75]    (10.93,-3.29) .. controls (6.95,-1.4) and (3.31,-0.3) .. (0,0) .. controls (3.31,0.3) and (6.95,1.4) .. (10.93,3.29)   ;

\draw (298,189) node [anchor=north west][inner sep=0.75pt]   [align=left] {...};
\draw (408,189) node [anchor=north west][inner sep=0.75pt]   [align=left] {...};
\draw (559,190) node [anchor=north west][inner sep=0.75pt]   [align=left] {...};
\draw (400,265) node [anchor=north west][inner sep=0.75pt]   [align=left] {$S_1$};
\draw (272,263) node [anchor=north west][inner sep=0.75pt]   [align=left] {$S_2$};
\draw (540,262) node [anchor=north west][inner sep=0.75pt]   [align=left] {$U$};
\draw (461,68) node [anchor=north west][inner sep=0.75pt]   [align=left] {$\Tau$};
\draw (419,292.75) node [anchor=north west][inner sep=0.75pt]   [align=left] {${\cal J}_\sigma$};

\end{tikzpicture}
        \caption{Timeline showing the arrival pattern under permutation $\sigma$. 
        By time $\Tau$, exactly $k$ pivots have arrived, partitioning vertices into three categories: (1) vertices in $S_1$ (clusters with $O(\log n)$ vertices at time $\Tau$), (2) vertices in $S_2$ (pivot clusters that arrived by time $\Tau$ and are not in $S_1$), and (3) vertices in $U$ (uncovered vertices from pivot clusters that have not yet had their first vertex arrive by time $\Tau$).  Any vertex arriving after time $\Tau$ is either assigned to a cluster using an earlier pivot neighbor or becomes a pivot and is assigned to a cluster that has least current load.}

    \label{fig:sigma_figure}
\end{figure}
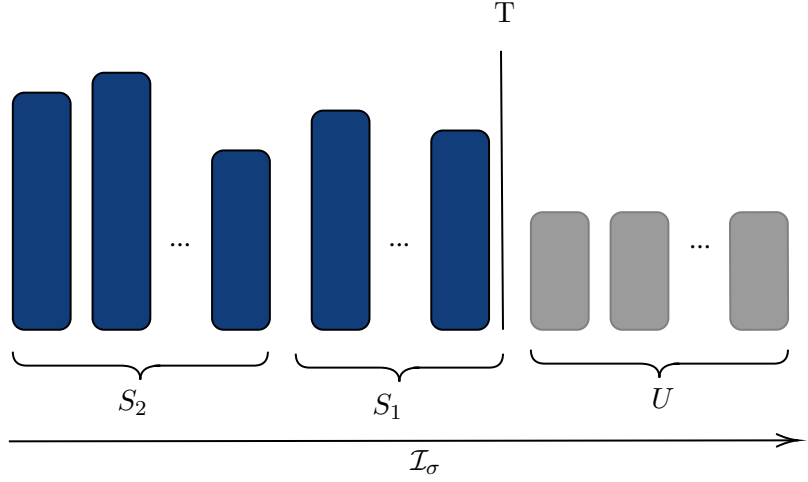

\medskip
\noindent\textbf{Case 2: Small $U$.}
We now bound the contribution of the Small-$U$ case, using a split
based on deterministic parameters of the input. Write $\operatorname{OPT}_k(\instance)$ for the
optimal $k$-clustering cost. Let $\mathrm{Cost}_{\mathrm{merge}}$ denote
the cost of type~(iii), namely the number of non-edges between distinct
pivot clusters assigned to the same output cluster. In particular,
$0\le \mathrm{Cost}_{\mathrm{merge}}\le n^2$ for every arrival order.

As before, $\Tau$ is the timestep immediately preceding the arrival of
the $(k+1)$-st pivot. Set $\Tau=n$ if there are at most $k$ pivots.
When $\Tau=n$, no two pivot clusters are assigned to the same output
cluster, and $\mathrm{Cost}_{\mathrm{merge}}=0$. Thus it suffices to
consider $\Tau<n$. The case $n=1$ is immediate, so assume $n\ge2$.

For integer $i$ such that $0\le i\le\lfloor\log_2(n-1)\rfloor$, we recall the definitions of $L_i$ and $H_i$ and introduce $p_i$:
\[
 L_i=2^i,\qquad H_i=\min\{2^{i+1},n\},\qquad p_i=\frac {H_i}{n},
\]
The intervals
$[L_i,H_i)$ partition the possible values $1,\ldots,n-1$ of $\Tau$ and satisfy
$L_i<H_i\le2L_i$. Define
\begin{equation}
 \mathcal B_i=
 \{L_i\le \Tau<H_i\}\cap
 \left\{|U|\le c_2|S_1|\frac{n\ln n}{L_i}\right\},
 \label{eq:smallu-event}
\end{equation}
where $c_2$ is the same fixed constant used in the Large-$U$ case.
If $\Tau<n$, then $U\ne\varnothing$, since the vertex
arriving at $\Tau+1$ is uncovered at $\Tau$. Hence $S_1\ne\varnothing$ on
$\mathcal B_i$.

\medskip
\noindent\emph{The projected-load bound.}
We use the preceding projected-load estimates in their unconditional
high-probability form. More precisely, for every $i$ there is an event
$\mathcal G_i^{\mathrm{load}}$ with
$\Pr[(\mathcal G_i^{\mathrm{load}})^c]=O(n^{-3})$ on which, whenever
$\mathcal B_i$ occurs, the projected load of the least loaded cluster
immediately before any subsequent pivot is assigned is at most
\[
 \Lambda_i=C_{\mathrm{load}}\frac{\ln^2 n}{p_i},
\]
for an absolute constant $C_{\mathrm{load}}$.
Take ${\cal G}_i^{\mathrm{load}}$ to be ${\cal G}^{load}$ defined in the proof of Lemma~\ref{lem:proj_load_of_least_load}.  Its first condition, evaluated at $t=\Tau$, implies for each leading pivot $v$ of a cluster in $S_1$ that $p_n(v)$ is at most $4np_\Tau(v)/\Tau + 2cn\ln n/\Tau + 1$, which is at most $(4\beta + 2c)n\ln n/L_i + 1$.  Hence $\mu(S_1)=O(n\ln n/L_i)$ on the same event.  
Substitute the bound on $\mu(S_1)$ and the facts $\Tau\ge L_i$ and
$|U|/|S_1|\le c_2n\ln n/L_i$ into the claim of Lemma~\ref{lem:proj_load_of_least_load}, yielding the stated bound on $\Lambda_i = O(n\ln^2n/L_i)=O(\ln^2n/p_i)$. The constants in the preceding
concentration estimates are chosen sufficiently large to give the
stated unconditional failure probability.

Charge the interaction of two co-clustered pivot clusters when the
later of their two pivots is assigned. Its contribution is at most the
size of the later pivot cluster times the existing projected load of
its destination cluster. The pivot clusters founded after $\Tau$
partition $U$, so
\begin{equation}
 \mathrm{Cost}_{\mathrm{merge}}\le\Lambda_i|U|
 \qquad\text{on }\mathcal B_i\cap\mathcal G_i^{\mathrm{load}}.
 \label{eq:smallu-load}
\end{equation}
The rest of the argument for the small $U$ case is focused on bounding $|U|$ in terms of the other parameters of the process, for which we derive upper bounds that can be charged to terms that contribute to the lower bound on optimal cost.  Before presenting an upper bound on $|U|$, we analyze three ingredients of our argument: (i) missing edges that contribute to the cost of an optimal unconstrained correlation clustering (CC) solution, (ii) contribution to $|U|$ from groups whose first pivot has not arrived by $\Tau$, and (iii) a lower bound on optimal cost expressed in terms of the parameters of ${\cal I}_{\mathrm clique}$.

\medskip
\noindent\emph{Missing-edge certificates.}
Fix an optimal unconstrained CC partition
$Q_1,\ldots,Q_\ell$, ordered by non-increasing size, independently of
the arrival permutation $\sigma$. Let $e_m$ be the number of missing
edges within these groups and let $e_c$ be the number of edges between
different groups. Thus
\begin{equation}
 e_m+e_c=\operatorname{OPT}_{CC}(\instance)
 \le\operatorname{OPT}_k(\instance).
 \label{eq:smallu-optcc}
\end{equation}
Write $V_{\le t}=\{\sigma(1),\ldots,\sigma(t)\}$ for the arrival
prefix, and let $U_t$ be the vertices uncovered after timestep $t$;
thus $U=U_T$. For $v\in Q_j$, let
\[
 d^-(v)=|\{w\in Q_j\setminus\{v\}:\{v,w\}\notin E\}|,
 \qquad
 M_{H_i}=\sum_{v\in V_{\le H_i}}d^-(v).
\]
Since the missing degrees are deterministic and
$\sum_v d^-(v)=2e_m$,
\begin{equation}
 \mathbb E[M_{H_i}]=2p_ie_m.
 \label{eq:smallu-one-endpoint}
\end{equation}

Let $t_j$ be the arrival time of the first pivot from $Q_j$, with
$t_j=\infty$ if there is none. Let $J_{H_i}$ count pivots arriving by $H_i$
that are not the first pivot in their own group. If $v$ is such a
pivot and $w_j$ is the first pivot in its group, then
$\{w_j,v\}\notin E$, since otherwise $w_j$ would have covered $v$.
Both endpoints have arrived by $H_i$, and distinct non-first pivots
produce distinct missing edges. Therefore, with
\[
 \rho_i=\frac{H_i(H_i-1)}{n(n-1)}\le p_i^2,
\]
we have
\begin{equation}
 \mathbb E[J_{H_i}]\le\rho_i e_m\le p_i^2e_m.
 \label{eq:smallu-two-endpoints}
\end{equation}
All expectations in these identities and below are unconditional.

\medskip
\noindent\emph{Groups whose first pivot has not arrived.}
Set
\[
 \kappa_i=\frac{6n\ln n}{L_i}
 \le\frac{12\ln n}{p_i}.
\]
By Lemma~\ref{lem: uncovered} with constant $6$, followed by a union bound over at
most $n$ groups, there is an event $\mathcal G_i^{\mathrm{groups}}$
with failure probability at most $n^{-5}$ such that
\begin{equation}
 t_j>\Tau\ge L_i\quad\Longrightarrow\quad
 |U\cap Q_j|\le\kappa_i
 \qquad\text{for every }j.
 \label{eq:smallu-unrepresented-size}
\end{equation}
Indeed, $t_j>\Tau\ge L_i$ implies $t_j>L_i$, and $U_\Tau\subseteq U_{L_i}$.
If $\kappa_i\ge n$,
this bound is deterministic. Let
\[
 \mathcal G_i=
 \mathcal G_i^{\mathrm{load}}\cap\mathcal G_i^{\mathrm{groups}},
 \qquad \Pr[\mathcal G_i^c]=O(n^{-3}).
\]

\medskip
\noindent\emph{Lower bound on optimal cost using parameters of $\mathcal I_{\mathrm clique}$ when $\ell>k$.}
First suppose $\ell>k$. Recall that $\mathcal I_{\mathrm{clique}}$ is obtained
by turning all $Q_j$ into disjoint cliques. Define
\[
 s_0=s(\mathcal I_{\mathrm{clique}}),\quad
 a_0=a(\mathcal I_{\mathrm{clique}}),\quad
 r_0=r(\mathcal I_{\mathrm{clique}}),\quad
 d=k-r_0,
\]
and the two fixed sets
\[
 D_0=\bigcup_{j>k}Q_j,\qquad D=\bigcup_{j>r_0}Q_j.
\]
The definitions of the clique parameters give
\begin{equation*}
 1\le d\le k,\qquad D_0\subseteq D,\qquad
 |D_0|=s_0,\qquad |D|=da_0,\qquad s_0\le da_0.
\end{equation*}
In particular, none of these parameters depends on $\sigma$.
Define
\[
 K_0=s_0a_0+e_m.
\]
By Lemma~\ref{lem:opt-theta},
$s_0a_0\le4\operatorname{OPT}_k(\mathcal I_{\mathrm{clique}})$.
Forming $\mathcal I_{\mathrm{clique}}$ changes exactly $e_m+e_c$ vertex-pair
labels. Applying an optimal $k$-clustering of the given instance $\mathcal I$ to the modified
instance and using~\eqref{eq:smallu-optcc} gives
\[
 \operatorname{OPT}_k(\mathcal I_{\mathrm{clique}})
 \le\operatorname{OPT}_k(\mathcal I)+e_m+e_c
 \le2\operatorname{OPT}_k(\mathcal I).
\]
Consequently,
\begin{equation}
 K_0\le9\operatorname{OPT}_k(\mathcal I).
 \label{eq:smallu-static-charge}
\end{equation}

\medskip
\noindent\emph{A common decomposition of the uncovered vertices leading to an upper bound on $|U|$.}
On the event $\{\Tau<n\}$, let $Z_\Tau$ be the number among the first $k$ pivots that
are either in $D_0$, or are non-first pivots in a group $Q_j$ with
$j\le k$. These two classes are disjoint. Set $Z_\Tau=0$ when $\Tau=n$.
For $\Tau<n$, exactly $k-Z_\Tau$ of the first $k$ pivots are first pivots
in distinct top-$k$ groups. Hence exactly $Z_\Tau$ of those groups have
no pivot by time $\Tau$.

Recall that $P_j$ is the set left uncovered immediately after the
first pivot from $Q_j$. If $t_j\le \Tau$, permanence of coverage gives
$U\cap Q_j=U\cap P_j$. Thus the following is a disjoint partition:
\begin{equation}
 U={}(U\cap D_0)
 \mathbin{\dot\cup}
 \bigcup_{\substack{j\le k\\t_j\le \Tau}}(U\cap P_j) 
 \mathbin{\dot\cup}
 \bigcup_{\substack{j\le k\\t_j>\Tau}}(U\cap Q_j).
 \label{eq:smallu-partition}
\end{equation}
The first piece has size at most $s_0$. If $w_j$ is the first pivot
from $Q_j$ and $t_j\le \Tau$, every vertex of $U\cap P_j$ is a
non-neighbor of $w_j$. Hence $|U\cap P_j|\le d^-(w_j)$.
On $\Tau<H_i$, these first pivots are distinct vertices of $V_{\le H_i}$,
so the second piece has size at most $M_{H_i}$. The third piece consists
of $Z_\Tau$ unrepresented groups, each contributing at most $\kappa_i$
vertices on $\mathcal G_i^{\mathrm{groups}}$. Therefore
\begin{equation}
 |U|\le s_0+M_{H_i}+\kappa_i Z_\Tau
 \qquad\text{on }\mathcal B_i\cap\mathcal G_i.
 \label{eq:smallu-decomposition}
\end{equation}
In particular, future-created sets $P_j$ with $t_j>\Tau$ are not counted
as already-created residual pieces.

Define the fixed-set arrival counts
\[
 X=|D_0\cap V_{\le H_i}|,\qquad
 N=|D\cap V_{\le H_i}|,\qquad J=J_{H_i}.
\]
They satisfy $0\le X\le N$ and
\begin{equation}
 Z_\Tau\le X+J\qquad\text{on }\{T<H_i\}.
 \label{eq:smallu-blocker-count}
\end{equation}
Moreover, $\Tau<H_i$ implies $\Tau+1\le H_i$, so at least $k+1$ pivots have
arrived by $H_i$. At most $r_0$ of them can be first pivots from
$Q_1,\ldots,Q_{r_0}$. Every other first pivot is counted by $N$, and
all non-first pivots are counted by $J$. Thus
\begin{equation}
 N+J\ge d+1\qquad\text{on }\mathcal B_i.
 \label{eq:smallu-excess-pivots}
\end{equation}
Because $D_0\subseteq D$ are deterministic, counting ordered pairs
of distinct vertices gives the exact identities
\begin{align}
 \mathbb E[X]&=p_is_0, \label{eq:smallu-tail-first-moment}\\
 \mathbb E[X(N-1)]&=\rho_i s_0(|D|-1),
 \label{eq:smallu-mixed-pairs}\\
 \mathbb E[N(N-1)]&=\rho_i |D|(|D|-1).
 \label{eq:smallu-tail-pairs}
\end{align}
Note that for the above three identities, no independence between $X,N,J$, or $\mathcal B_i$ is asserted or
needed.  We now split according to whether $p_ia_0\ge1/4$ or $p_ia_0<1/4$.

\medskip
\noindent\emph{Sub-case 1: $p_i a_0\ge1/4$.}
By~\eqref{eq:smallu-blocker-count},
\eqref{eq:smallu-tail-first-moment}, and
\eqref{eq:smallu-two-endpoints},
\[
 \mathbb E[Z_\Tau\mathbf 1_{\mathcal B_i}]
 \le\mathbb E[X]+\mathbb E[J]
 \le p_is_0+p_i^2e_m.
\]
Using~\eqref{eq:smallu-load}, \eqref{eq:smallu-decomposition}, and
\eqref{eq:smallu-one-endpoint}, and then dropping good-event
indicators only from nonnegative upper bounds, we obtain
\begin{align}
 \mathbb E[\mathrm{Cost}_{\mathrm{merge}}
            \mathbf 1_{\mathcal B_i\cap\mathcal G_i}]
 &\le\Lambda_i\left(
       s_0\Pr[\mathcal B_i]+\mathbb E[M_{H_i}]
       +\kappa_i\mathbb E[Z_T\mathbf 1_{\mathcal B_i}]
       \right)\notag\\
 &\le\Lambda_i\left(s_0+2p_ie_m
                    +\kappa_i(p_i s_0+p_i^2e_m)\right)\notag\\
 &=O(\ln^3 n)\left(\frac{s_0}{p_i}+e_m\right)
  =O(\ln^3 n)K_0.
 \label{eq:smallu-large-pa}
\end{align}
The final step uses $1/p_i\le4a_0$.

\medskip
\noindent\emph{Sub-case 2: $p_i a_0<1/4$.}
We use two elementary counting inequalities. Whenever
$N+J\ge d+1$, either $N\le(d+1)/2$, in which case $J\ge N\ge X$,
or $N>(d+1)/2$, in which case integrality gives $N-1\ge d/2$.
Together with~\eqref{eq:smallu-blocker-count}, this proves
\begin{equation}
 Z_\Tau\mathbf 1_{\mathcal B_i}
 \le2J+\frac2d X(N-1).
 \label{eq:smallu-blocker-pair-inequality}
\end{equation}
The right-hand side is nonnegative also outside $\mathcal B_i$.
The same two cases give
\begin{equation}
 \mathbf 1_{\mathcal B_i}
 \le\frac{2J}{d}+\frac4{d^2}N(N-1).
 \label{eq:smallu-event-pair-inequality}
\end{equation}
Indeed, in the first case $2J/d\ge1$, and in the second case
$4N(N-1)/d^2\ge1$.

Taking expectations in~\eqref{eq:smallu-blocker-pair-inequality}
and using~\eqref{eq:smallu-two-endpoints},
\eqref{eq:smallu-mixed-pairs}, and $|D|=da_0$, we obtain
\begin{equation}
 \mathbb E[Z_\Tau\mathbf 1_{\mathcal B_i}]
 \le 2\rho_i e_m+\frac{2\rho_i}{d}s_0(|D|-1)
 \le 2p_i^2(e_m+s_0a_0)=2p_i^2K_0.
 \label{eq:smallu-blocker-second-moment}
\end{equation}
Similarly, \eqref{eq:smallu-event-pair-inequality} and
\eqref{eq:smallu-tail-pairs} yield
\[
 \Pr[\mathcal B_i]
 \le\frac{2p_i^2e_m}{d}+4p_i^2a_0^2.
\]
Since $s_0/d\le a_0$ and $p_ia_0<1/4$,
\begin{equation}
 s_0\Pr[\mathcal B_i]
 \le\frac{2s_0p_i^2e_m}{d}+4p_i^2s_0a_0^2
 \le2p_i(p_ia_0)e_m+4p_i(p_ia_0)s_0a_0
 \le\frac12p_ie_m+p_is_0a_0\le p_i K_0.
 \label{eq:smallu-tail-event-bound}
\end{equation}
Substituting~\eqref{eq:smallu-blocker-second-moment},
\eqref{eq:smallu-tail-event-bound}, and
\eqref{eq:smallu-one-endpoint} into
\eqref{eq:smallu-load} and~\eqref{eq:smallu-decomposition} gives
\begin{equation}
 \mathbb E[\mathrm{Cost}_{\mathrm{merge}}
            \mathbf 1_{\mathcal B_i\cap\mathcal G_i}]
 \le\Lambda_i\left(p_iK_0+2p_ie_m+2\kappa_i p_i^2K_0\right)
 =O(\ln^3 n)K_0.
 \label{eq:smallu-small-pa}
\end{equation}
This argument holds for every $d\ge1$.

\medskip
\noindent\emph{The case $\ell\le k$.}
In this case there is no deterministic tail, and we do not define
$a_0,r_0$, or $d$. Let $J_\Tau$ count the non-first pivots among the
first $k$ pivots on $\{\Tau<n\}$. Exactly $k-J_\Tau$ optimal groups are
represented at $\Tau$, so the number of unrepresented groups is
\[
 \ell-(k-J_\Tau)\le J_\Tau\le J_{H_i}\qquad\text{on }\{\Tau<H_i\}.
\]
The same residual-vertex charge and group-size bound used above
therefore give
\[
 |U|\le M_{H_i}+\kappa_iJ_{H_i}
 \qquad\text{on }\mathcal B_i\cap\mathcal G_i.
\]
By~\eqref{eq:smallu-one-endpoint} and
\eqref{eq:smallu-two-endpoints},
\begin{equation}
 \mathbb E[\mathrm{Cost}_{\mathrm{merge}}
            \mathbf 1_{\mathcal B_i\cap\mathcal G_i}]
 \le\Lambda_i(2p_ie_m+\kappa_i p_i^2e_m)
 =O(\ln^3 n)e_m
 =O(\ln^3 n)\operatorname{OPT}_k(\instance).
 \label{eq:smallu-no-tail}
\end{equation}

\medskip
\noindent\emph{Bad events and summation over time intervals.}
For every $i$, use the actual cost bound, rather than
$\Lambda_i|U|$, on the complement of the good event:
\[
 \mathbb E[\mathrm{Cost}_{\mathrm{merge}}
            \mathbf 1_{\mathcal B_i\cap\mathcal G_i^c}]
 \le n^2\Pr[\mathcal G_i^c]=O(1/n).
\]
Combining~\eqref{eq:smallu-static-charge},
\eqref{eq:smallu-large-pa}, \eqref{eq:smallu-small-pa}, and
\eqref{eq:smallu-no-tail}, we conclude for every $i$ that
\begin{equation*}
 \mathbb E[\mathrm{Cost}_{\mathrm{merge}}
            \mathbf 1_{\mathcal B_i}]
 \le O(\ln^3 n)\operatorname{OPT}_k(\mathcal I)+O(1/n).
\end{equation*}
The events $\mathcal B_i$ are disjoint and there are $O(\ln n)$ of
them. Therefore the total contribution of the Small-$U$ case is
\begin{equation*}
 \mathbb E\left[\mathrm{Cost}_{\mathrm{merge}}
        \mathbf 1_{\bigcup_i\mathcal B_i}\right]
 =\sum_i\mathbb E[\mathrm{Cost}_{\mathrm{merge}}
                    \mathbf 1_{\mathcal B_i}]
 \le O(\ln^4 n)\operatorname{OPT}_k(\mathcal I)
      +O\!\left(\frac{\ln n}{n}\right).
 \label{eq:smallu-total}
\end{equation*}
If $\operatorname{OPT}_k(\mathcal I)>0$, it is an integer at least $1$, so
the additive term is absorbed into the multiplicative bound. If
$\operatorname{OPT}_k(\mathcal I)=0$, then $\instance$ is a union of at most $k$
disjoint cliques; PIVOT creates at most $k$ pivots and the algorithm
has zero cost. Hence the Small-$U$ contribution is
$O(\ln^4 n)\operatorname{OPT}_k(\mathcal I)$ in all cases. Combining this with
the Large-$U$ estimate and the earlier bound for costs of types~(i)
and~(ii) completes the proof of Theorem~4.

\end{proof}

\appendix

\section{NP-Hardness of Offline \kCC\ on Disjoint Cliques}
\label{app:np-hardness}

Let $\instance$ be an instance consisting of $\ell$ disjoint cliques
$K_1, \ldots, K_\ell$ with sizes $n_1, \ldots, n_\ell$ and let
$n = \sum_{i=1}^{\ell} n_i$.
For a $k$-clustering $\mathcal{C} = \{C_1, \ldots, C_k\}$ that keeps
every clique intact (which is without loss of generality by
Lemma~\ref{lem:cliques-intact}), write $c_j = |C_j|$ for the size of
cluster $j$.  The only disagreements arise from pairs of vertices
belonging to \emph{different} cliques that are placed in the
\emph{same} cluster.  Counting such pairs gives

\begin{equation}
  \mathrm{Cost}(\mathcal{C})
  \;=\;
  \sum_{j=1}^{k}
    \sum_{\substack{i < i' \\ K_i,\,K_{i'} \subseteq C_j}}
      n_i n_{i'}
  \;=\;
  \frac{1}{2}
  \Bigl(
    \sum_{j=1}^{k} c_j^2
    \,-\,
    \sum_{i=1}^{\ell} n_i^2
  \Bigr).
  \label{eq:cost-squares}
\end{equation}

Since $\sum_{i=1}^{\ell} n_i^2$ is a constant determined by the input,
minimizing $\mathrm{Cost}(\mathcal{C})$ is \emph{equivalent} to
minimizing $\sum_{j=1}^{k} c_j^2$, i.e., finding a $k$-partition of
$\{n_1,\ldots,n_\ell\}$ whose group sums are as balanced as possible.

\begin{theorem}
\label{thm:np-hard-cliques-strong}
\kCC\ when the input graph is a collection of disjoint cliques is strongly NP-hard.
\end{theorem}

We reduce from \textsc{3-Partition}, which is strongly NP-complete \cite{DBLP:books/fm/GareyJ79}:

\begin{quote}
\textbf{\textsc{3-Partition}.}
  Given $3m$ positive integers $a_1, \ldots, a_{3m}$ satisfying
  $B/4 < a_i < B/2$ for every $i$, where $B = \tfrac{1}{m}\sum_{i=1}^{3m} a_i$
  is an integer, does there exist a partition of $[3m]$ into $m$
  triples $T_1, \ldots, T_m$ such that $\sum_{i \in T_j} a_i = B$
  for every $j \in [m]$?
\end{quote}

Note that the size constraint $B/4 < a_i < B/2$ implies that every
feasible triple contains \emph{exactly} three elements: any two
elements sum to less than $B$, and any four elements sum to more
than $B$.

\begin{proof}
    Let $a_1, \ldots, a_{3m}$ with $B/4 < a_i < B/2$ and
    $\sum_{i=1}^{3m} a_i = mB$ be an arbitrary \textsc{3-Partition}
    instance in which the numerical parameters  $a_i$'s are polynomially bounded in $m$.
    Form an instance $\instance$ of \kCC\ as follows:
    \begin{itemize}
        \item Create $\ell = 3m$ disjoint cliques $K_1, \dots , K_{3m}$ with $|K_i| = a_i$ for each $i \in [3m]$.
        \item Set the number of clusters $k = m$.
        \item Set cost threshold $T = \frac{1}{2} \big( mB^2 - \sum_{i = 1}^{3m} a_i^2 \big)$.
    \end{itemize}

    The total vertex count is $n = \sum_{i=1}^{3m} a_i = mB$, so the
    average cluster size is $B$ and $n$ is polynomial
    in the input length. The construction runs in polynomial time
    in the input size; since \textsc{3-Partition} is strongly NP-complete,
    the above restriction in which all $a_i$'s are polynomially bounded in $3m$ is also NP-hard.

    \noindent\textbf{Correctness ($\Rightarrow$).}

    Suppose the \textsc{3-Partition} instance is a YES-instance; let
    $T_1, \ldots, T_m$ be a valid partition into triples each summing to
    $B$.  Assign the cliques $\{K_i : i \in T_j\}$ to cluster $C_j$ for
    each $j \in [m]$.  Every cluster has size exactly $B$, so by
    \eqref{eq:cost-squares},
    \[
      \mathrm{Cost}(\mathcal{C})
      \;=\;
      \frac{1}{2}\Bigl(mB^2 - \sum_{i=1}^{3m} a_i^2\Bigr)
      \;=\; T.
    \]

    \noindent\textbf{Correctness ($\Leftarrow$).}
    Suppose $\mathrm{Cost}(\mathcal{C}) \le T$ for some $k$-clustering
    $\mathcal{C}$.  By Lemma~\ref{lem:cliques-intact} we may assume
    $\mathcal{C}$ keeps every clique intact, so the cluster sizes
    $c_1, \ldots, c_m$ are non-negative integers summing to $n = mB$.
    By \eqref{eq:cost-squares} and the assumption $\mathrm{Cost}(\mathcal{C}) \le T$,
    \[
      \sum_{j=1}^{m} c_j^2
      \;\le\;
      2T + \sum_{i=1}^{3m} a_i^2
      \;=\; mB^2.
    \]
    By the Cauchy--Schwarz inequality,
    \[
      \sum_{j=1}^{m} c_j^2
      \;\ge\;
      \frac{1}{m}\Bigl(\sum_{j=1}^{m} c_j\Bigr)^{\!2}
      \;=\;
      \frac{(mB)^2}{m}
      \;=\; mB^2,
    \]
    with equality if and only if $c_j = B$ for every $j \in [m]$.
    Therefore, $\mathrm{Cost}(\mathcal{C}) \le T$ forces every cluster to
    have size exactly $B$.
    
    It remains to show that each cluster contains exactly three cliques. Since $a_i > B/4$, any four elements sum to $> B$. Since $a_i < B/2$, any two elements sum to $< B$. Thus, a sum of exactly $B$ must involve exactly three elements.
    The indices of the cliques in each cluster $C_j$ form a triple
    $T_j$ with $\sum_{i \in T_j} a_i = c_j = B$, yielding a valid
    solution to the \textsc{3-Partition} instance.

    Since the reduction is polynomial, \kCC\
    on disjoint cliques is NP-hard.
    
\end{proof}

\section{PTAS for Offline \texorpdfstring{$k$}{k}-CC on Disjoint Cliques}
\label{app:PTAS}
Our PTAS closely follows the approach of~\cite{alon1998approximation} for a closely related scheduling and load balancing problem.
Let $\instance$ be an instance of $k$-CC consisting of disjoint cliques
$K_1, \ldots, K_\ell$ with sizes $n_1 \geq n_2 \geq \cdots \geq n_\ell$, and let $n = \sum_{i=1}^\ell n_i$.
By Lemma~\ref{lem:cliques-intact}, we may assume without loss of generality that every
optimal clustering keeps each clique intact. Thus, the problem reduces to partitioning
$\{n_1, \ldots, n_\ell\}$ into $k$ clusters to minimize $\sum_{j=1}^k c_j^2$, where $c_j$ is
the total size of cluster $j$. Let $L = n/k$ denote the average cluster size.

\begin{claim}[Large cliques occupy their own cluster]
\label{claim:large-cliques}
Without loss of generality, every clique has size strictly less than $L = n/k$.
\end{claim}

\begin{proof}
We first observe that by equation~\eqref{eq:cost-squares}, minimizing the
disagreement cost over $k$-clusterings of a disjoint cliques instance is
equivalent to minimizing $\sum_{j=1}^k c_j^2$, since $\sum_{i=1}^\ell n_i^2$
is a constant determined by the input.  Thus the problem is precisely an instance
of the scheduling problem $P|\cdot|\sum f(C_i)$ of \cite{alon1998approximation} with
$f(x) = x^2$.  We note that while the raw disagreement cost
$\frac{1}{2}(\sum_j c_j^2 - \sum_i n_i^2)$ does not itself satisfy condition
$(F^*)$ of \cite{alon1998approximation} due to the additive constant $-\frac{1}{2}\sum_i
n_i^2$, this constant plays no role in comparing solutions: any two clusterings
differ in cost only through the $\sum_j c_j^2$ term, so it suffices to apply
the observations of \cite{alon1998approximation} to $f(x) = x^2$ directly (which is
non-negative, convex, and satisfies $(F^*)$).

With this reduction in hand, Observation~2.1 of \cite{alon1998approximation} applies
directly. If a clique has size $q \geq L$, then in any optimal solution
it must occupy its own cluster.  Otherwise, if the clique shares a cluster of load $d$ with another clique of size $h > 0$, then $d > L$, so some other cluster has load $c < L$; moving the size-$h$ clique there changes the disagreement cost by $h(c - d + h) \le h(c - q) < 0$, contradicting optimality.   We
iteratively assign each such large clique to its own dedicated cluster, remove it
from the instance, and reduce $k$ by one.  After each removal, if no vertices remain, terminate the peeling process.  Otherwise, recompute $L = n/k$ for the remaining instance and continue until every remaining clique has size strictly less than $L$.
\end{proof}

Claim~\ref{claim:large-cliques} may remove some of the largest cliques and
correspondingly decrease $k$. From this point on, $n$, $k$, and
$L = n/k$ refer to the \emph{reduced} instance obtained after this
peeling process --- i.e.\ the sub-instance consisting of the cliques that
remain once every clique of size $\geq L$ (recomputed at each step) has
been removed and assigned its own cluster. This reduced instance satisfies
the hypothesis of Claim~\ref{claim:cluster-sizes}: every clique has size
strictly less than its $L$.

\begin{claim}
\label{claim:cluster-sizes}
There exists an optimal $k$-clustering in which every cluster has total size $c_j$
satisfying $L/2 < c_j < 2L$.
\end{claim}

\begin{proof}
We show that any optimal $k$-clustering $\mathcal{C}^*$, minimizing
$\sum_{j=1}^k (c_j^*)^2$ subject to $\sum_j c_j^* = n = kL$, must satisfy
$L/2 < c_j^* < 2L$ for all $j$.  In both cases we derive a contradiction by
exhibiting a clique transfer that strictly decreases $\sum_j (c_j^*)^2$.
Recall that transferring a clique of size $n_i$ from cluster $C_a$ to cluster $C_b$
changes $\sum_j (c_j^*)^2$ by
\begin{equation}
\label{eq:exchange}
    2n_i\bigl(c_b^* - c_a^* + n_i\bigr),
\end{equation}
which is strictly negative if and only if $c_a^* - c_b^* > n_i$.

Suppose $c_a^* \geq 2L$ for some cluster $C_a$.  Since $\sum_j c_j^* = kL$, there
exists a cluster $C_b$ with $c_b^* \leq L$.  Pick any clique $K_i \subseteq C_a$.
By Claim~\ref{claim:large-cliques}, $n_i < L$, so
\[
    c_a^* - c_b^* \geq 2L - L = L > n_i.
\]
By~\eqref{eq:exchange}, transferring $K_i$ from $C_a$ to $C_b$ strictly
decreases $\sum_j (c_j^*)^2$, contradicting optimality.

Suppose $c_b^* \leq L/2$ for some cluster $C_b$.  Since $\sum_j c_j^* = kL$,
there exists a cluster $C_a$ with $c_a^* > L$.  Since every clique has size
$< L$ (Claim~\ref{claim:large-cliques}) and $c_a^* > L$, cluster $C_a$ contains
at least two cliques; let $K_i$ be a clique in $C_a$ of minimum size, so that
$n_i \leq c_a^*/2$.  Then
\[
    c_a^* - c_b^* \geq c_a^* - L/2 \geq c_a^*/2 \geq n_i,
\]
where the second inequality uses $c_a^* \geq L$ and the third uses
$n_i \leq c_a^*/2$.  Hence~\eqref{eq:exchange} is $\leq 0$.  If the
inequality is strict (i.e.\ $c_a^* - c_b^* > n_i$), we immediately have a
contradiction.  Otherwise $c_a^* - c_b^* = n_i$, meaning $c_a^* = c_b^* + n_i
\leq L/2 + n_i \leq L/2 + c_a^*/2$, giving $c_a^* \leq L$, which contradicts
$c_a^* > L$.  In either case we reach a contradiction, so no cluster can have size
$\leq L/2$.
\end{proof}

\begin{claim}
\label{claim:macro-cliques}
Let $\varepsilon \in (0,1/8]$.  Given instance $\instance$ in which all clique sizes are
less than $L$, define instance $\mathcal{J}$ as follows.
Cliques of size $\geq \varepsilon L$ are kept unchanged.  Cliques of size
$< \varepsilon L$ are greedily packed into \emph{macro-cliques}: scan the small
cliques in arbitrary order, accumulating them into the current macro-clique
until its total size first reaches or exceeds $\varepsilon L$, then start a new
macro-clique.  Each resulting macro-clique thus has total size in
$[\varepsilon L,\, 2\varepsilon L)$, with at most one remainder macro-clique of
size $< \varepsilon L$.  A $k$-clustering of $\mathcal{J}$ is a clique-preserving $k$-clustering of $\instance$ subject to the additional constraint that the constituent cliques
of each macro-clique lie in a common cluster; no edges of $\instance$ are
added or removed.  Then:
\begin{enumerate}
\item Every feasible $k$-clustering of $\mathcal{J}$ is a feasible $k$-clustering
of $\instance$ of the \emph{same} cost.
\item $\mathrm{OPT}(\mathcal{J}) \leq (1 + \bigo(\varepsilon))\,\mathrm{OPT}(\instance)$.
\end{enumerate}
\end{claim}

\begin{proof}
\textbf{Part (1).}
By construction, a feasible $k$-clustering of $\mathcal{J}$ is precisely a
$k$-clustering of $\instance$ satisfying the extra constraint that each
macro-clique's constituent cliques lie in a common cluster. Since no edges
of $\instance$ are added or removed in forming $\mathcal{J}$, the cost
function is unchanged, so the cost of such a clustering, viewed as a
clustering of $\mathcal{J}$ or of $\instance$, coincides.

\textbf{Part (2).}
For a macro-clique $M$ with constituent cliques of sizes $n_i$, $i \in M$,
let
\[
\Delta_M \;=\; \sum_{\substack{i<j\\ i,j\in M}} n_i n_j, \qquad
\Delta \;=\; \sum_M \Delta_M,
\]
the total number of vertex pairs lying in two different constituent cliques
of a common macro-clique. We show
\begin{equation}
\label{eq:opt-J-upper}
\mathrm{OPT}(\mathcal{J}) \;\le\; \mathrm{OPT}(\instance) + \Delta,
\end{equation}
and
\begin{equation}
\label{eq:delta-bound}
\Delta \;\le\; \frac{4\varepsilon}{1-2\varepsilon}\,\mathrm{OPT}(\instance).
\end{equation}
Together these give, since $\varepsilon \le 1/8$,
\[
\mathrm{OPT}(\mathcal{J}) \;\le\; \left(1+\frac{4\varepsilon}{1-2\varepsilon}\right)
\mathrm{OPT}(\instance) \;\le\; (1+8\varepsilon)\,\mathrm{OPT}(\instance),
\]
proving Part (2).

\smallskip
\noindent\emph{Proof of \eqref{eq:opt-J-upper}.}
We exhibit a $k$-clustering of $\mathcal{J}$, built from an optimal
clique-preserving $k$-clustering of $\instance$ (Lemma~\ref{lem:cliques-intact}),
whose cost exceeds $\mathrm{OPT}(\instance)$ by at most $\Delta$.

\begin{lemma}
\label{lem:one-move}
Consider a clique-preserving $k$-clustering of a disjoint-cliques instance,
and let $G_1, G_2$ be disjoint vertex sets, each a union of entire original
cliques, currently assigned to different clusters whose loads excluding
$G_1$ and $G_2$ respectively are $x$ and $y$. Moving $G_2$ into $G_1$'s
cluster (if $x \le y$; otherwise $G_1$ into $G_2$'s) changes the total cost
by at most $|G_1|\,|G_2|$, and the resulting clustering is again
clique-preserving.
\end{lemma}

\begin{proof}
Say $x \le y$ and move $G_2$ into $G_1$'s cluster; write $g_1=|G_1|$,
$g_2=|G_2|$. Since $G_2$ is a union of entire cliques and every other
vertex of its cluster lies in a different clique, all pairs between $G_2$
and its former surroundings are non-edges, and moving $G_2$ cuts no edges;
the same holds for $G_1$ and its surroundings. Removed: the $g_2 y$
disagreements between $G_2$ and its former surroundings. Added: $g_2 x$
disagreements between $G_2$ and $G_1$'s former surroundings, plus $g_1 g_2$
disagreements between $G_1$ and $G_2$ themselves (they have no edges
between them and are co-located for the first time). No other pair changes
status. The net change is
\[
(g_2 x + g_1 g_2) - g_2 y \;=\; g_1 g_2 + g_2(x-y) \;\le\; g_1 g_2,
\]
since $x \le y$. Since $G_2$ is moved as a whole and is a union of entire
cliques, no clique is split.
\end{proof}

Fix a macro-clique $M = \{K_1,\dots,K_r\}$ of sizes $n_1,\dots,n_r$. Since
distinct original cliques have no edges between them,
Lemma~\ref{lem:one-move} applies to any two of $K_1,\dots,K_r$ not yet
co-located. Apply it $r-1$ times: at step $j=2,\dots,r$, let $G_1$ be the
block $K_1 \cup \cdots \cup K_{j-1}$ already forced together by the
previous steps (size $n_1+\cdots+n_{j-1}$) and $G_2 = K_j$ (size $n_j$);
skip the step at zero cost if they are already co-located. By
Lemma~\ref{lem:one-move}, step $j$ costs at most
$(n_1+\cdots+n_{j-1})\,n_j$. Summing over $j=2,\dots,r$ and using the
identity $\sum_{j=2}^r \big(\sum_{i<j} n_i\big) n_j = \sum_{i<j} n_i n_j$,
the total cost of co-locating all of $M$ is at most $\Delta_M$.

Performing this for every macro-clique $M$ in turn (each step moves either
the block $K_1\cup\cdots\cup K_{j-1}$ or the clique $K_j$, both unions of
whole cliques of $M$, so the clustering stays clique-preserving and
co-location constraints enforced for previously processed macro-cliques
are never broken) yields a clustering satisfying $\mathcal{J}$'s
co-location constraint, whose cost exceeds that of the starting optimal
clustering of $\instance$ by at most $\sum_M \Delta_M = \Delta$. This
proves \eqref{eq:opt-J-upper}.

\smallskip
\noindent\emph{Proof of \eqref{eq:delta-bound}.}
Write $\mathcal{S}$ for the set of small cliques and $W = \sum_{i \in
\mathcal{S}} n_i$. If $W=0$ then $\Delta = 0$ and \eqref{eq:delta-bound}
is immediate, so assume $W>0$. Every macro-clique (including the possible
remainder) has size strictly less than $2\varepsilon L$, so
\[
\Delta \;\le\; \frac12 \sum_M w_M^2
\;\le\; \varepsilon L \sum_M w_M \;=\; \varepsilon L W,
\]
where $w_M$ denotes the total size of macro-clique $M$.

Fix an optimal clique-preserving $k$-clustering of $\instance$
(Lemma~\ref{lem:cliques-intact}), and for each small clique $i \in
\mathcal{S}$ let $c(i)$ denote the load of the cluster containing it. A
disagreement between two small cliques is counted once in each of their
two incidence terms $n_i(c(i)-n_i)$ (hence twice in $\sum_i
n_i(c(i)-n_i)$), while a disagreement between a small and a non-small
clique is counted once, so
\[
\mathrm{OPT}(\instance) \;\ge\; \frac12 \sum_{i \in \mathcal{S}} n_i\big(c(i)-n_i\big).
\]
By Claim~\ref{claim:cluster-sizes}, $c(i) > L/2$ for every cluster, and
$n_i < \varepsilon L$ for $i \in \mathcal{S}$, so
\[
\mathrm{OPT}(\instance) \;\ge\; \frac12 \sum_{i \in \mathcal{S}} n_i
\left(\frac{L}{2} - \varepsilon L\right) \;=\; \frac{1-2\varepsilon}{4}\,LW.
\]
Combining the two displayed bounds,
\[
\Delta \;\le\; \varepsilon L W \;\le\;
\frac{4\varepsilon}{1-2\varepsilon}\,\mathrm{OPT}(\instance),
\]
which is \eqref{eq:delta-bound}. \qedhere
\end{proof}

\begin{claim}
\label{claim:rounding}
Given instance $\mathcal{J}$ in which all objects (unchanged cliques and
macro-cliques) have sizes in $[\varepsilon L, L)$, except for at most one
remainder macro-clique of size $< \varepsilon L$
(Claim~\ref{claim:macro-cliques}), define instance $\mathcal{K}$ by rounding
each object size of the former kind \emph{up} to the nearest value of the
form $(1+\varepsilon)^i \cdot \varepsilon L / 2$ for a non-negative integer
$i$, and by leaving the remainder's size, if present, unrounded.
Then $\mathrm{OPT}(\mathcal{K}) \leq (1 + \bigo(\varepsilon))\,\mathrm{OPT}(\mathcal{J})$,
and every packing has cost in $\mathcal{J}$ at most its cost in $\mathcal{K}$.
Furthermore, an optimal solution to $\mathcal{K}$ can be found in polynomial time.
\end{claim}
\begin{proof}
 
    By Claim~\ref{claim:macro-cliques}, a $k$-clustering $\mathcal{C}$ of
    $\mathcal{J}$ is a packing of its objects into clusters. Every pair of
    vertices in distinct constituent cliques of a common macro-clique is a
    disagreement, and every other disagreement is a pair from two distinct
    objects in a common cluster, so its cost is
    \[
    F_{\mathcal{J}}(\mathcal{C}) \;=\; \Delta + \sum_{B\in\mathcal{C}}\ \sum_{\{M,N\}\subseteq B} w_M w_N,
    \]
    where $\Delta=\sum_M\Delta_M$ does not depend on $\mathcal{C}$. Let
    $\tilde w_M$ denote the rounded size of object $M$ ($\tilde w_M=w_M$ for
    the remainder), and define the objective of $\mathcal{K}$ as
    \[
    F_{\mathcal{K}}(\mathcal{C}) \;=\; \Delta + \sum_{B\in\mathcal{C}}\ \sum_{\{M,N\}\subseteq B} \tilde w_M \tilde w_N.
    \]
    Since every object other than the possible remainder has size at least
    $\varepsilon L$, rounding up inflates its size by a factor of at most
    $1+\varepsilon$, so $w_M\le\tilde w_M\le(1+\varepsilon)w_M$ for every $M$.
    Comparing termwise,
    \[
    F_{\mathcal{J}}(\mathcal{C}) \;\le\; F_{\mathcal{K}}(\mathcal{C}) \;\le\; (1+\varepsilon)^2F_{\mathcal{J}}(\mathcal{C}) \;\le\; (1+3\varepsilon)\,F_{\mathcal{J}}(\mathcal{C})
    \]
    for every packing $\mathcal{C}$ and $\varepsilon\le 1$. Applying the right
    inequality to an optimal packing of $\mathcal{J}$ gives
    $\mathrm{OPT}(\mathcal{K}) \le (1+3\varepsilon)\,\mathrm{OPT}(\mathcal{J})$.
    Consequently, an $F_{\mathcal{K}}$-optimal packing $\mathcal{C}_{\mathcal{K}}$
    satisfies $F_{\mathcal{J}}(\mathcal{C}_{\mathcal{K}}) \le
    F_{\mathcal{K}}(\mathcal{C}_{\mathcal{K}}) \le (1+3\varepsilon)\,\mathrm{OPT}(\mathcal{J})$.

    We now present a polynomial time algorithm for finding an optimal solution
    to $\mathcal{K}$. The number of distinct rounded object sizes is bounded by
    a constant $\alpha = O\!\left(\frac{1}{\varepsilon}\log\frac{1}{\varepsilon}\right)$;
    the (at most one) remainder is treated as one further, distinguished
    object. Let $L_{\mathcal{K}}$ denote the average rounded load and
    $w_{\max}$ the largest rounded object size. Every object of $\mathcal{J}$
    has size $<L$ and rounding inflates sizes by at most $1+\varepsilon$, so
    $L_{\mathcal{K}}\le(1+\varepsilon)L$ and $w_{\max}<(1+\varepsilon)L$. In any
    $F_{\mathcal{K}}$-optimal packing, every cluster has rounded load at most
    $L_{\mathcal{K}}+w_{\max}<2(1+\varepsilon)L$: otherwise some cluster $B$ has
    load $c_B>L_{\mathcal{K}}+w_{\max}$, some cluster $B'$ has load
    $c_{B'}\le L_{\mathcal{K}}$, and moving any object of rounded size $w$ from
    $B$ to $B'$ changes $F_{\mathcal{K}}$ by $w(c_{B'}-c_B+w)<0$. Since every
    non-remainder object has rounded size at least $\varepsilon L$, each
    cluster of an optimal packing contains at most
    $\beta=\lfloor 2(1+\varepsilon)/\varepsilon\rfloor\le\lfloor 4/\varepsilon\rfloor$
    such objects, plus possibly the remainder. Therefore, the number of
    different arrangements of object types in a cluster is at most a constant
    $\gamma = 2\binom{\alpha+\beta}{\alpha}$, the number of multisets of at
    most $\beta$ elements from $\alpha$ types (the factor of $2$ accounting
    for whether the remainder is placed in this cluster); we refer to each
    such arrangement as a cluster type.

    Note that a candidate solution must
respect the actual supply of each rounded object size in $\mathcal{K}$:
across all $k$ clusters, the total number of objects used of each type
must exactly equal the number of objects of that rounded size present in
$\mathcal{K}$ (and the single remainder, if present, must be
placed in exactly one cluster). Any solution in which every cluster has
one of the $\gamma$ cluster types can be specified by the
number of clusters of each type; thus, the number of candidate solutions
is at most $(k+1)^\gamma$, and by the bound above this set contains an
optimal solution. We find an optimal solution to $\mathcal{K}$ by
going over each of the $(k+1)^\gamma$ candidate solutions, discarding those
that do not use exactly $k$ clusters or do not exhaust the supply of every
object type exactly, and among
the remaining feasible candidates selecting the one with minimum value of
$F_{\mathcal{K}}$.
    The running time is bounded by $\bigo(n^2 (k+1)^\gamma)$, which
    is polynomial in $n$, assuming $\varepsilon$ is a constant.
\end{proof}

Given a disjoint cliques instance $\instance$ and a target accuracy
$\varepsilon>0$, we run the procedure below with
$\varepsilon'=\min\{\varepsilon/14,1/8\}$; since $\varepsilon'$ is a constant whenever $\varepsilon$ is, the running time of Claim~\ref{claim:rounding} remains polynomial. We first apply Claim~\ref{claim:large-cliques}, placing each peeled clique in its own cluster; this incurs no cost and does not change the optimum, so it suffices to solve the reduced instance (if it is empty, we are done). We then reduce the reduced instance to $\mathcal{J}$ by grouping small cliques into macro-cliques (Claim~\ref{claim:macro-cliques}), and reduce $\mathcal{J}$ to $\mathcal{K}$ by rounding object sizes to a geometric grid (Claim~\ref{claim:rounding}). We solve $\mathcal{K}$ exactly via the enumeration of Claim~\ref{claim:rounding}. By
Claim~\ref{claim:rounding}, the resulting packing has cost at most
$(1+3\varepsilon')\,\mathrm{OPT}(\mathcal{J})$ in $\mathcal{J}$; by
Claim~\ref{claim:macro-cliques}, it is a $k$-clustering of $\instance$ with
the same cost, which is at most
\[
(1+3\varepsilon')(1+8\varepsilon')\,\mathrm{OPT}(\mathcal{I}) \;=\; (1 + 11\varepsilon' + 24\varepsilon'^2)\,\mathrm{OPT}(\mathcal{I}) \;\le\; (1+14\varepsilon')\,\mathrm{OPT}(\mathcal{I}) \;\le\; (1+\varepsilon)\,\mathrm{OPT}(\mathcal{I}),
\]
where we used $24\varepsilon'^2 \le 3\varepsilon'$, which holds since $\varepsilon' \le 1/8$.

\section{Usage of AI}
\label{app:ai}
The problem formulation, the lower bounds, the algorithms, as well as their first analyses were all obtained without any use of AI.  
An initial version of this paper, entirely written by hand, derived upper bounds of $\bigo(\log^5 n)$ and $\bigo(\log^7 n)$ for the disjoint cliques and general graphs cases, respectively.  
Using Claude and ChatGPT Pro, we discovered some gaps in the original argument, especially in the small $|U|$ case of the proof of Theorem~\ref{thm:BalancedPivotUpper}.  
LLM reviews and suggestions from both Claude and ChatGPT Pro were used to resolve these technical gaps throughout, and in fact, improved the bounds to the ones we present here, $\bigo(\log^3 n)$ and $\bigo(\log^4 n)$, respectively.  
The STOC 2027 pre-submission AI feedback was also used to improve presentation, especially some technical arguments in \Cref{sec:lower-bound}.  
We have written down all the technical proofs ourselves. 
We wrote all of \Cref{sec:intro,sec:tech-overview} without any use of AI.  We take full responsibility for the correctness and originality of all content.  

For our formalized proofs, we initially used Aristotle to build out Lean code for the entire project. We did this by giving Aristotle our fully written paper and instructing it to formally verify every numbered Lemma and Theorem statement, assuming external results.  
We then complied and hand-checked all Hypotheses, Definitions, Lemmas, and Theorem statements to ensure correctness of the Lean code.  
After this, we used Claude as an additional correctness check to again ensure the paper was encoded correctly.  
While doing this, we improved two Lemma statements (Lemmas 3.4 and 6.2). 
We then used Claude to update the Lean encodings of these Lemmas, and once again hand-checked their encodings for correctness.

\bibliographystyle{alpha}
\bibliography{ref}

\end{document}